\documentclass[10pt,letterpaper]{amsart}
\pdfoutput=1	
\usepackage[utf8]{inputenc}

\usepackage{verbatim}
\usepackage{amsmath}	
\usepackage{amssymb}	
\usepackage{amsfonts}
\usepackage{amsthm}				

\usepackage{appendix}				
\usepackage{xcolor}	

\usepackage{soul}				
\sethlcolor{yellow}				

\usepackage{graphicx}			
\usepackage[labelformat=simple]{subcaption}	
\graphicspath{{images/}}

\usepackage[backend=bibtex,style=numeric-comp,citestyle=numeric,maxbibnames=99,giveninits,doi=false,isbn=false,url=true,eprint=false,sortcites]{biblatex}
\renewbibmacro{in:}{\ifentrytype{article}{}{\printtext{\bibstring{in}\intitlepunct}}}	
\allowdisplaybreaks

\renewcommand{\vec}[1]{\boldsymbol{#1}}
\let\oldhat\hat
\renewcommand{\hat}[1]{\oldhat{\boldsymbol{#1}}}

\newcommand{\de}{\ensuremath{\partial}}						
\newcommand{\dee}{\ensuremath{\textrm{d}}}

\newcommand{\pdf}[1]{\ensuremath{ \frac{\de}{\de #1}}}

\newcommand{\inty}[4]{\ensuremath{ \int_{#1}^{#2} \! #3 \, \dee#4 }}

\newcommand{\field}[1]{\mathbb{#1}}
\newcommand{\ip}[2]{\ensuremath{ \left< #1 \left| #2 \right. \right> } }

\DeclareMathOperator{\diag}{diag}

\newtheorem{assumption}{Assumption}
\newtheorem{remark}{Remark}
\newtheorem{theorem}{Theorem}

\newtheorem{lemma}{Lemma}
\newtheorem{proposition}{Proposition}
\newtheorem{corollary}{Corollary}
\newtheorem{example}{Example}
\numberwithin{lemma}{section}
\numberwithin{example}{section}
\numberwithin{proposition}{section}
\numberwithin{equation}{section}
\numberwithin{theorem}{section}
\numberwithin{corollary}{section}
\numberwithin{remark}{section}
\numberwithin{definition}{section}
\numberwithin{assumption}{section}
\numberwithin{figure}{section}
\newtheorem*{definition*}{Definition}		
\newtheorem*{assumption*}{Assumption}
\newtheorem*{remark*}{Remark}
\newtheorem*{theorem*}{Theorem}
\newtheorem*{lemma*}{Lemma}
\newtheorem*{proposition*}{Proposition}
\newtheorem*{corollary*}{Corollary}
\newtheorem*{example*}{Example}
\newtheorem*{conjecture*}{Conjecture}

\let\Re\relax
\DeclareMathOperator{\Re}{Re}

\title{Bistritzer-MacDonald dynamics in twisted bilayer graphene}

\author[A. B. Watson]{Alexander B. Watson}
\author[T. Kong]{Tianyu Kong}
\author[A. H. MacDonald]{Allan H. MacDonald}
\author[M. Luskin]{Mitchell Luskin}
\thanks{AW's research was supported in part by NSF DMREF Award No. 1922165, ML's research was supported in part by NSF Award DMS-1906129, and ML's, AM's and TK's research was supported in part by Simons Targeted Grant Award No. 896630. The authors thank Eric Canc\`es, Diyi Liu, and Max Engelstein for stimulating discussions.
}
\begin{document}
\begin{abstract}
    The Bistritzer-MacDonald (BM) model, introduced in \cite{Bistritzer2011}, attempts to capture the electronic properties of twisted bilayer graphene (TBG),
    even at incommensurate twist angles, by an effective periodic model over the bilayer moir\'e pattern.
    Starting from a tight-binding model, we identify a regime where the BM model emerges as the effective dynamics for electrons modeled as wave-packets spectrally concentrated at the monolayer Dirac points, up to error that can be rigorously estimated. Using measured values of relevant physical constants, we argue that this regime is realized in TBG at the first ``magic'' angle.
\end{abstract}

\maketitle

\section{Introduction}

\subsection{Overview}

Twisted bilayer graphene (TBG) consists of two graphene monolayers deposited on top of each other with a relative twist. Although, for generic twist angles, twisted bilayer graphene has no \emph{exact} periodic cell, the atomic structure of twisted bilayer graphene has an \emph{approximate} twist angle-dependent periodic structure known as the bilayer moir\'e pattern. Bistritzer and MacDonald \cite{Bistritzer2011,Bistritzer2010} proposed a model in 2011 for the electronic properties of TBG which is periodic over the bilayer moir\'e pattern, and used it to predict a series of ``magic'' angles where the Fermi velocity would vanish and electron-electron interaction effects might be enhanced. These predictions were dramatically validated by the experimental observations of Mott insulation and superconductivity at the first (largest) magic angle $\theta \approx 1^\circ$ in 2018 \cite{Cao2018a,Cao2018}. Since these discoveries, superconductivity and other strongly correlated electronic phases have been observed in many other twisted multilayer systems, now often referred to collectively as moir\'e materials. The ubiquity of moir\'e scale models in the field of moir\'e materials makes it important to understand their range of validity. Although compelling formal arguments for the validity of such models exist \cite{Bistritzer2011,Bistritzer2010,Catarina2019,CancesGarrigueGontier2022}, to our knowledge, no mathematically rigorous justification of any of these models exists in the literature. 

The present work identifies a parameter regime where the Bistritzer-MacDonald model \cite{Bistritzer2011} can be used to approximate the electronic properties of TBG up to error which can be \emph{rigorously estimated}. More specifically,
it considers the single-particle time-dependent Schr\"odinger equation, in the tight-binding approximation, for electrons in TBG, making a two-center approximation for the interlayer hopping function (see \eqref{eq:form_of_mathfrak_h} and Assumption \ref{as:h_regularity}). We non-dimensionalize the model, identifying its fundamental dimensionless parameters: $\ell > 0$, the ratio of interlayer distance to the monolayer lattice constant, and $\theta > 0$, the twist angle. We introduce sufficiently regular (Assumption \ref{as:f_regularity}) wave-packet initial data \eqref{eq:WP_0}, spectrally concentrated at the monolayer Dirac points, whose spread in momentum space is controlled by a new dimensionless parameter $\gamma > 0$\footnote{Equivalently, the scale of variation of the wave-packet envelope in real space is proportional to $\gamma^{-1}$.}. We consider the parameter regime where 
\begin{equation} \label{eq:balance__0}
    \oldhat{h}(|\vec{\kappa}|;\ell) \propto \gamma, \text{ and } \theta \lesssim \gamma, \text{ as } \gamma \rightarrow 0,
\end{equation}
where $\oldhat{h}(|\vec{\kappa}|;\ell)$ denotes the Fourier transform of the (non-dimensionalized) interlayer hopping function evaluated at the (non-dimensionalized) monolayer $\vec{K}$ point $\vec{\kappa} := a \vec{K}$, where $a$ is the monolayer lattice constant. We then prove that, in the regime \eqref{eq:balance__0}, the solution evolves as a wave-packet \eqref{eq:WP} over timescales $\propto \gamma^{-(1 + \nu)}$, for any $0 \leq \nu < c$, for some constant $c > 0$ depending on the interlayer hopping function (Assumption \ref{as:h_regularity}), up to error which $\rightarrow 0$ as $\gamma \rightarrow 0$, \emph{if and only if\footnote{For the ``only if'' part of this statement, see equation \eqref{eq:residual_decomp} of Lemma \ref{lem:residual_lemma}.} the wave-packet envelopes evolve according to a (scaled) Bistritzer-MacDonald model} (Lemma \ref{lem:residual_lemma} and Theorem \ref{thm:final_result}). That BM model is, specifically, 
\begin{equation} \label{eq:BM_model_0}
    i \de_\tau f = \begin{pmatrix} \vec{\sigma} \cdot (- i \nabla) & \oldhat{h}(|\vec{\kappa}|;\ell) \sum_{n = 1}^3 T_n e^{- i \mathfrak{s}_n \cdot \vec{r}} \\ \oldhat{h}(|\vec{\kappa}|;\ell) \sum_{n = 1}^3 T_n^\dagger e^{i \mathfrak{s}_n \cdot \vec{r}} & \vec{\sigma} \cdot (- i \nabla) \end{pmatrix} f,
\end{equation}
$\vec{\sigma} = ({\sigma}_1,{\sigma}_2)$ denotes the vector of Pauli matrices, $f = \left( f_1^A, f_1^B, f_2^A, f_2^B \right)^\top$ denotes the vector of envelope functions on layers $1, 2$ and monolayer sublattices $A, B$, and $T_n$ and $\mathfrak{s}_n := a \vec{s}_n$, $n \in \{1,2,3\}$, denote the hopping matrices \eqref{eq:hopping_matrices} and (non-dimensionalized) momentum shifts \eqref{eq:momentum_space_hops}. It is straightforward to check that the model \eqref{eq:BM_model_0} is invariant (modulo unitary transformations) under translations in the moir\'e lattice (Lemma \ref{lem:moire_periodicity}), and under arbitrary interlayer displacements (Lemma \ref{lem:interlayer_displacement}). All of the $\theta$-dependence of the model \eqref{eq:BM_model_0} enters through the length of the momentum shifts $\vec{s}_n$, which sets the moir\'e scale.

Restoring physical units\footnote{Note that we abuse notation between \eqref{eq:BM_model_0} and \eqref{eq:physical_BM} by using $\tau, \vec{r}$, and $f$ both for dimensionless quantities and their physical counterparts.} in \eqref{eq:BM_model_0}, recovers the model proposed by Bistritzer-MacDonald's \cite{Bistritzer2011}
\begin{equation} \label{eq:physical_BM}
    i \hbar \de_\tau f = \begin{pmatrix} v \vec{\sigma} \cdot (- i \nabla) & w \sum_{n = 1}^3 T_n e^{- i \vec{s}_n \cdot \vec{r}} \\ w \sum_{n = 1}^3 T_n^\dagger e^{i \vec{s}_n \cdot \vec{r}} & v \vec{\sigma}\cdot (- i \nabla) \end{pmatrix} f, \quad  v := \hbar v_D, \quad w := \frac{ \oldhat{\mathfrak{h}}(|\vec{K}|;L) }{ |\Gamma| },
\end{equation}
apart from a $\theta$-dependence of the monolayer Dirac cones \cite{Bistritzer2011}, which we prove can be neglected in the regime \eqref{eq:balance__0}\footnote{That this $\theta$-dependence is small was already noted in \cite{Bistritzer2011}, and neglecting it does not affect the prediction of the first magic angle.}. In \eqref{eq:physical_BM}, $\hbar$ denotes Planck's constant, $v_D$ is the Dirac velocity of monolayer graphene, and $\oldhat{\mathfrak{h}}(|\vec{K}|;L)$ denotes the Fourier transform of the interlayer hopping function evaluated at the monolayer $\vec{K}$ point and interlayer distance $L > 0$. Using physically realistic values for $v$ and $w$, Bistritzer and MacDonald numerically computed the band structure of the Hamiltonian \eqref{eq:physical_BM}, finding a sequence of ``magic'' twist angles where the Fermi velocity vanishes. They also gave an \emph{analytical} prediction of the first (largest) magic angle $\theta \approx 1^\circ$ by a formal perturbation theory, and we review this calculation (Section \ref{sec:magic_angles}). 

We emphasize that, although the analysis of this work is specialized to time-propagation of electronic wave-packets, we view our basic ansatz \eqref{eq:WP} as probing the single-particle electronic properties of TBG more generally. It is natural to ask, then, whether the regime \eqref{eq:balance__0} is really the relevant regime for understanding TBG's most interesting physical properties. In particular, one can ask: \emph{given the estimated physical value of $\oldhat{h}(|\vec{\kappa}|;\ell)$, and $\theta$ near to the first magic angle, does there exist a range of $\gamma$ values, with $\gamma \ll 1$, such that both conditions of \eqref{eq:balance__0} hold, with constants on the order of $1$?} We verify this as follows. Using the values given in \cite{Bistritzer2011} for the interlayer hopping energy and monolayer graphene $\pi$-band energy scale
\begin{equation} \label{eq:interlayer_energy}
    \frac{\oldhat{\mathfrak{h}}(|\vec{K}|;L)}{|\Gamma|} \approx 110 \text{ meV}, \quad \frac{\hbar v_D}{a} \approx 2.6 \text{ eV},
\end{equation}
we find that
\begin{equation} \label{eq:energy_ratio}
    \oldhat{h}(|\vec{\kappa}|;\ell) = \frac{\oldhat{\mathfrak{h}}(|\vec{K}|;L)}{|\Gamma| \left(\frac{ \hbar v_D }{ a }\right) } \approx \frac{110 \text{ meV}}{2.6 \text{ eV}} \approx 0.042.
\end{equation}
On the other hand, the first magic angle is at
\begin{equation}
    \theta \approx 1^\circ \approx 0.017 \text{ radians}.
\end{equation}
The statement now follows immediately: take any $\gamma$ such that $\oldhat{h}(|\vec{\kappa}|;\ell) = \lambda_0 \gamma$, with $\lambda_0$ on the order of $1$. Then, we have both $\gamma \ll 1$, and $\theta \leq \lambda_1 \gamma$, for all $\theta$ up to and including the magic angle, with $\lambda_1$ also on the order of $1$. For such $\gamma$, our analysis shows that \eqref{eq:physical_BM} will be a valid model of the \emph{single-particle} electronic properties of TBG over the energy scale $\gamma \frac{\hbar v_D}{a} \sim \frac{\mathfrak{h}(|\vec{K}|;L)}{|\Gamma|}$, and the length scale $\gamma^{-1} a$. 

We comment finally on the relevance of the BM model \eqref{eq:physical_BM} to modeling of strongly-correlated (many-body) electronic phases in TBG. At the first magic angle, the number of atoms per moir\'e cell can be estimated by squaring the ratio of the moir\'e lattice constant $\approx \frac{a}{\theta}$ to the monolayer lattice constant $a$ and multiplying by 4 (to capture the layer and sublattice degrees of freedom), giving
\begin{equation}
    \frac{1}{\theta^2} \approx 14000.
\end{equation}
The largeness of this number makes it imperative to develop many-body models of TBG at the moir\'e scale. We now argue, by a simple comparison of energy scales, that the BM model \eqref{eq:physical_BM} accurately approximates the 
single-particle part of many-body models (we do not attempt any rigorous analysis in this direction). Fixing $\theta$ again at the magic angle, the Coulomb e-e interaction in hexagonal Boron nitride (hBN, which typically encapsulates TBG in experiments) at the moir\'e scale defines an energy via
\begin{equation}
    \frac{ e^2 }{ 4 \pi \epsilon_{\text{hBN}} \left( \frac{a}{\theta} \right) } \approx 23 \text{ meV},
\end{equation}
where $\epsilon_{\text{hBN}}$ is the permittivity of hBN $\approx 5 \epsilon_0$ \cite{Laturia2018}, $\epsilon_0$ is the vaccuum permittivity, and $e$ is the elementary charge. The smallness of this number relative to the interlayer hopping energy \eqref{eq:interlayer_energy} suggests that e-e interaction terms couple only states that are described accurately by
the BM model \eqref{eq:physical_BM}. Note that this is not the same as saying that interactions are insignificant. 
Indeed they are known to significantly affect the dynamics in TBG's flat moir\'e bands, whose width can be as small as $\approx 10$ meV \cite{Carr2020}, leading to the wealth of strongly-correlated phases seen in experiments to date. Similar considerations apply to lattice-mediated interactions between electrons, which are characterized by a smaller energy scale $\sim 1$ meV \cite{wu2018theory,chou2022acoustic}.


\subsection{Related work} We comment on related mathematical work. Our results are closely related to those of Fefferman-Weinstein \cite{fefferman_weinstein}, who considered wave-packets spectrally localized at the Dirac points of a graphene \emph{monolayer}, starting from a PDE model rather than the tight-binding model we consider here. Under additional assumptions, we recover the same timescale of validity ($\sim \gamma^{- 2 + \epsilon}$ for any $\epsilon > 0$) as they prove (see Remark \ref{rem:Feff_Wein}). 

During the preparation of this work, we were made aware of parallel work deriving Bistritzer-MacDonald-type models for twisted bilayer graphene by Canc\`es-Garrigue-Gontier \cite{CancesGarrigueGontier2022}. Our works developed independently, and indeed work in different settings and derive different results. While \cite{CancesGarrigueGontier2022} work in the setting of a continuum model, we work in the setting of a discrete tight-binding model. While \cite{CancesGarrigueGontier2022} derive a moir\'e scale model by a formal variational approximation, we identify a parameter scaling limit where the Bistritzer-MacDonald model emerges up to higher-order terms which can be rigorously estimated. Finally, in contrast to \cite{CancesGarrigueGontier2022}, the Bistritzer-MacDonald model we derive has only \emph{nearest-neighbor} momentum space hopping, consistent with \cite{Bistritzer2011} (see Figure \ref{fig:TBG_momentum_lattice} for a discussion of the meaning of ``nearest-neighbor'' here).

We finally comment on related mathematical work on continuum PDE models of moir\'e materials. Becker-Embree-Wittsten-Zworski have introduced a spectral formulation of magic angles in \cite{becker2020mathematics,Becker2020}, while Becker-Kim-Zhu \cite{Becker2022} have obtained expansions of the density of states of TBG in the presence of strong magnetic fields using semiclassical analysis. Becker-Wittsten \cite{Becker2022_1} have proved a semiclassical quantization rule for the one-dimensional moir\'e-scale model introduced by Timmel-Mele \cite{Timmel2020}. Becker-Ge-Wittsten have also studied the transport and spectral properties of this model in \cite{Becker2022_2}. Two of the authors of the present work gave a computer-assisted proof of existence of the first magic angle of the ``chiral'' approximation to the Bistritzer-MacDonald model in \cite{Watson2021}. Bal-Cazeaux-Massatt-Quinn \cite{Quinn2022} have studied topological protection of edge states propagating along domain wall interfaces in relaxed twisted bilayer graphene. Numerical methods for computing electronic properties of twisted multilayer materials which exploit similar approximations to those made in the Bistritzer-MacDonald model have been developed \cite{Carr2017,Cances2017a, Massatt2017,2018CarrMassattTorrisiCazeauxLuskinKaxiras,Massatt2020,doi:10.1137/17M1141035}. 



\subsection{Conclusions and Perspectives}

The main results of the present work are (1) the identification of the parameter regime \eqref{eq:balance__0}, and (2) the identification of sufficient assumptions, most importantly on the interlayer hopping function (Assumption \ref{as:h_regularity}), so that the BM model \eqref{eq:physical_BM} can be rigorously justified as an effective model of wave-packet propagation in TBG. As alluded to above, however, the BM model \eqref{eq:physical_BM} is by no means only a model of wave-packet propagation. Rigorously justifying the use of the model \eqref{eq:physical_BM}, for example, to calculate the density of states or conductivity of TBG in the parameter regime identified here and under similar assumptions will be the subject of future works. 

Another important potential development of the present work would justify the various proposed models of TBG which account for effects of mechanical relaxation (see Remark \ref{rem:relaxing}). Here, it would be especially interesting to identify whether there exists a parameter regime where TBG can be rigorously approximated by the ``chiral'' model \cite{Tarnopolsky2019}, a simpler and more symmetric variant of the BM model. A further question would then be how close this regime is to that realized in real TBG. We expect that these analyses could be generalized to essentially arbitrary moir\'e materials such as twisted multilayer transition metal dichalcogenides (TMDs) or even twisted heterostructures consisting of layers of distinct 2D materials. Here an important question is how the distinct dispersion relations of these 2D materials affects the analysis; TMDs, for example, exhibit quadratic dispersion, quite unlike the linear dispersion exhibited by graphene. Finally, it would be very interesting to generalize these analyses to fundamental PDE models of moir\'e materials, like the model of TBG studied in \cite{CancesGarrigueGontier2022}.

\subsection{Structure of paper}

The structure of this work is as follows. We will first review the tight-binding models of graphene and twisted bilayer graphene in Section \ref{sec:tight-binding_model}. We will then introduce the wave-packet {\it ansatz} we will use to generate approximate solutions of the tight-binding model Schr\"odinger equation, and explain the strategy of the proof of convergence of these solutions, in Section \ref{sec:multiple_scales_analysis}. The strategy reduces the proof to a series of technical lemmas and a single main lemma, which identifies the parameter regime where the BM model represents the dominant effective dynamics, whose proof we give in Section \ref{sec:proof_key_lemma}. In Section \ref{sec:magic_angles}, we review Bistritzer-MacDonald's formal derivation of magic angles. We postpone various detailed calculations and proofs to the Appendices.

\section{The tight-binding model of twisted bilayer graphene} \label{sec:tight-binding_model}

The main object of this section is to introduce the tight-binding model of twisted bilayer graphene, which we will take as fundamental in what follows. Before we can do this, we must first recall the tight-binding model of monolayer graphene.

\subsection{The tight-binding model of graphene}

Graphene is a sheet of carbon atoms arranged in a honeycomb lattice; that is, a triangular lattice with two atoms per unit cell. We introduce Bravais lattice vectors as
\begin{equation}
    \vec{a}_1 := \frac{a}{2} \left( 1, \sqrt{3} \right)^\top, \quad \vec{a}_2 := \frac{a}{2} \left( -1 , \sqrt{3} \right)^\top, \quad A := \begin{pmatrix} \vec{a}_1, \vec{a}_2 \end{pmatrix},
\end{equation}
where $a > 0$ is the lattice constant (see Figure~\ref{fig:graphene_lattice}). The graphene Bravais lattice $\Lambda$ and a fundamental cell $\Gamma$ can then be written succinctly as
\begin{equation}
    \Lambda := \left\{ \vec{R} = A \vec{m} : \vec{m} \in \mathbb{Z}^2 \right\}, \quad \Gamma := \left\{ A \alpha : \alpha \in \left[0,1\right)^2 \right\}.
\end{equation}
It is straightforward to check that the area of a fundamental cell $|\Gamma| = \frac{\sqrt{3}}{2} a^2$. Within the $\vec{R}$th fundamental cell of the lattice, there are two atoms, at positions $\vec{R} + \vec{\tau}^A$ and $\vec{R} + \vec{\tau}^B$, which we will take as 
\begin{equation}
    \vec{\tau}^A := \vec{0}, \quad \vec{\tau}^B := \left(0, d\right)^\top, \quad d := \frac{a}{\sqrt{3}}.
\end{equation}
\begin{figure}
\centering
    \begin{subfigure}[b]{.49\columnwidth}
        \centering
        \includegraphics[scale=.49]{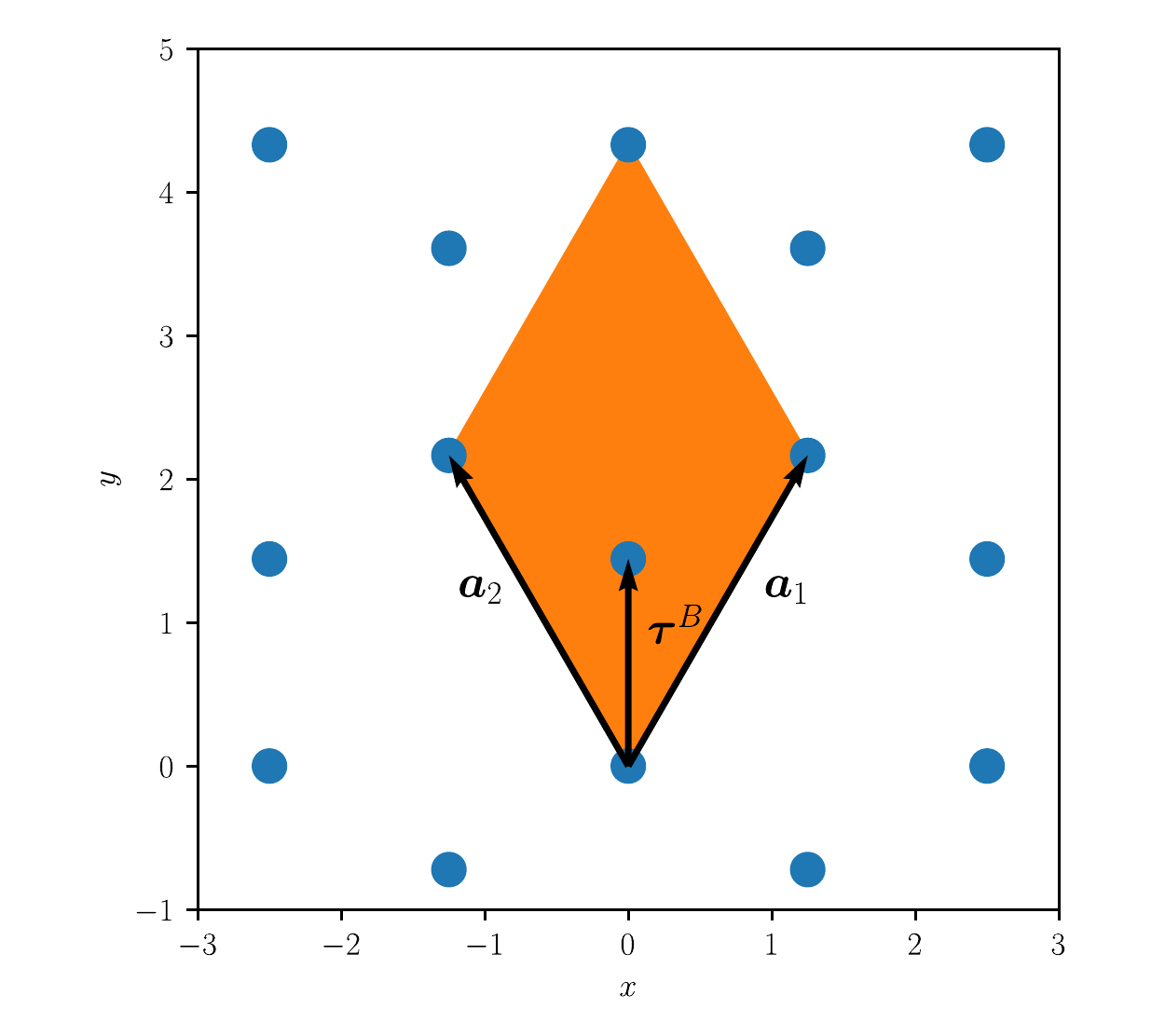}
        \caption{\label{fig:graphene_lattice}}
    \end{subfigure}
    \begin{subfigure}[b]{.49\columnwidth}
        \centering
        \includegraphics[scale=.49]{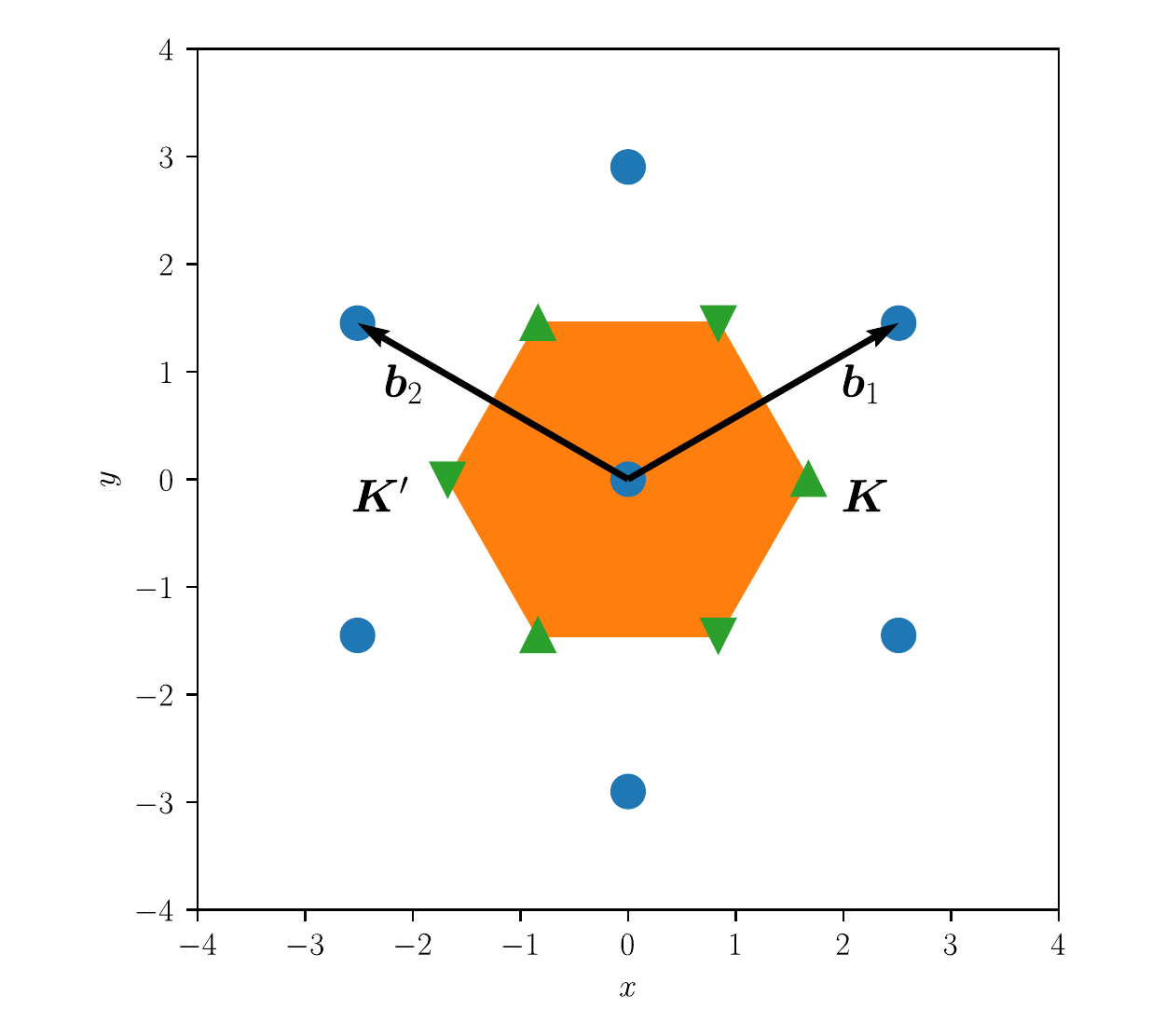}
        \caption{\label{fig:graphene_recip_lattice}}
    \end{subfigure}
    \begin{subfigure}[b]{.9\columnwidth}
    \centering
    \includegraphics[scale=.3]{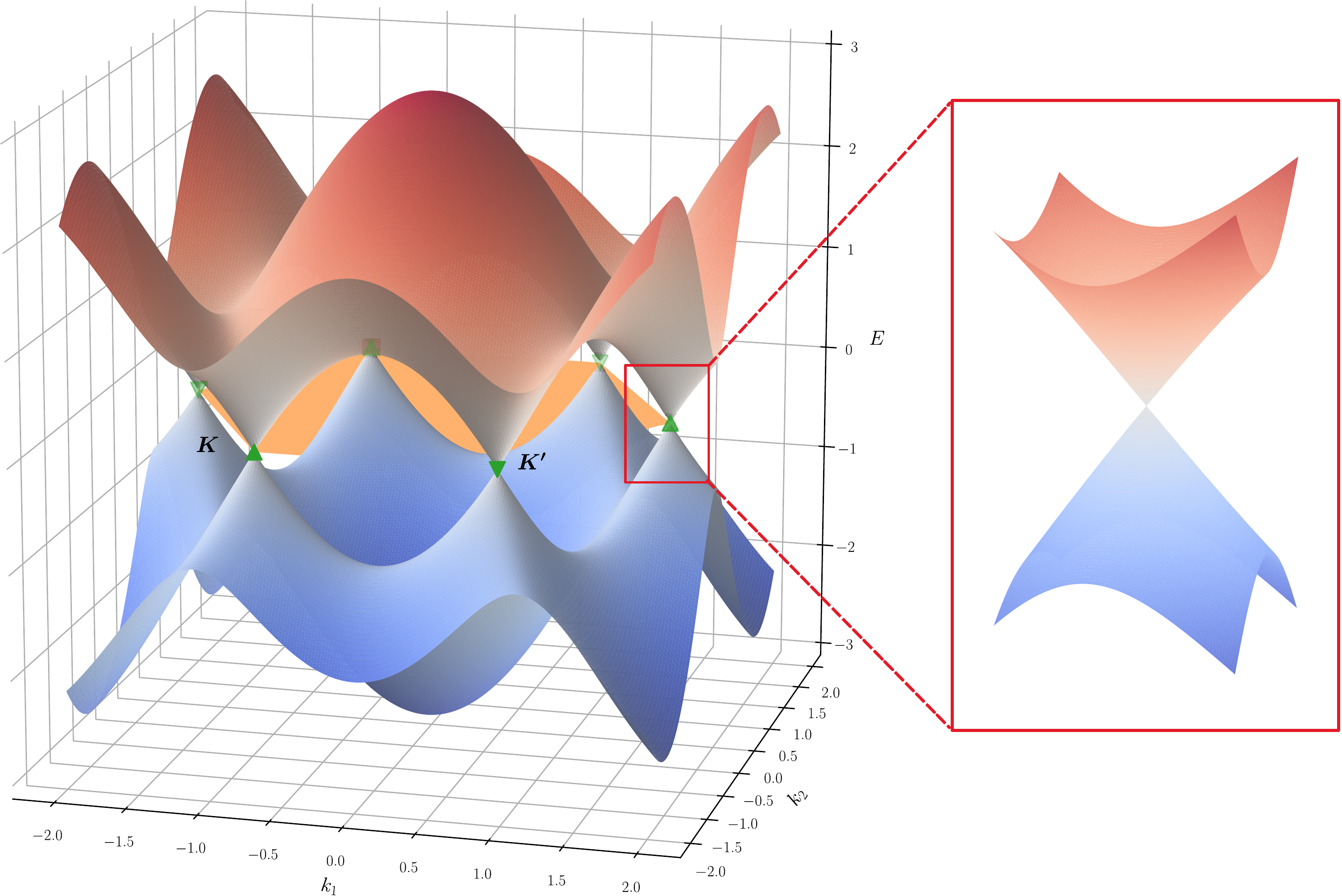}
    \caption{\label{fig:graphene_bands}}
    
        \end{subfigure}

    \caption{\subref{fig:graphene_lattice} Plot of the graphene monolayer atomic structure (blue dots), showing the lattice vectors $\vec{a}_1, \vec{a}_2$, $B$ site shift $\vec{\tau}^B$, and a unit cell $\Gamma$ (shaded in orange). Here, and elsewhere, we set the monolayer lattice constant $a = 2.5$, consistent with its experimental value in Angstroms. The $A$ site shift is trivial ($\vec{\tau}^A = \vec{0}$) because of our choice of origin. \subref{fig:graphene_recip_lattice} Plot of the graphene monolayer reciprocal lattice (blue dots), showing the reciprocal lattice vectors $\vec{b}_1, \vec{b}_2$, a fundamental cell (Brillouin zone) $\Gamma^*$ (shaded in orange), and the $\vec{K}$ and $\vec{K}'$ points (marked by green up and down arrows respectively). \subref{fig:graphene_bands} Plot of the monolayer graphene eigenvalue bands $E_\pm$ over the Brillouin zone $\Gamma^*$ shown in (B), showing their conical structure at the Dirac points $\vec{K}$ and $\vec{K}'$.}
    
\end{figure}

We model wave functions of electrons in graphene as elements $\psi$ of the Hilbert space $\mathcal{H}_{\text{mono}} := \ell^2(\mathbb{Z}^2;\mathbb{C}^2)$, writing $\psi = \left( \psi_{\vec{R}} \right)_{\vec{R} \in \Lambda} = \left( \psi^A_{\vec{R}}, \psi^B_{\vec{R}} \right)^\top_{\vec{R} \in \Lambda}$, where $|\psi^\sigma_{\vec{R}}|^2$ represents the electron density on sublattice $\sigma \in \{A,B\}$ in the $\vec{R}$th cell. The graphene tight-binding Hamiltonian with nearest-neighbor hopping acts as
\begin{equation} \label{eq:nearest_neighbor}
    \left( H \psi \right)_{\vec{R}} = - t \begin{pmatrix} \psi_{\vec{R}}^B + \psi_{\vec{R}-\vec{a}_1}^B + \psi_{\vec{R}-\vec{a}_2}^B \\ \psi_{\vec{R}}^A + \psi^A_{\vec{R} + \vec{a}_1} + \psi^A_{\vec{R} + \vec{a}_2} \end{pmatrix}, \quad \vec{R} \in \Lambda,
\end{equation}
where $t > 0$ is the nearest-neighbor hopping energy. 

The Hamiltonian is invariant under lattice translations so it is natural to pass to the Bloch domain. The reciprocal lattice vectors are defined by the equation $\vec{b}_i \cdot \vec{a}_j = 2 \pi \delta_{ij}$ for $1 \leq i,j \leq 2$. They are explicitly
(see Figure~\ref{fig:graphene_recip_lattice})
\begin{equation}
    \vec{b}_1 := \frac{4 \pi}{3 d} \left( \frac{\sqrt{3}}{2},\frac{1}{2} \right)^\top, \quad \vec{b}_2 := \frac{4 \pi}{3 d} \left( - \frac{ \sqrt{3} }{2},\frac{1}{2} \right)^\top, \quad B := \begin{pmatrix} \vec{b}_1, \vec{b}_2 \end{pmatrix}.
\end{equation}
We then define the reciprocal lattice $\Lambda^*$, and a fundamental cell of the reciprocal lattice $\Gamma^*$ (known in this context as the Brillouin zone), by
\begin{equation} \label{eq:BZ}
    \Lambda^* := \left\{ \vec{G} = B \vec{n} : \vec{n} \in \mathbb{Z}^2 \right\}, \quad \Gamma^* := \left\{ B \beta : \beta \in \left[0,1\right)^2 \right\}.
\end{equation}

We write wave-functions in the Bloch domain $L^2(\Gamma^*;\mathbb{C}^2)$ as $\tilde{\psi} = \left( \tilde{\psi}(\vec{k}) \right)_{\vec{k} \in \Gamma^*} = \left( \tilde{\psi}^A(\vec{k}), \tilde{\psi}^B(\vec{k}) \right)^\top_{\vec{k} \in \Gamma^*}$\footnote{Note that we use tildes here to reserve $\oldhat{f}$ for the Fourier transform of a continuous function $f$.}, and the unitary Bloch transform $\mathcal{H}_{\text{mono}} \rightarrow L^2(\Gamma^*;\mathbb{C}^2)$ and its inverse by 
\begin{equation}
    \begin{split}
        [ \mathcal{G} \psi ]^\sigma(\vec{k}) &:= \frac{1}{|\Gamma^*|^{\frac12}}\sum_{\vec{R} \in \Lambda} e^{- i \vec{k} \cdot (\vec{R} + \vec{\tau}^\sigma)} \psi^\sigma_{\vec{R}}, \\ 
        \left[ \mathcal{G}^{-1} \tilde{\psi}\right]^\sigma_{\vec{R}} &:= \frac{1}{|\Gamma^*|^{\frac12}} \inty{\Gamma^*}{}{ e^{i \vec{k} \cdot (\vec{R} + \vec{\tau}^\sigma)} \tilde{\psi}^\sigma(\vec{k}) }{\vec{k}}, \quad \sigma \in \{A,B\},
    \end{split}
\end{equation}
where $|\Gamma^*|$ denotes the Brillouin zone area. Note that with this convention, Bloch transforms are quasi-periodic with respect to $\Lambda^*$ in the sense that 
\begin{equation}
    [ \mathcal{G} \psi ]^\sigma(\vec{k} + \vec{G}) = e^{- i \vec{G} \cdot \vec{\tau}^\sigma} [ \mathcal{G} \psi ]^\sigma(\vec{k}), \quad \vec{G} \in \Lambda^*.
\end{equation}

The Hamiltonian is block diagonal in the Bloch domain, taking the form
\begin{equation}
    \left( \mathcal{G} H \mathcal{G}^{-1} \tilde{\psi} \right)(\vec{k}) = H(\vec{k}) \tilde{\psi}(\vec{k}),
\end{equation}
\begin{equation} \label{eq:graphene_block}
    H(\vec{k}) := - t \begin{pmatrix} 0 & F(\vec{k}) \\ \overline{F(\vec{k})} & 0 \end{pmatrix}, \quad F(\vec{k}) := e^{i \vec{k} \cdot ( \vec{\tau}^B - \vec{\tau}^A )} ( 1 + e^{- i \vec{k} \cdot \vec{a}_1} + e^{- i \vec{k} \cdot \vec{a}_2} ).
\end{equation}
The Bloch Hamiltonian $H(\vec{k})$ is easily diagonalized, with eigenpairs 
\begin{equation} \label{eq:monolayer_bands}
    E_\pm(\vec{k}) := \pm t |F(\vec{k})|, \quad \Phi_\pm(\vec{k}) := \frac{1}{\sqrt{2}} \left( 1 , \mp \frac{\overline{F(\vec{k})}}{ | F(\vec{k}) | } \right)^\top.
\end{equation}
The functions $E_\pm : \Gamma^* \rightarrow \mathbb{R}$ are known as the Bloch band functions. The Bloch band functions are degenerate at the Dirac points (see Figure~
\ref{fig:graphene_bands}), defined (up to $\vec{G} \in \Lambda^*$) by
\begin{equation} \label{eq:Dirac_pts}
    \vec{K} := \frac{ 4 \pi }{ 3 a } ( 1, 0 )^\top, \quad \vec{K}' := - \vec{K}.
\end{equation}

Taylor-expansion around Dirac point $\vec{K}$ in terms of $\vec{q} = \vec{k} - \vec{K}$ shows that\footnote{Here we are intentionally vague about the Taylor-expansion remainder for the sake of readability. We will make the error term precise when necessary for our proofs.}
\begin{equation} \label{eq:Taylor_f}
    F(\vec{K} + \vec{q}) = - \frac{\sqrt{3} a}{2} ( q_1 - i q_2 ) + O(|\vec{q}|^2),
\end{equation}
and hence
\begin{equation} \label{eq:monolayer_Dirac}
    H(\vec{K} + \vec{q}) = \hbar v_D \vec{\sigma} \cdot \vec{q} + O(|\vec{q}|^2), \quad v_D := \frac{\sqrt{3} a t}{2 \hbar}, \quad \vec{\sigma} = (\sigma_1,\sigma_2)^\top,
\end{equation}
where $\vec{\sigma} = (\sigma_1,\sigma_2)^\top$ is the vector of Pauli matrices, and $v_D$ is the Fermi velocity. Taylor-expansion around Dirac point $\vec{K}'$ yields
\begin{equation}
    F(\vec{K}' + \vec{q}) = - \frac{\sqrt{3} a}{2}( - q_1 - i q_2 ) + O(|\vec{q}|^2),
\end{equation}
leading to an identical expansion to \eqref{eq:monolayer_Dirac} other than the sign change in $q_1$. We finally record the transformation of $H(\vec{k})$ under translation of $\vec{k}$ in the reciprocal lattice 
\begin{equation} \label{eq:H_translation}
    H^{\sigma \sigma'}(\vec{k} + \vec{G}) = e^{i \vec{G} \cdot (\vec{\tau}^{\sigma'} - \vec{\tau}^\sigma)} H^{\sigma \sigma'}(\vec{k}), \quad \sigma, \sigma' \in \{A,B\}, \quad \vec{G} \in \Lambda^*. 
\end{equation}

\subsection{The tight-binding model of twisted bilayer graphene}

Twisted bilayer graphene (TBG) consists of two graphene monolayers with a relative twist. Let $R(\zeta)$ denote the matrix which rotates counterclockwise by angle $\zeta$
\begin{equation}
    R(\zeta) := \begin{pmatrix} \cos \zeta & - \sin \zeta \\ \sin \zeta & \cos \zeta \end{pmatrix}.
\end{equation}
Then, denoting the twist angle by $\theta > 0$, the lattice vectors of each layer are defined by
\begin{equation}
    \vec{a}_{1,i} := R\left(- \frac{\theta}{2}\right) \vec{a}_i, \quad \vec{a}_{2,i} := R\left(\frac{\theta}{2}\right) \vec{a}_i, \quad A_i := \begin{pmatrix} \vec{a}_{i,1}, \vec{a}_{i,2} \end{pmatrix}, \quad i \in \{1,2\},
\end{equation}
where the first subscript denotes the layer. We can then define lattices and fundamental cells just as in the monolayer
\begin{equation}
    \Lambda_i := \left\{ \vec{R}_i = A_i \vec{m} : \vec{m} \in \mathbb{Z}^2 \right\}, \quad \Gamma_i := \left\{ A_i \alpha : \alpha \in \left[0,1\right)^2 \right\}, \quad i \in \{1,2\}.
\end{equation}
The site shifts are similarly rotated, and we allow for a uniform interlayer displacement $\vec{\mathfrak{d}}$ so that
\begin{equation} \label{eq:interlayer_displacement}
    \vec{\tau}_1^\sigma := R\left(- \frac{\theta}{2}\right) \left( \vec{\tau}^\sigma - \frac{\vec{\mathfrak{d}}}{2} \right), \quad \vec{\tau}^\sigma_{2} := R\left(\frac{\theta}{2}\right) \left( \vec{\tau}^\sigma + \frac{\vec{\mathfrak{d}}}{2} \right), \quad \sigma \in \{A,B\}.
\end{equation}
\begin{figure}
    \centering
    \includegraphics[scale=.4]{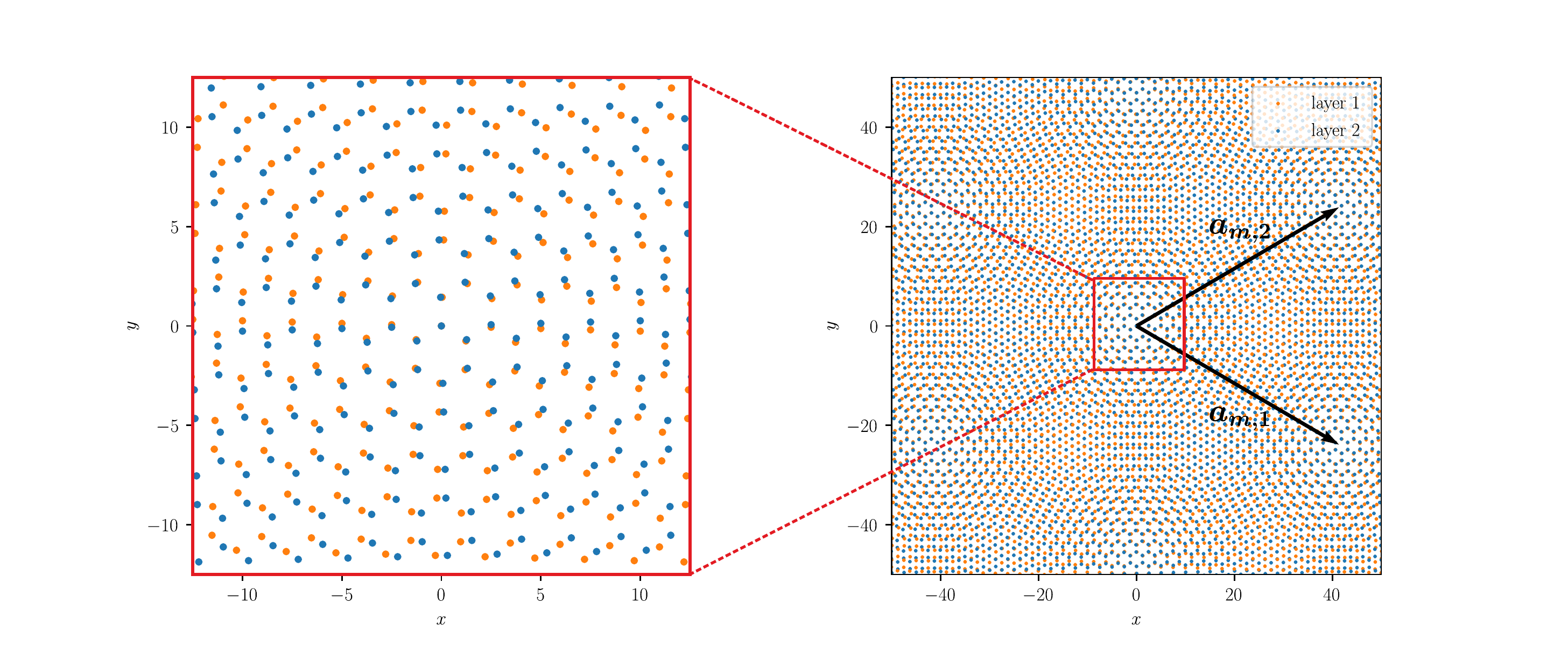} 
    \includegraphics[scale=.4]{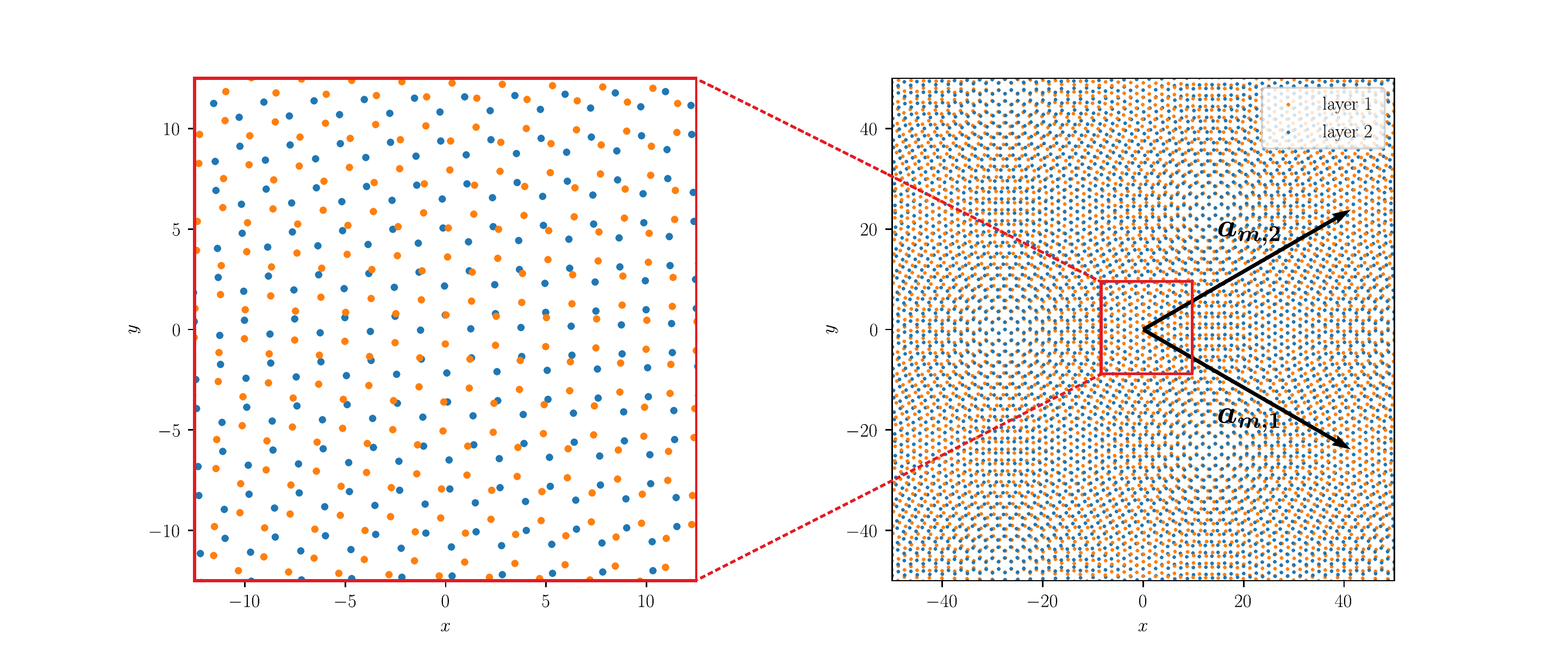}
    \caption{Plots of the atomic structures of twisted bilayer graphene, with interlayer displacement $\mathfrak{d} = 0$ (above), and $\mathfrak{d} = \vec{\tau}^B$ (below), at twist angle $\theta = 3^\circ$, near to the origin. Both structures generate the same moir\'e pattern, up to a shift, and BM models with identical spectral properties (Lemma \ref{lem:interlayer_displacement}).} 
    \label{fig:TBG_atomic_structure}
\end{figure}

There are two especially natural choices of $\vec{\mathfrak{d}}$ (see Figure~
\ref{fig:TBG_atomic_structure}). When $\vec{\mathfrak{d}} = \vec{0} \mod \Lambda$, the two untwisted graphene sheets are exactly aligned on top of each other. When $\vec{\mathfrak{d}} = \pm \vec{\tau}^B \mod \Lambda$, the untwisted graphene bilayer is said to be in ``Bernal stacking'' configuration. This is, mechanically speaking, the most energetically favorable configuration for the untwisted system \cite{2018CarrMassattTorrisiCazeauxLuskinKaxiras}. We will see that the spectral properties of the effective model for the twisted bilayer ($\theta > 0$) we derive are independent of $\vec{\mathfrak{d}}$; see Lemma \ref{lem:interlayer_displacement}. Note that this was already observed in \cite{Bistritzer2011}.

We model wave functions of electrons in TBG as elements $\psi$ of the Hilbert space $\mathcal{H} := \ell^2(\mathbb{Z}^2;\mathbb{C}^2) \oplus \ell^2(\mathbb{Z}^2;\mathbb{C}^2) = \left( \ell^2(\mathbb{Z}^2;\mathbb{C}^2) \right)^2$, writing $\psi = \left( \psi_1, \psi_2 \right)^\top$, where $\psi_i$ $= \left( \psi_{\vec{R}_i} \right)_{\vec{R}_i \in \Lambda_i}$ $= \left( \psi^A_{\vec{R}_i}, \psi^B_{\vec{R}_i} \right)^\top_{\vec{R}_i \in \Lambda_i}$, $i \in \{1,2\}$. It will sometimes be convenient to write instead $\psi = \left( \psi_{\vec{R}_1,\vec{R}_2} \right)_{\vec{R}_1,\vec{R}_2 \in \Lambda_1 \times \Lambda_2}$. The twisted bilayer graphene tight-binding Hamiltonian acts as
\begin{equation} \label{eq:block_H}
    H \psi = \begin{pmatrix} H_{11} & H_{12} \\ H_{12}^\dagger & H_{22} \end{pmatrix} \begin{pmatrix} \psi_1 \\ \psi_2 \end{pmatrix},
\end{equation}
where the diagonal (intralayer) part of the Hamiltonian is the monolayer graphene Hamiltonian
\begin{equation}
    \left( H_{ii} \psi_i \right)_{\vec{R}_i} := - t \begin{pmatrix} \psi_{\vec{R}_i}^B + \psi_{\vec{R}_i-\vec{a}_{i,1}}^B + \psi_{\vec{R}_i-\vec{a}_{i,2}}^B \\ \psi_{\vec{R}_i}^A + \psi^A_{\vec{R}_i + \vec{a}_{i,1}} + \psi^A_{\vec{R}_i + \vec{a}_{i,2}} \end{pmatrix}, \quad i \in \{1,2\},
\end{equation}
and the off-diagonal (interlayer) part is defined by (the $21$ term is similar so we omit it)
\begin{equation}
    \left( H_{12} \psi_2 \right)^\sigma_{\vec{R}_1} := \sum_{\vec{R}_2 \in \Lambda_2} \sum_{\sigma' \in \{A,B\}} \mathfrak{h}\left( \vec{R}_{12} + \vec{\tau}_{12}^{\sigma \sigma'} ; L \right) \psi^{\sigma'}_{\vec{R}_2}, \quad \sigma \in \{A,B\},
\end{equation}
where $\vec{R}_{12} := \vec{R}_1 - \vec{R}_2$ and $\vec{\tau}_{12}^{\sigma \sigma'} := \vec{\tau}_{1}^\sigma - \vec{\tau}_2^{\sigma'}$. Here, $L > 0$ denotes the interlayer distance, while $\mathfrak{h}$ is the interlayer hopping function, which arises from the overlap of ionic (or Wannier) orbitals on each layer \cite{ashcroft_mermin,Fefferman2018}. 
In this work, we will always assume that $\mathfrak{h}$ is radial, i.e., a function of $r := |\vec{r}|$ only, and make the standard abuse of notation to write $\mathfrak{h}(\vec{r};L) = \mathfrak{h}(r;L)$. More specifically, we will always assume that $h$ is a function of the \emph{three-dimensional} distance between sites so that
\begin{equation} \label{eq:form_of_mathfrak_h}
    \mathfrak{h}(r;L) := \mathfrak{h}^\sharp\left( \sqrt{ r^2 + L^2 } \right),
\end{equation}
for some decaying function $\mathfrak{h}^\sharp$ (see Figure~\ref{fig:interlayer_hopping_function}). Roughly speaking, the fact that the interlayer hopping function has the form \eqref{eq:form_of_mathfrak_h} allows us to assume that the two-dimensional Fourier transform of $\mathfrak{h}$ decays exponentially with a rate proportional to $L$. 
We postpone a discussion of the precise properties we require of $\mathfrak{h}$ (Assumption \ref{as:h_regularity}), and of the validity of these assumptions (see Proposition \ref{prop:Paley-Wiener}, and Examples \ref{ex:example_1}-\ref{ex:example_3}) until Section \ref{sec:nondim}, when we will nondimensionalize our model. There we will also discuss possible relaxations of our assumptions on the form of the interlayer hopping function (Remark \ref{rem:relaxing}). 
\begin{figure}
    \centering
    \begin{subfigure}[b]{.49\columnwidth}
        \centering
        \includegraphics[scale=.4]{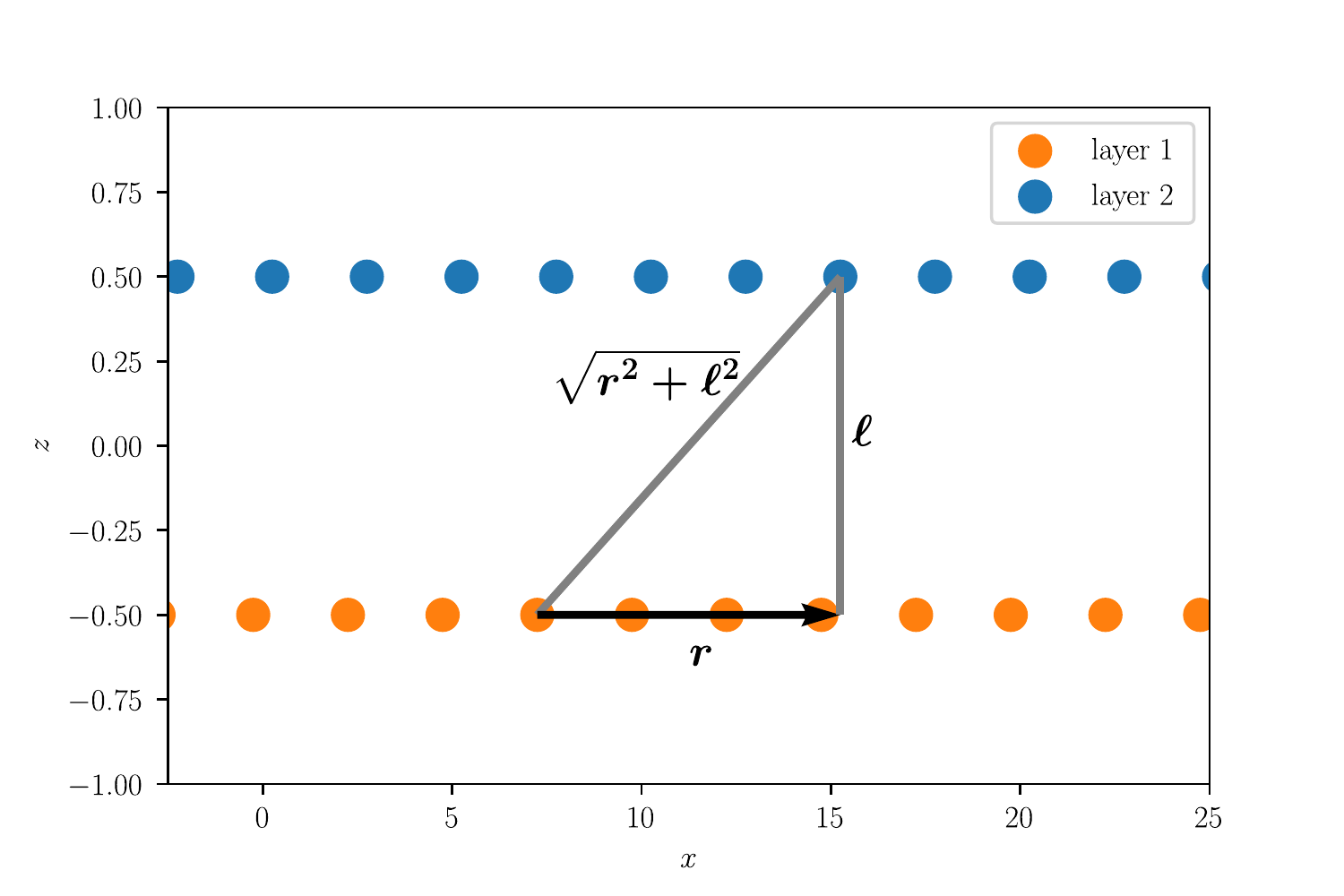}
        \caption{\label{fig:TBG_schematic}}
    \end{subfigure}
    \begin{subfigure}[b]{.49\columnwidth}
        \centering
        \includegraphics[scale=.35]{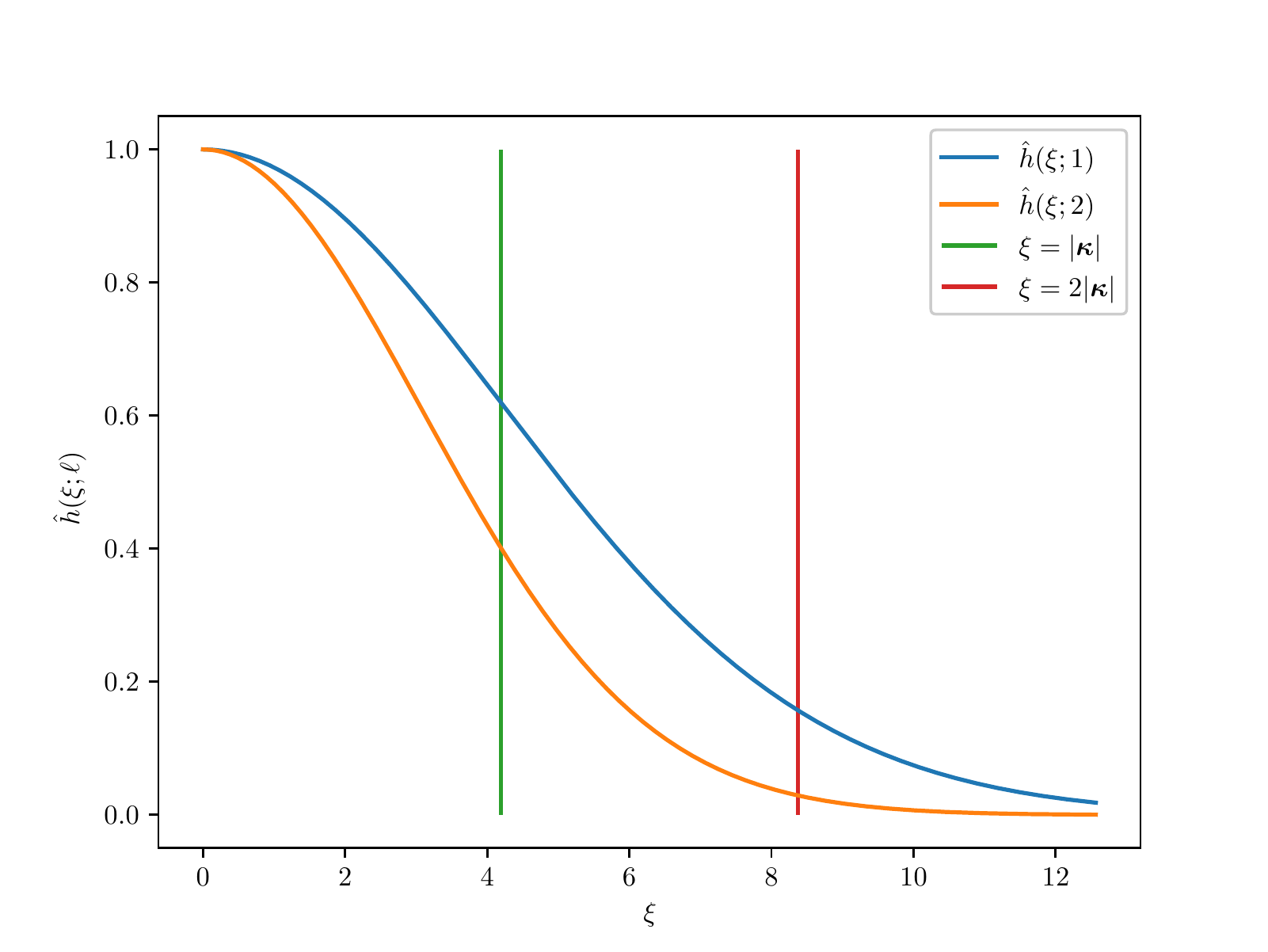}
        \caption{\label{fig:hhat_ell}}
    \end{subfigure}
    \caption{(A) Schematic illustrating how, in our tight-binding model, the $z$ direction perpendicular to the bilayer is suppressed, and we view the hopping as a function of the two-dimensional vector $\vec{r}$, but through on the three-dimensional distance $\sqrt{r^2 + \ell^2}$, where $\ell$ is the (non-dimensionalized) interlayer distance. This leads to the (two-dimensional) Fourier transform of the interlayer hopping function $h(\vec{r};\ell)$ decaying exponentially with a rate proportional to $\ell$. (B) Plot of the Fourier transform of the non-dimensionalized, radial, interlayer hopping function, $\oldhat{h}(\xi;\ell)$, for different values of the (non-dimensionalized) interlayer distance $\ell$. Specifically, we take $\oldhat{h}$ as in Example \ref{ex:example_2}, with $\alpha = 20$. The decay is much more rapid for larger values of $\ell$. Also shown are the key values of $\xi$ at $|\vec{\kappa}|$ and $2 |\vec{\kappa}|$. Our results rely crucially on the assumption that ``nearest-neighbor'' momentum space hopping dominates all other momentum space hopping, which amounts to assuming that $\frac{\oldhat{h}(2 |\vec{\kappa}|;\ell)}{\oldhat{h}(|\vec{\kappa}|;\ell)} \rightarrow 0$ as $\ell \rightarrow \infty$; see Assumption \ref{as:h_regularity} and the discussion below.}
    \label{fig:interlayer_hopping_function}
\end{figure}

Even though the interlayer hopping breaks translation symmetry\footnote{Here we refer to the translation symmetry of the individual layers, which is generally broken by the interlayer hopping. For specific ``rational'' twist angles, $H$ will retain exact translation symmetry with respect to ``supercell'' lattice vectors which are distinct from the monolayer lattice vectors. Note that we do not assume rationality of the twist angle anywhere in the present work. BM models are periodic with respect to the moir\'e lattice, which is well-defined for \emph{generic} twist angles.} of $H$, it is still useful to transform both layers to the Bloch domain. The reciprocal lattice vectors of each layer are defined by
\begin{equation}
    \vec{b}_{1,i} := R\left(- \frac{\theta}{2}\right) \vec{b}_i, \quad \vec{b}_{2,i} := R\left(\frac{\theta}{2}\right) \vec{b}_i, \quad B_i := \begin{pmatrix} \vec{b}_{1,i}, \vec{b}_{2,i} \end{pmatrix} \quad i \in \{1,2\},
\end{equation}
and the reciprocal lattices and Brillouin zones similarly
\begin{equation}
    \Lambda^*_i := \{ \vec{G}_i = B_i \vec{n} : \vec{n} \in \mathbb{Z}^2 \}, \quad \Gamma_i^* := \left\{ B_i \beta : \beta \in \left[0,1\right)^2 \right\}, \quad i \in \{1,2\}.
\end{equation}
The $\vec{K}$ and $\vec{K}'$ points of each layer are simply
\begin{equation}
    \vec{K}_1 := R\left(-\frac{\theta}{2}\right) \vec{K}, \quad \vec{K}_2 := R\left(\frac{\theta}{2}\right) \vec{K}, \quad \vec{K}_i' := - \vec{K}_i, \quad i \in \{1,2\}.
\end{equation}

We denote wave functions in the Bloch domain $L^2(\Gamma_1^*;\mathbb{C}^2) \oplus L^2(\Gamma_2^*;\mathbb{C}^2)$ by $\tilde{\psi} = \left( \tilde{\psi}_1, \tilde{\psi}_2 \right)^\top$, where $\tilde{ \psi }_i$ $= \left( \tilde{\psi}_i(\vec{k}_i) \right)_{\vec{k}_i \in \Gamma^*_i}$ $= \left( \tilde{\psi}_i^A(\vec{k}_i), \tilde{\psi}_i^B(\vec{k}_i) \right)^\top_{\vec{k}_i \in \Gamma^*_i}$, $i \in \{1,2\}$. It will sometimes be convenient to write instead $\tilde{\psi}$ $= \left( \tilde{\psi}(\vec{k}_1,\vec{k}_2) \right)_{(\vec{k}_1,\vec{k}_2) \in \Gamma^*_1 \times \Gamma^*_2}$. We then define a unitary Bloch transform $\mathcal{G} \rightarrow L^2(\Gamma_1^*;\mathbb{C}^2) \oplus L^2(\Gamma_2^*;\mathbb{C}^2)$, and its inverse, componentwise, by 
\begin{equation}
    \left[ \mathcal{G} \psi \right](\vec{k}_1,\vec{k}_2) := \begin{pmatrix} \left[ \mathcal{G}_1 \psi_1 \right](\vec{k}_1) \\ \left[ \mathcal{G}_2 \psi_2 \right](\vec{k}_2) \end{pmatrix}, \quad \left[ \mathcal{G}^{-1} \tilde{\psi} \right]_{\vec{R}_1,\vec{R}_2} := \begin{pmatrix} \left[ \mathcal{G}_1^{-1} \tilde{\psi}_1 \right]_{\vec{R}_1} \\ \left[ \mathcal{G}_2^{-1} \tilde{\psi}_2 \right]_{\vec{R}_2} \end{pmatrix},
\end{equation}
where
\begin{equation}
    \begin{split}
        [ \mathcal{G}_i \psi_i ]^\sigma(\vec{k}_i) &:= \frac{1}{|\Gamma^*|^{\frac12}}\sum_{\vec{R}_i \in \Lambda_i} e^{- i \vec{k}_i \cdot (\vec{R}_i + \vec{\tau}_i^\sigma)} \psi^\sigma_{\vec{R}_i}, \\
        \left[ \mathcal{G}_i^{-1} \tilde{\psi}_i\right]^\sigma(\vec{k}_i) &:= \frac{1}{|\Gamma^*|^{\frac12}} \inty{\Gamma^*_i}{}{ e^{i \vec{k}_i \cdot (\vec{R}_i + \vec{\tau}_i^\sigma)} \tilde{\psi}_i^\sigma(\vec{k}_i) }{\vec{k}_i}, \quad \sigma \in \{A,B\}, i \in \{1,2\}.
    \end{split}
\end{equation}

The Hamiltonian now acts by 
\begin{equation} \label{eq:Bloch_block_H}
    \left( \mathcal{G} H \mathcal{G}^{-1} \tilde{\psi} \right)(\vec{k}_1,\vec{k}_2) = \begin{pmatrix} \left( \mathcal{G}_1 H_{11} \mathcal{G}_1^{-1} \tilde{\psi}_1 \right)(\vec{k}_1) + \left( \mathcal{G}_1 H_{12} \mathcal{G}_2^{-1} \tilde{\psi}_2 \right)(\vec{k}_1) \\ \left( \mathcal{G}_2 H_{21} \mathcal{G}_1^{-1} \tilde{\psi}_1 \right)(\vec{k}_2) + \left( \mathcal{G}_2 H_{22} \mathcal{G}_2^{-1} \tilde{\psi}_2 \right)(\vec{k}_2) \end{pmatrix},
\end{equation}
where the diagonal terms are as before
\begin{equation} \label{eq:diag_Bloch_block_H}
    \left( \mathcal{G}_i H_{ii} \mathcal{G}_i^{-1} \tilde{\psi}_i \right)(\vec{k}_i) = H_i(\vec{k}_i) \tilde{\psi}_i(\vec{k}_i), \quad i \in \{1,2\},
\end{equation}
\begin{equation}
    H_{i}(\vec{k}_i) := - t \begin{pmatrix} 0 & F(\vec{k}_i) \\ \overline{F(\vec{k}_i)} & 0 \end{pmatrix}, \quad F(\vec{k}_i) := e^{i \vec{k}_i \cdot ( \vec{\tau}_i^B - \vec{\tau}_i^A )} ( 1 + e^{- i \vec{k}_i \cdot \vec{a}_{i,1}} + e^{- i \vec{k}_i \cdot \vec{a}_{i,2}} ).
\end{equation}

Just as in the monolayer, formal Taylor-expansion of $H_i(\vec{k}_i)$ around Dirac points $\vec{K}_i$ in terms of $\vec{q}_i = \vec{k}_i - \vec{K}_i$ yields effective Dirac operators. Specifically, we have
\begin{equation}
    F_i(\vec{K}_i + \vec{q}_i) = - \frac{\sqrt{3} a}{2} e^{i \theta_i/2} ( q_{i,1} - i q_{i,2} ) + O( |\vec{q}_i|^2 ), \quad \vec{q}_i = ( q_{i,1}, q_{i,2} )^\top,
\end{equation}
where
\begin{equation} \label{eq:theta_cases}
    \theta_i := \begin{cases} - \theta & i = 1, \\ \ \ \theta & i = 2, \end{cases}
\end{equation}
and hence
\begin{equation} \label{eq:monolayer_Dirac_twisted}
    \begin{split}
        &H_i(\vec{K}_i + \vec{q}_i) = \hbar v_D \vec{\sigma}_{\theta_i/2} \cdot \vec{q}_i + O(|\vec{q}_i|^2),   \\
        &\vec{\sigma}_{\theta_i/2} \cdot \vec{q}_i := \begin{pmatrix} 0 & e^{i \theta_i/2} ( q_{i,1} - i q_{i,2} ) \\ e^{- i \theta_i/2} ( q_{i,1} + i q_{i,2} ) & 0 \end{pmatrix}, \quad i \in \{1,2\},
    \end{split}
\end{equation}
where $v_D$ is the monolayer Fermi velocity \eqref{eq:monolayer_Dirac}.
The analogous calculation at the $\vec{K}'$ points of layer $i$ leads to the same results but with $q_{i,1}$ replaced everywhere by $- q_{i,1}$. 

A short calculation, presented in Appendix \ref{sec:off_diagonal_Bloch}, shows that the off-diagonal terms take the form
\begin{equation} \label{eq:off_diagonal_Bloch}
    \begin{split}
        &\left( \mathcal{G}_1 H_{12} \mathcal{G}_2^{-1} \tilde{\psi}_2 \right)^\sigma(\vec{k}_1) =  \\
        &\frac{1}{|\Gamma|} \sum_{\sigma' \in \{A,B\}} \inty{\Gamma^*_2}{}{ \sum_{\vec{G}_1 \in \Lambda_1^*} \sum_{\vec{G}_2 \in \Lambda_2^*} e^{i [\vec{G}_1 \cdot \vec{\tau}_1^\sigma - \vec{G}_2 \cdot \vec{\tau}_2^{\sigma'}]} \oldhat{\mathfrak{h}}( \vec{k}_1 + \vec{G}_1 ; L ) \delta( \vec{k}_1 + \vec{G}_1 - \vec{k}_2 - \vec{G}_2 ) \tilde{\psi}^{\sigma'}_2(\vec{k}_2) }{\vec{k}_2},
    \end{split}
\end{equation}
where we define the Fourier transform pair\footnote{Note that, with this definition, $\mathfrak{h}$ has units of energy, but $\oldhat{\mathfrak{h}}$ has units of energy times area. Since $|\Gamma|$ has units of area, the previous equation has units of energy as expected.}
\begin{equation}
    \oldhat{\mathfrak{h}}(\vec{\xi};L) := \inty{\field{R}^2}{}{ e^{- i \vec{\xi} \cdot \vec{r}} \mathfrak{h}(\vec{r};L) }{\vec{r}}, \quad \mathfrak{h}(\vec{r};L) = \frac{1}{(2 \pi)^2} \inty{\field{R}^2}{}{ e^{i \vec{\xi} \cdot \vec{r}} \oldhat{\mathfrak{h}}(\vec{\xi};L) }{\vec{k}}.
\end{equation}
Note that the off-diagonal part of $H$ is not block diagonal with respect to quasimomentum after Bloch transformation, reflecting the model's lack of translation symmetry. The assumption that $\oldhat{\mathfrak{h}}$ exponentially decays with rate proportional to $L$ (Assumption \ref{as:h_regularity}), because of the form \eqref{eq:form_of_mathfrak_h} of $\mathfrak{h}$, will allow us to considerably simplify \eqref{eq:off_diagonal_Bloch} in what follows. 

\subsection{Moir\'e lattice, moir\'e reciprocal lattice, and moir\'e potential of twisted bilayer graphene}

The moir\'e lattice, moir\'e reciprocal lattice, and moir\'e potential of TBG will play a fundamental role in the analysis of the following sections. The moir\'e reciprocal lattice vectors are (see Figure~\ref{fig:TBG_twisted_BZs})
\begin{equation}
    \vec{b}_{m,1} := \vec{b}_{1,1} - \vec{b}_{2,1}, \quad \vec{b}_{m,2} := \vec{b}_{1,2} - \vec{b}_{2,2}.
\end{equation}

In terms of the distance between the Dirac points of the layers
\begin{equation} \label{eq:Delta_K}
    | \Delta \vec{K} | := 2 | \vec{K} | \sin\left( \frac{\theta}{2} \right),
\end{equation}
the moir\'e reciprocal lattice vectors have the explicit forms
\begin{equation} \label{eq:moire_reciprocal_lattice_vecs}
    \vec{b}_{m,1} = \sqrt{3} | \Delta \vec{K} | \left( \frac{1}{2}, -\frac{\sqrt{3}}{2} \right)^\top, \quad \vec{b}_{m,2} = \sqrt{3} | \Delta \vec{K} | \left( \frac{1}{2} , \frac{\sqrt{3}}{2} \right)^\top.
\end{equation}
We define the moir\'e reciprocal lattice and Brillouin zone in terms of $B_m := (\vec{b}_{m,1},\vec{b}_{m,2})$ as (see Figure~\ref{fig:TBG_momentum_lattice})
\begin{equation}
    \Lambda_m^* := \{ \vec{G}_m = B_m \vec{n} : \vec{n} \in \mathbb{Z}^2 \}, \quad \Gamma_m^* := \left\{ B_m \beta : \beta \in \left[ 0 , 1 \right)^2 \right\}.
\end{equation}
\begin{figure}
    \centering
    \begin{subfigure}[b]{.56\columnwidth}
    \centering
    \includegraphics[scale=.4]{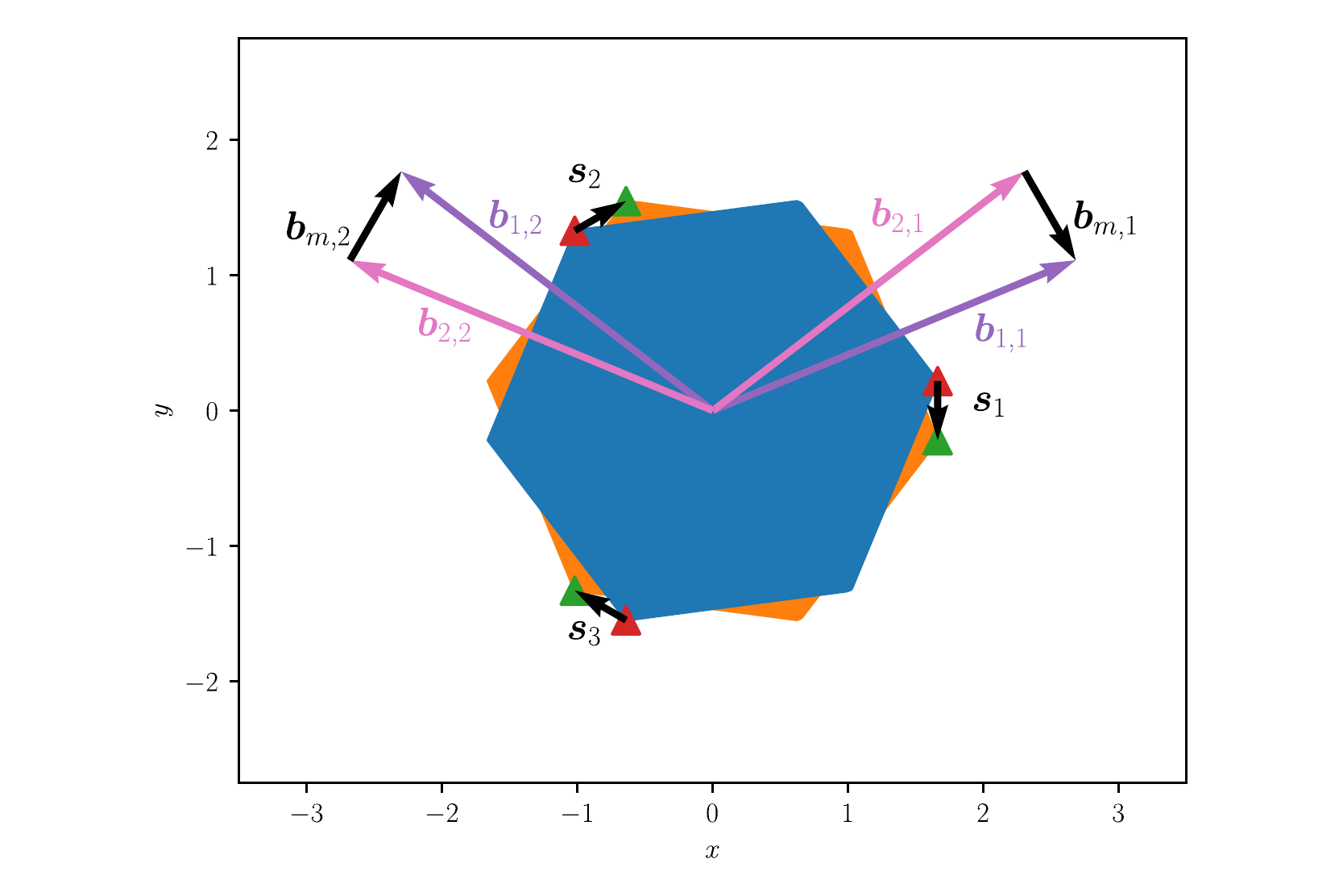}
    \caption{\label{fig:TBG_twisted_BZs}}
    \end{subfigure}
    \begin{subfigure}[b]{.42\columnwidth}
    \centering
    \includegraphics[scale=.4]{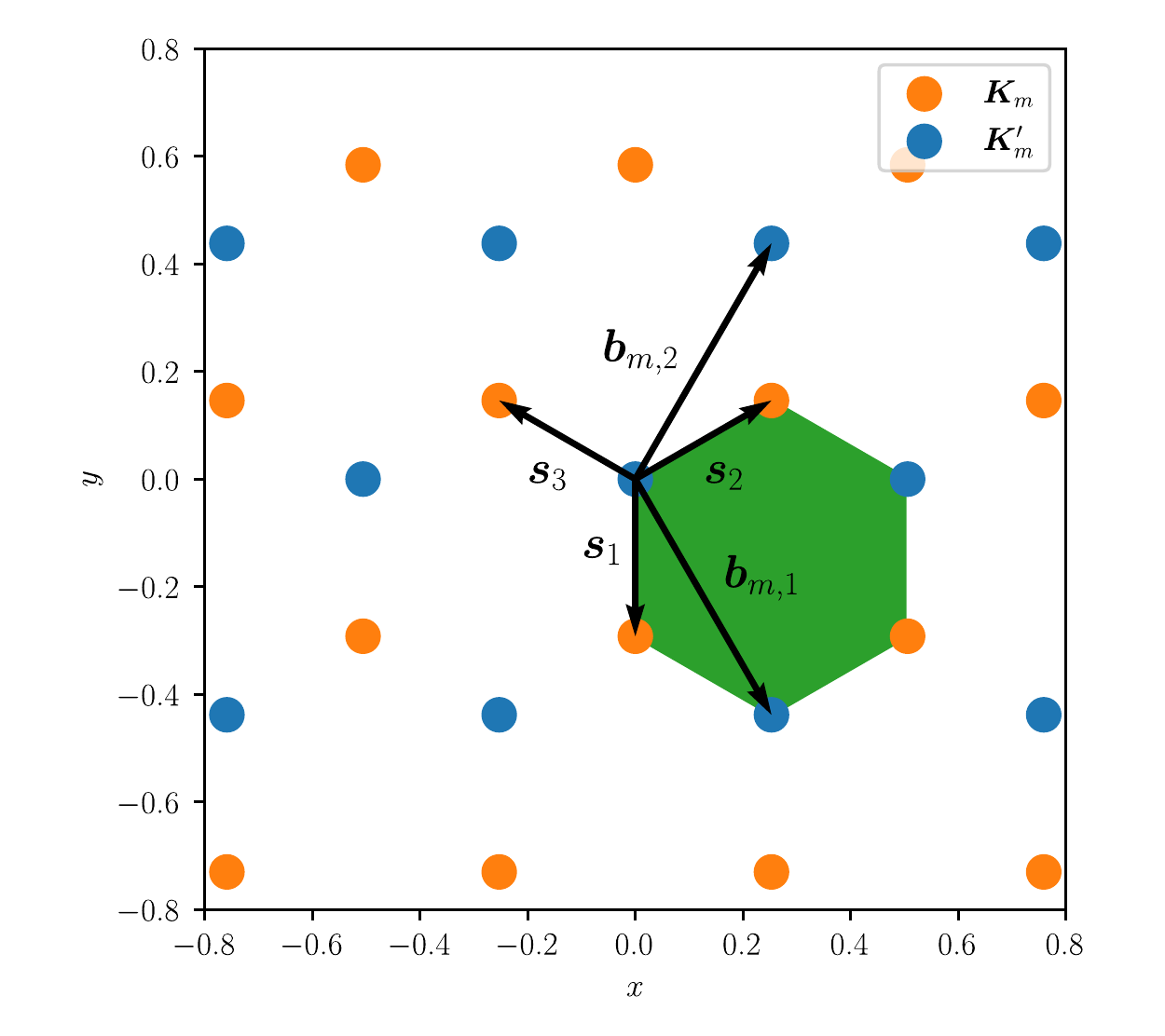}
    \caption{\label{fig:TBG_momentum_lattice}}
    \end{subfigure}
    \caption{ \subref{fig:TBG_twisted_BZs} Plot of the rotated monolayer reciprocal lattice vectors, the moir\'e reciprocal lattice vectors, the rotated Brillouin zones $\Gamma_1^*$ (orange) and $\Gamma_2^*$ (blue) and the momentum hops $\vec{s}_i$, $i \in \{1,2,3\}$. The $\vec{K}$ points of layers $1$ and $2$ are marked with green and red triangles, respectively. \subref{fig:TBG_momentum_lattice} Plot of the moir\'e momentum lattice, formed by translating the monolayer Dirac points $\vec{K}_1$ and $\vec{K}_2$ (labeled here as the moir\'e Dirac points $\vec{K}_m$ and $\vec{K}_m'$, respectively) by moir\'e reciprocal lattice vectors. The Bistritzer-MacDonald model we derive keeps only the nearest-neighbor hops in this lattice, corresponding to $\vec{s}_i$, $i \in \{1,2,3\}$. A choice of moir\'e Brillouin zone is shown shaded in green.}
    
\end{figure}

The associated moir\'e lattice vectors, defined by the equation $\vec{b}_{m,i} \cdot \vec{a}_{m,j} = 2 \pi \delta_{ij}$ for $1 \leq i,j \leq 2$, are
\begin{equation} \label{eq:moire_lattice_vecs}
    \vec{a}_{m,1} := \frac{4 \pi}{3 | \Delta \vec{K} |} \left( \frac{\sqrt{3}}{2}, -\frac{1}{2} \right)^\top, \quad \vec{a}_{m,2} := \frac{4 \pi}{3 | \Delta \vec{K} |} \left( \frac{\sqrt{3}}{2} , \frac{1}{2} \right)^\top.
\end{equation}
For small $\theta$, we see that the moir\'e reciprocal lattice vectors are much longer than the single layer lattice vectors: $|\vec{a}_{m,i}| \gg |\vec{a}_i|, i \in \{1,2\}$. The moir\'e lattice and unit cell are then
\begin{equation} \label{eq:moire_lattice}
    \Lambda_m := \{ \vec{R}_m = A_m \vec{m} : \vec{m} \in \mathbb{Z}^2 \}, \quad \Gamma_m := \left\{ A_m \alpha : \alpha \in \left[ 0 , 1 \right)^2 \right\},
\end{equation}
where $A_m := ( \vec{a}_{m,1} , \vec{a}_{m,2} )$. The monolayer $\vec{K}$ points are known as the moir\'e $\vec{K}$ and $\vec{K}'$ points
\begin{equation}
    \vec{K}_m := \vec{K}_1, \quad \vec{K}'_m := \vec{K}_2.
\end{equation}
We introduce notation for the momentum difference between the monolayer Dirac points and its rotations\footnote{The rotations by $\frac{2 \pi}{3}$ are, equivalently, the momentum differences between the Dirac points measured with respect to the equivalent monolayer Dirac points obtained by rotating the Dirac points \eqref{eq:Dirac_pts} by $\frac{2 \pi}{3}$.} by $\frac{2 \pi}{3}$ (see Figure~\ref{fig:TBG_twisted_BZs})
\begin{equation} \label{eq:momentum_space_hops}
    \begin{split}
        &\vec{s}_1 := \vec{K}_1 - \vec{K}_2, \quad \vec{s}_2 := \vec{K}_1 - \vec{K}_2 - \vec{b}_{2,2} + \vec{b}_{1,2} = \vec{s}_1 + \vec{b}_{m,2}, \\
        &\vec{s}_3 := \vec{K}_1 - \vec{K}_2 + \vec{b}_{2,1} - \vec{b}_{1,1} = \vec{s}_1 - \vec{b}_{m,1}.
    \end{split}
\end{equation}
\begin{figure}
    
\end{figure}
These vectors all have length $|\Delta \vec{K}|$, and can be written explicitly as
\begin{equation} 
    \vec{s}_1 = |\Delta \vec{K}| \left( 0 , -1 \right)^\top, \quad \vec{s}_2 = |\Delta \vec{K}| \left( \frac{\sqrt{3}}{2} , \frac{1}{2} \right)^\top, \quad \vec{s}_3 = |\Delta \vec{K}| \left( - \frac{\sqrt{3}}{2} , \frac{1}{2} \right)^\top.
\end{equation}
We emphasize that the moir\'e reciprocal lattice, moir\'e lattice, and momentum hops \eqref{eq:momentum_space_hops} all depend on $\theta$ through $|\Delta \vec{K}|$.

We finally introduce the continuum (matrix-valued) moir\'e potential, defined for $\sigma, \sigma' \in \{A,B\}$ and $\vec{r} \in \mathbb{R}^2$ by
\begin{equation} \label{eq:moire_potential}
    T_\theta^{\sigma \sigma'}(\vec{r}) := \sum_{j = 1}^3 T_j^{\sigma \sigma'} e^{- i \vec{s}_j \cdot \vec{r}},
\end{equation}
where
\begin{equation} \label{eq:hopping_matrices}
    \begin{split}
        &T_1^{\sigma \sigma'} := 1, \quad T_2^{\sigma \sigma'} := e^{i \left[ \vec{b}_{1,2} \cdot \vec{\tau}_1^\sigma - \vec{b}_{2,2} \cdot \vec{\tau}_2^{\sigma'} \right]}, \\ 
        &T_3^{\sigma \sigma'} := e^{i \left[ - \vec{b}_{1,1} \cdot \vec{\tau}_1^\sigma + \vec{b}_{2,1} \cdot \vec{\tau}_2^{\sigma'}\right]},
    \end{split}
\end{equation}
are the hopping matrices. Written out in full, the hopping matrices are
\begin{equation} \label{eq:hopping_matrices_full}
    \begin{split}
        &T_1 = \begin{pmatrix} 1 & 1 \\ 1 & 1 \end{pmatrix}, \quad T_2 = e^{- i \vec{b}_2 \cdot \vec{\mathfrak{d}}} \begin{pmatrix} 1 & e^{- i \phi} \\ e^{i \phi} & 1 \end{pmatrix},  \\ 
        &T_3 = e^{i \vec{b}_1 \cdot \vec{\mathfrak{d}}} \begin{pmatrix} 1 & e^{i \phi} \\ e^{- i \phi} & 1 \end{pmatrix}.
    \end{split}
\end{equation}
Note that when $\vec{\mathfrak{d}} = - \vec{\tau}^B$ we have $\vec{b}_1 \cdot \vec{\mathfrak{d}} = \vec{b}_2 \cdot \vec{\mathfrak{d}} = - \phi$ and we recover the hopping matrices appearing in equation (7) of \cite{Bistritzer2011}. We provide a detailed motivation and derivation of the moir\'e potential in the proof of Lemma \ref{lem:off_diag_lemma}. We will see that the $\mathfrak{d}$-dependent shifts in \eqref{eq:hopping_matrices_full} can always be removed by conjugating by a unitary (Lemma \ref{lem:interlayer_displacement}).

The moir\'e potential has two important and easily verifiable properties. It is periodic (up to a phase) with respect to the moir\'e lattice
\begin{equation} \label{eq:moire_potential_periodic}
    T_\theta^{\sigma \sigma'}(\vec{r} + \vec{R}_m) = e^{- i \vec{s}_j \cdot \vec{R}_m } T_\theta^{\sigma \sigma'}(\vec{r}), \quad \forall \vec{R}_m \in \Lambda_m, \quad \sigma, \sigma' \in \{A,B\},\quad j=1,2,3,
\end{equation}
and its rate of oscillation is controlled by $\theta$. The specific property we will require is the existence of a constant $C_T > 0$ such that for all $\sigma, \sigma' \in \{A,B\}$,
\begin{equation} \label{eq:moire_potential_theta_bound}
    \sup_{1 \leq a \leq 2} \left| \pdf{r_a} T_\theta^{\sigma \sigma'}(\vec{r}) \right| \leq C_T \theta, \quad \forall \vec{r} \in \mathbb{R}^2.
\end{equation}
This is, of course, a reflection of the fact that the moir\'e cell becomes larger for smaller angles.

\subsection{Non-dimensionalization and assumptions on interlayer hopping function} \label{sec:nondim}

We close this section by non-dimensionalizing the tight-binding model of twisted bilayer graphene in order to identify fundamental dimensionless parameters. In the following sections, we will identify the parameter regime where the Bistritzer-MacDonald model represents the dominant effective dynamics. We start from the time-dependent Schr\"odinger equation for $\psi(\tau) : [0,\infty) \rightarrow \left( \ell^2(\mathbb{Z}^2;\mathbb{C}^2) \right)^2$
\begin{equation} \label{eq:TD_Schro}
    i \hbar \de_\tau \psi = H \psi, \quad \psi(0) = \psi_0.
\end{equation}
It is natural to non-dimensionalize so that $v_D = a = \hbar = 1$. To this end, we introduce dimensionless time and position variables
\begin{equation} \label{eq:tauprimerprime}
    \tau' := \frac{ \sqrt{3} t }{ 2 \hbar } \tau, \quad \vec{r}' := \frac{\vec{r}}{a}, \quad \psi'(\vec{r}',\tau') := \psi(\vec{r},\tau).
\end{equation}
With this change, dividing \eqref{eq:TD_Schro} through by the energy scale 
\begin{equation} \label{eq:energy_scale}
    \mathcal{E} := \frac{ \sqrt{3} t }{ 2 } = \frac{ v_D \hbar }{ a }
\end{equation}
leaves both sides dimensionless. Recall the definition of the interlayer hopping function \eqref{eq:form_of_mathfrak_h}. We introduce a dimensionless interlayer hopping function $h$ by\footnote{At first sight, it would appear natural to define $h$ so that $h(\vec{r}';\ell) = \frac{2}{\sqrt{3} t} \mathfrak{h}(\vec{r};L)$. However, the choice \eqref{eq:h_def} is more convenient because, using $|\Gamma| = \frac{\sqrt{3}}{2} a^2$, it ensures that $\oldhat{h}(\vec{\xi}';\ell) = \frac{2}{\sqrt{3} t} \frac{ \oldhat{\mathfrak{h}}(\vec{\xi};L) }{ |\Gamma| }$, where $\vec{\xi}' = a \vec{\xi}$ (we use this in \eqref{eq:energy_ratio}).}
\begin{equation} \label{eq:h_def}
    h( \vec{r}' ; \ell ) := h^\sharp( \sqrt{ |\vec{r}'|^2 + \ell^2 } ) := \frac{ 4 }{ 3 t } \mathfrak{h}^\sharp\left( a \sqrt{ |\vec{r}'|^2 + \ell^2 } \right) = \frac{ 4 }{ 3 t } \mathfrak{h}( a \vec{r}' ; a \ell ),
\end{equation}
where $\ell := \frac{L}{a}$ is the dimensionless ratio of the interlayer distance to the monolayer lattice constant, and $h^\sharp$ is a fixed, dimensionless, decaying function. We finally drop the $'$s and just write $\vec{r}$ and $\tau$ for the dimensionless spatial and temporal variables. Just as for $\mathfrak{h}$, we define the Fourier transform pair
\begin{equation} \label{eq:h_FT}
    \oldhat{{h}}(\vec{\xi};\ell) := \inty{\field{R}^2}{}{ e^{- i \vec{\xi} \cdot \vec{r}} {h}(\vec{r};\ell) }{\vec{r}}, \quad {h}(\vec{r};\ell) = \frac{1}{(2 \pi)^2} \inty{\field{R}^2}{}{ e^{i \vec{\xi} \cdot \vec{r}} \oldhat{{h}}(\vec{\xi};\ell) }{\vec{k}}.
\end{equation}

In what follows, we will often use notation defined in the previous sections for dimensionful quantities, such as vectors in the monolayer lattice $\vec{R} \in \Lambda$, for the \emph{dimensionless} quantities obtained by setting $a = 1$. We do this wherever it is unlikely to cause confusion. The exceptions are our notations for the non-dimensionalized Dirac points $\vec{\kappa} := \vec{K} a$ and non-dimensionalized momentum shifts $\mathfrak{s}_n := \vec{s}_n a$, where we introduce separate notation so that we can write expressions such as \eqref{eq:energy_ratio}, and compare \eqref{eq:BM_model_0} and \eqref{eq:physical_BM}, without ambiguity. 

The non-dimensionalized model in real space is, then, 
\begin{equation} \label{eq:nondimensional_schro}
    i \de_\tau \psi = H \psi, \quad \psi(0) = \psi_0,
\end{equation}
where $H$ is as in \eqref{eq:block_H}, with 
\begin{equation} \label{eq:block_H_diag}
    \left( H_{ii} \psi_i \right)_{\vec{R}_i} = - \frac{2}{\sqrt{3}} \begin{pmatrix} \psi_{\vec{R}_i}^B + \psi_{\vec{R}_i-\vec{a}_{i,1}}^B + \psi_{\vec{R}_i-\vec{a}_{i,2}}^B \\ \psi_{\vec{R}_i}^A + \psi^A_{\vec{R}_i + \vec{a}_{i,1}} + \psi^A_{\vec{R}_i + \vec{a}_{i,2}} \end{pmatrix}, \quad i \in \{1,2\},
\end{equation}
\begin{equation} \label{eq:block_H_offdiag}
    \left( H_{12} \psi_2 \right)^\sigma_{\vec{R}_1} = \frac{\sqrt{3}}{2} \sum_{\vec{R}_2 \in \Lambda_2} \sum_{\sigma' \in \{A,B\}} {h}\left( \vec{R}_{ij} + \vec{\tau}_{ij}^{\sigma \sigma'} ; \ell \right) \psi^{\sigma'}_{\vec{R}_j}, \quad \vec{R}_{ij} = \vec{R}_i - \vec{R}_j, \quad \vec{\tau}_{ij}^{\sigma \sigma'} = \vec{\tau}_{i}^\sigma - \vec{\tau}_j^{\sigma'}.
\end{equation}
In momentum space, the model is as in \eqref{eq:Bloch_block_H}, where the diagonal terms are
\begin{equation}
    \left( \mathcal{G}_i H_{ii} \mathcal{G}_i^{-1} \tilde{\psi}_i \right)(\vec{k}_i) = H_i(\vec{k}_i) \tilde{\psi}_i(\vec{k}_i), \quad i \in \{1,2\},
\end{equation}
\begin{equation}
    H_{i}(\vec{k}_i) = - \frac{2}{\sqrt{3}} \begin{pmatrix} 0 & F(\vec{k}_i) \\ \overline{F(\vec{k}_i)} & 0 \end{pmatrix}, \quad F(\vec{k}_i) = e^{i \vec{k}_i \cdot ( \vec{\tau}_i^B - \vec{\tau}_i^A )} ( 1 + e^{- i \vec{k}_i \cdot \vec{a}_{i,1}} + e^{- i \vec{k}_i \cdot \vec{a}_{i,2}} ),
\end{equation}
and the off-diagonal terms are
\begin{equation}  \label{eq:off_diag_1}
    \begin{split}
        &\left( \mathcal{G}_1 H_{12} \mathcal{G}_2^{-1} \tilde{\psi}_2 \right)^\sigma(\vec{k}_1) =  \\
        &\sum_{\sigma' \in \{A,B\}} \inty{\Gamma^*_2}{}{ \sum_{\vec{G}_1 \in \Lambda_1^*} \sum_{\vec{G}_2 \in \Lambda_2^*} e^{i [\vec{G}_1 \cdot \vec{\tau}_1^\sigma - \vec{G}_2 \cdot \vec{\tau}_2^{\sigma'}]} \oldhat{h}( \vec{k}_1 + \vec{G}_1 ; \ell ) \delta( \vec{k}_1 + \vec{G}_1 - \vec{k}_2 - \vec{G}_2 ) \tilde{\psi}^{\sigma'}_2(\vec{k}_2) }{\vec{k}_2}.
    \end{split}
\end{equation}
Note that in the final model, there are just two dimensionless parameters:
\begin{itemize}
    \item $\ell$, the ratio of interlayer distance to the monolayer lattice constant,
    \item $\theta$, the twist angle.
\end{itemize}
We will assume in what follows that these parameters are positive, i.e.,
\begin{equation} \label{eq:nonnegative}
    \ell > 0, \quad \theta > 0.
\end{equation}
This assumption on $\ell$ is natural. We make the assumption on $\theta$ without loss of generality since it is straightforward to generalize all our results to the case where $\theta < 0$. We avoid the case $\theta = 0$, where the moir\'e lattice is not well-defined. In this case, the system is straightforward to analyze since it is periodic with respect to the monolayer graphene lattice (see, e.g. \cite{RevModPhys.81.109}).

We now discuss the assumptions we require on the interlayer hopping function $h$, starting with the following proposition, which follows immediately from \eqref{eq:h_def}.
\begin{proposition} \label{prop:hankel}
    Let $r := |\vec{r}|$ and $\xi := |\vec{\xi}|$, and the nondimensionalized interlayer hopping function $h$ have the form \eqref{eq:h_def}. Then the Fourier transform of $h$ is radial, and can be calculated through the Hankel transform 
    \begin{equation}
        \oldhat{h}(\vec{\xi};\ell) = \oldhat{h}(\xi;\ell) = 2 \pi \inty{0}{\infty}{ J_0(\xi r) r h^\sharp\left( \sqrt{ r^2 + \ell^2 } \right) }{r},
    \end{equation}
    where $J_0(\zeta)$ is the Bessel function of the first kind
    \begin{equation}
        J_0(\zeta) := \frac{1}{2 \pi} \inty{0}{2 \pi}{ e^{i \zeta \cos \theta} }{\theta}.
    \end{equation}
\end{proposition}
The two-dimensional Fourier transform of the interlayer hopping function can alternatively be calculated as the one-dimensional inverse Fourier transform of the three-dimensional Fourier transform of the interlayer hopping function.
\begin{proposition} \label{prop:alternative_formulation}
    Let $r, \xi$, and $h$ be as in Proposition \ref{prop:hankel}. It is straightforward to see that the three-dimensional (with respect to $\vec{r}$ and $\ell$) Fourier transform of $h^\sharp\left(\sqrt{r^2 + \ell^2}\right)$ is
    \begin{equation} \label{eq:3d_FT}
        \mathring{h}^\sharp\left(\sqrt{ \xi^2 + \omega^2 }\right) = 4 \pi \inty{0}{\infty}{ \frac{ \sin\left( \sqrt{ \xi^2 + \omega^2 } R \right) }{ \sqrt{ \xi^2 + \omega^2 } R } R^2 h^\sharp(R) }{R}.
    \end{equation}
    We then have that
    \begin{equation}
        \oldhat{h}(\vec{\xi};\ell) = \oldhat{h}(\xi;\ell) = \frac{1}{2 \pi} \inty{\mathbb{R}}{}{ e^{ i \omega \ell } \mathring{h}^\sharp\left( \sqrt{ \xi^2 + \omega^2 } \right) }{\omega}.
    \end{equation}

\end{proposition}

We now make our main assumptions on $h$.
\begin{assumption} \label{as:h_regularity}
    Let $r := |\vec{r}|$ and $\xi := |\vec{\xi}|$. 
    We assume that $h$ has the form \eqref{eq:h_def} for some $h^\sharp$, and that the Fourier transform \eqref{eq:h_FT} exists. We assume that $\oldhat{h}$ can be bounded above and below by decaying exponentials, more specifically, that there exist constants $C_1$, $C_2$, $D_1$, $D_2$, $\ell_0 > 0$, with $D_1 \geq D_2$, such that
    \begin{equation} \label{eq:D_def}
        C_1 e^{- D_1 \ell \xi} \leq \oldhat{h}(\vec{\xi};\ell) \leq C_2 e^{- D_2 \ell \xi}, \quad \forall \vec{\xi} \in \mathbb{R}^2, \forall \ell \geq \ell_0.
    \end{equation}
    We assume further that the constants $D_1$ and $D_2$ satisfy\footnote{Note that the lower bound already follows from $D_1 \geq D_2$, so the upper bound is the non-trivial assumption here.}
    \begin{equation} \label{eq:D1D2}
        1 \leq \frac{D_1}{D_2} < 2.
    \end{equation}
    We finally assume that $\oldhat{h}$ is Lipschitz continuous with a Lipschitz constant which exponentially decays, more specifically, that there exists a constant $C_3 > 0$ such that 
    \begin{equation} \label{eq:Lipschitz}
        | \oldhat{h}(\vec{\xi};\ell) - \oldhat{h}(\vec{\xi}';\ell) | \leq C_3 e^{- D_2 \ell d\left(\vec{0},\left[\vec{\xi},\vec{\xi}'\right]\right)} | \vec{\xi} - \vec{\xi}' |, \quad \forall \vec{\xi}, \vec{\xi}' \in \mathbb{R}^2, \forall \ell \geq \ell_0,
    \end{equation}
    where $d\left(\vec{0},\left[\vec{\xi},\vec{\xi}'\right]\right)$ denotes the minimum distance between the straight line connecting $\vec{\xi}$ and $\vec{\xi}'$ and the origin (see Remark \ref{rem:lipschitz}).
\end{assumption}
We discuss briefly the importance of each part of Assumption \ref{as:h_regularity}. Radial symmetry of $\oldhat{h}$ allows us to simplify the moir\'e potential \eqref{eq:moire_potential} (see \eqref{eq:simplify_moire_potential}), and simplifies the proof of Lemma \ref{lem:off_diag_lemma}. We expect that this assumption can be relaxed; see Remark \ref{rem:relaxing}. 

The lower bound in \eqref{eq:D_def}, together with the upper bound in \eqref{eq:D_def} and the inequality \eqref{eq:D1D2}, guarantee that the interlayer terms corresponding to the momentum space hops \eqref{eq:momentum_space_hops} are rigorously larger, in the regime \eqref{eq:balance__0}, than all other momentum space hops arising from \eqref{eq:off_diag_1}. To see this, note that the interlayer terms corresponding to the momentum space hops \eqref{eq:momentum_space_hops} are proportional to $\oldhat{h}(|\vec{\kappa}|;\ell)$, while the next largest interlayer terms are bounded by $C_2 e^{- 2 D_2 \ell |\vec{\kappa}|}$. Combining \eqref{eq:D_def} with \eqref{eq:D1D2} ensures that 
\begin{equation}
    \frac{ e^{- 2 D_2 \ell |\vec{\kappa}|} }{ \oldhat{h}(|\vec{\kappa}|;\ell) } \lesssim e^{ (D_1 - 2 D_2) \ell |\vec{\kappa}| } \rightarrow 0 \quad \text{ as } \ell \rightarrow \infty. 
\end{equation}
For more detail, see Section \ref{sec:proof_key_lemma}, in particular the proof of Lemma \ref{lem:off_diag_lemma}. The lower bound in \eqref{eq:D_def} may fail in some more realistic models; see Remark \ref{rem:relaxing}. 

Finally, the estimate \eqref{eq:Lipschitz} allows for the bound \eqref{eq:r2_bound} in the proof of Lemma \ref{lem:off_diag_lemma}. This bound justifies an approximation where the momentum space interlayer hopping strength $\oldhat{h}$, which is, in general, wavenumber-dependent, is treated as wavenumber-independent. The wavenumber-dependence of the interlayer hopping strength can be seen clearly in \eqref{eq:off_diagonal_Bloch}, where $\oldhat{h}$ depends on $\vec{k}_1$. The approximation can be seen clearly in \eqref{eq:local_approx}, where the un-approximated interlayer hopping is the left-hand side, and the approximated interlayer hopping is the first term on the right-hand side. This approximation is known in the physics literature as the ``local'' approximation; see, for example, \cite{Xie2021}.



Using a table of Hankel transform pairs \cite{Bateman1954a}, we find that Assumption \ref{as:h_regularity} can be verified for some specific $h$. Note that exponential decay of the interlayer hopping function in real space is to be expected when the tight-binding model has been defined through a basis of exponentially localized Wannier functions \cite{MarzariVanderbilt1997,MarzariMostofiYatesSouzaVanderbilt2012} (see also Remark \ref{rem:orbital_decay}).
\begin{example} \label{ex:example_1}
The function
    \begin{equation} \label{eq:example_1}
        h(\vec{r};\ell) = \frac{1}{(r^2 + \ell^2)^{3/2}}
    \end{equation}
    satisfies Assumption \ref{as:h_regularity}, with $D_2 = 1$, any $\ell_0 > 0$ and any $1 < D_1 \leq 2$. Its Fourier transform is
    \begin{equation} \label{eq:example_FT_1}
        \oldhat{h}(\vec{\xi};\ell) = 2 \pi \frac{ e^{- \ell \xi} }{ \ell }.
    \end{equation}
\end{example}
\begin{example} \label{ex:example_2}
For any $\alpha > 0$, the function
    \begin{equation} \label{eq:example_2}
        h(\vec{r};\ell) = e^{- \alpha \sqrt{ r^2 + \ell^2 }}
    \end{equation}
    satisfies Assumption \ref{as:h_regularity}, for any $\ell_0 > 0$, and any pair $D_1 > 1$ and $0 < D_2 < 1$ such that \eqref{eq:D1D2} holds. Its Fourier transform is
    \begin{equation} \label{eq:example_FT_2}
        \oldhat{h}(\vec{\xi};\ell) = 2 \pi \frac{\alpha e^{- \ell \sqrt{ \xi^2 + \alpha^2 } } \left( 1 + \ell \sqrt{ \xi^2 + \alpha^2 } \right) }{ ( \xi^2 + \alpha^2 )^{3/2} }.
    \end{equation}
\end{example}
\begin{example} \label{ex:example_3}
For any $\alpha > 0$, the function
    \begin{equation} \label{eq:example_3}
        h(\vec{r};\ell) = \frac{ e^{- \alpha \sqrt{ r^2 + \ell^2 }} }{ \sqrt{ r^2 + \ell^2 } }
    \end{equation}
    satisfies Assumption \ref{as:h_regularity}, for any $\ell_0 > 0$, and any pair $D_1 > 1$ and $0 < D_2 < 1$ such that \eqref{eq:D1D2} holds. Its Fourier transform is
    \begin{equation} \label{eq:example_FT_3}
        \oldhat{h}(\vec{\xi};\ell) = 2 \pi \frac{e^{- \ell \sqrt{ \xi^2 + \alpha^2 } }}{ \sqrt{ \xi^2 + \alpha^2 } }.
    \end{equation}
\end{example}
Note that in all of Examples \ref{ex:example_1}-\ref{ex:example_3}, we can take $D_1 = 1 + \epsilon$, and $D_2 = 1 - \epsilon$ for \emph{any} $\epsilon > 0$, so that \eqref{eq:D1D2} holds comfortably.

Examples \ref{ex:example_1}-\ref{ex:example_3} suggest that Assumption \ref{as:h_regularity} may hold for fairly general $h^\sharp$, with $D_1 = 1 + \epsilon$ and $D_2 = 1 - \epsilon$ for any $\epsilon > 0$. However, we can prove this only for the upper bound in \eqref{eq:D_def}, by a straightforward application of Cauchy's theorem (although see Remark \ref{rem:lower_bounds}). Although the following Proposition is surely well-known, we include a proof for completeness, especially since it is so simple.
\begin{proposition} \label{prop:Paley-Wiener}
    Assume that $h^\sharp(\zeta)$ is analytic in $\zeta$, except possibly at the origin, and satisfies the bound
    \begin{equation} \label{eq:hsharp_decay}
        \left| h^\sharp(\zeta) \right| \leq C |\text{\emph{Re} }\zeta|^{-2-\delta}, \quad \Re \zeta \geq R,
    \end{equation}
    for some fixed $C, R, \delta > 0$. Then, for any fixed $0 < \epsilon < 1$, there exists a constant $C_\epsilon > 0$ such that
    \begin{equation} \label{eq:bound}
        \left| \oldhat{h}(\xi;\ell) \right| \leq C_\epsilon e^{- \ell (1 - \epsilon) \xi}.
    \end{equation}
\end{proposition}
\begin{proof}
    Since $\oldhat{h}(\vec{\xi};\ell) = \oldhat{h}(\xi;\ell)$ is radial, it suffices to prove \eqref{eq:bound} by setting $\vec{\xi} = ( \xi , 0 )^\top$, with $\xi > 0$, in the Fourier transform formula \eqref{eq:h_FT}, resulting in
    \begin{equation} \label{eq:F_integral}
        \oldhat{h}(\xi;\ell) = \inty{\mathbb{R}}{}{ \inty{\mathbb{R}}{}{ e^{- i \xi x} h^\sharp\left( \sqrt{ x^2 + y^2 + \ell^2 } \right) }{x} }{y}.
    \end{equation}
    Pushing the $x$ integral below the real axis using Cauchy's theorem (using \eqref{eq:hsharp_decay} to control the ``ends'' of the rectangle) we have that
    \begin{equation} \label{eq:contour}
        \oldhat{h}(\xi;\ell) = e^{- \ell (1 - \epsilon) \xi} \inty{\mathbb{R}}{}{ \inty{\mathbb{R}}{}{ e^{- i \xi x} h^\sharp\left( \sqrt{ x^2 - 2 i \ell (1 - \epsilon) x - \ell^2 (1 - \epsilon)^2 + y^2 + \ell^2 } \right) }{x} }{y}.
    \end{equation}
    Using decay of $h^\sharp$ \eqref{eq:hsharp_decay}, the double integral converges and we are done.
\end{proof}
The alternative formulation given in Proposition \ref{prop:alternative_formulation} allows for the following
\begin{proposition} \label{prop:Paley-Wiener_alt}
    Assume that $\mathring{h}^\sharp(\zeta)$ is analytic in $\zeta$, except possibly at the origin, and satisfies the bound
    \begin{equation}
        | \mathring{h}^\sharp(\zeta) | \leq C | \Re \zeta |^{-1-\delta}, \quad \Re \zeta \geq R,
    \end{equation}
    for some fixed $C, R, \delta > 0$. Then, for any fixed $0 < \epsilon < 1$, there exists a constant $C_\epsilon > 0$ such that
    \begin{equation}
        \left| \oldhat{h}(\xi;\ell) \right| \leq C_\epsilon e^{- \ell ( 1 - \epsilon ) \xi}.
    \end{equation}
\end{proposition}
\begin{proof}
    The proof is exactly analogous to that of Proposition \ref{prop:Paley-Wiener}, except that we push the $\omega$ integral up to $\mathbb{R} + i (1 - \epsilon) \xi$.
\end{proof}
\begin{remark} \label{rem:lower_bounds}
    The argument given in the proof of Proposition \ref{prop:Paley-Wiener} could possibly be extended under additional assumptions on $h^\sharp$ to give lower bounds, or even to evaluate $\oldhat{h}(\xi;\ell)$ explicitly. For example, when $h$ is as in Example \ref{ex:example_1}, the integral over $y$ in \eqref{eq:F_integral} is explicit, and the resulting integrand in $x$ has a simple pole at $- i \ell$, so that pushing the integral below the real axis leads immediately to \eqref{eq:example_FT_1}. More generally, the integrand of \eqref{eq:contour} has a branch point at $- i \ell$. In this case, one can fix the branch cut to point downwards, and then further deform the $x$ integral of \eqref{eq:contour} into the integral either side of the branch cut. Such investigations do not seem to lead in a straightforward way to general conditions on $h^\sharp$ guaranteeing lower bounds, so we do not consider this further in this work.
\end{remark}

\begin{remark} \label{rem:lipschitz}
    Note that \eqref{eq:Lipschitz} follows easily if $\oldhat{h}$ is continuously differentiable everywhere and that derivative decays exponentially with distance from the origin (with rate proportional to $\ell$). We require the slightly more obscure statement \eqref{eq:Lipschitz} to allow for $\oldhat{h}$ merely Lipschitz, as in \eqref{eq:example_FT_1}.
\end{remark}
\begin{remark} \label{rem:orbital_decay}
    We remark on rigorous results on decay of orbitals (and hence of the interlayer hopping function). Precise exponential asymptotics of the hopping function \emph{in real space} have been obtained by Helffer-Sj\"ostrand \cite{1984HelfferSjostrand}. Exponential upper and lower bounds on the hopping function have been proved in the context of a continuum model of graphene by Fefferman-Lee-Thorp-Weinstein \cite{2017FeffermanLee-ThorpWeinstein}. These bounds were generalized to the setting of a strong magnetic field by Fefferman-Shapiro-Weinstein \cite{Fefferman2020a,Shapiro2020}. 
\end{remark}
\begin{remark} \label{rem:relaxing}
    The assumption that $h$ is radial, and, more specifically, a function of the three-dimensional distance \eqref{eq:h_def}, is very much an approximation to the physical $h$. It would be more realistic to assume that $h$ is a general decaying function which respects the symmetries of monolayer graphene, i.e., $\frac{2 \pi}{3}$-rotation, inversion ($\vec{r} \mapsto - \vec{r}$), and complex conjugation. Our results will go through, but will be more complicated to state and prove, as long as the bounds \eqref{eq:D_def}, \eqref{eq:D1D2}, and \eqref{eq:Lipschitz}, still hold for such an $h$. Interestingly, the Fourier transforms of some realistic choices of $h$, especially those which account for mechanical relaxation of the bilayer, actually have zeros, so that the lower bound in \eqref{eq:D_def} fails (see Figure 3 of \cite{Massatt2021}). It is actually unclear, in such cases, whether a similar reduction to effective dynamics is possible, and, if it is, what the effective dynamics should be.
\end{remark}
\begin{remark}
   The lattice structure, $\Lambda_i,$ for each layer has been observed to be significantly reconstructed at small twist angles to reduce its elastic energy ~\cite{2018CarrMassattTorrisiCazeauxLuskinKaxiras,KimRelax18,cazeaux2018energy}, and the Hamiltonian must thus be appropriately modified \cite{Massatt2021,KimRelax18,fang2019angledependent} to predict electronic properties accurately. It follows from the results in \cite{Massatt2021} that for small twist angles the decay rates $D_i
   \sim \theta/(d_i+\theta)$ for $d_i>0,$ but we can continue to expect that $1\le D_1/D_2< 2.$
    We will see in Section~\ref{sec:magic_angles} that the unrelaxed Bistritzer-MacDonald model is sufficiently accurate to predict the first magic angle at $\theta \sim 1^\circ$, although a relaxed model is necessary to predict the correct band gap between the nearly-flat bands at the Fermi energy and the other moir\'e bands~\cite{Massatt2021,KimRelax18,fang2019angledependent}.
   Generalizing our rigorous theory to allow for lattice relaxation is the subject of ongoing work.
\end{remark}
  
\section{Approximate solutions by a systematic multiple scales analysis} \label{sec:multiple_scales_analysis}

In this section, we will generate approximate solutions of the time-dependent Schr\"odinger equation by a systematic multiple scales analysis. In the following section, we will prove convergence of these approximate solutions to exact solutions. 

\subsection{Wave-packet \textit{ansatz}} 

In this section, we will introduce the basic ``wave-packet'' {\textit ansatz} we will use to generate approximate solutions. The non-dimensionalized time-dependent Schr\"odinger equation is \eqref{eq:nondimensional_schro}. We assume the following ``wave-packet'' initial condition
\begin{equation} \label{eq:WP_0}
    \begin{split}
        &\psi_0 = \begin{pmatrix} \psi_{1,0} \\ \psi_{2,0} \end{pmatrix}, \\ 
            &( \psi_{i,0} )_{\vec{R}_i}^\sigma = \gamma \left. f^\sigma_{i,0}( \vec{X} )\right|_{\vec{X} = \gamma (\vec{R}_i + \vec{\tau}^\sigma_i)} e^{i \vec{\kappa}_i \cdot (\vec{R}_i + \vec{\tau}_i^\sigma)}, \quad i \in \{1,2\}, \sigma \in \{A,B\},
    \end{split}
\end{equation}
where $\gamma > 0$ is a dimensionless small parameter which can be interpreted in a few (equivalent) ways. Most obviously, small $\gamma$ enforces a separation of scales between the spatial variation of the envelope and plane wave parts of the wave-packet. Equivalently, $\gamma$ measures the concentration of the wave-packet, in momentum space, relative to the monolayer $\vec{K}$ points (see Appendix \ref{sec:WP_momentum}). We can also, however, using the linear relationship between momentum relative to the $\vec{K}$ point and energy defined by the monolayer Dirac dispersion relation \eqref{eq:monolayer_Dirac}, think of $\gamma$ as measuring the spectral width of the wave function relative to the dimensionless energy scale $1$ (which, in physical units, corresponds to $\mathcal{E}$ \eqref{eq:energy_scale}).

We make the following assumption on the initial amplitudes $f_{i,0}$ which, in particular, ensures that $\psi_0 \in \mathcal{H}$.
\begin{assumption} \label{as:f_regularity}
    {\color{black}
    We assume that the initial amplitudes have bounded eighth Sobolev norm
    \begin{equation} \label{eq:f_regularity}
        \sup_{\sigma \in \{A,B\}} \sup_{i \in \{1,2\}} \| f_{i,0}^\sigma \|_{H^8(\mathbb{R}^2)} \leq C_{f_0},
    \end{equation}
    for some constant $C_{f_0}$.
    }
\end{assumption}
{\color{black}
\begin{remark}
    Our results require bounded first and second Sobolev norms because the proofs of Lemmas \ref{lem:off_diag_lemma} and \ref{lem:diag_lemma} involve estimating first and second order Taylor series remainders in momentum space. Requiring bounded \emph{fourth} Sobolev norms allows for estimate \eqref{eq:third_sob_estimate} in the proof of estimate \eqref{eq:r2_bound} of Lemma \ref{lem:off_diag_lemma}. We require higher Sobolev norms to be bounded in order to estimate the $\tilde{\vec{G}}_1 \neq \vec{0}$ terms in \eqref{eq:all_terms} by \eqref{eq:trick_estimate} (and other similar terms); see also Remark \ref{rem:higher_order_terms}. 
\end{remark}
}
We then seek a solution of \eqref{eq:nondimensional_schro} for $\tau \geq 0$ via a ``wave-packet'' {\it ansatz}\footnote{Note that here we use the notation $\Psi$ to distinguish between the wave-packet noted by $\psi$.} 
\begin{equation} \label{eq:corrector_def}
    \psi(\tau) = \Psi(\tau) + \eta(\tau),
\end{equation}
where
\begin{equation} \label{eq:WP}
    \begin{split}
        &\Psi(\tau) = \begin{pmatrix} \Psi_1(\tau) \\ \Psi_2(\tau) \end{pmatrix}, \\ 
        &( \Psi_i )_{\vec{R}_i}^\sigma(\tau) = \gamma \left. f^\sigma_i( \vec{X}, T )\right|_{\vec{X} = \gamma (\vec{R}_i + \vec{\tau}^\sigma_i), T = \gamma \tau} e^{i \vec{\kappa}_i \cdot (\vec{R}_i + \vec{\tau}^\sigma_i)}, \quad i \in \{1,2\}, \sigma \in \{A,B\}
    \end{split}
\end{equation}
is the wave-packet {\it ansatz}\footnote{Note that normally in the {\it ansatz} 
there would be a time-dependent phase part, but in this case, since (by convention) the monolayer Bloch bands equal zero at the Dirac points, this phase is zero for all time.}, and $\eta(\tau)$ is the corrector. To summarize, the additional parameter introduced by our {\it ansatz} is:
\begin{itemize}
    \item $\gamma$, the momentum space spread of the wave-packet relative to the inverse lattice constant. 
\end{itemize}
\begin{remark} \label{rem:WP_scaling}
    Assuming that $\tau$ scales linearly with $\gamma$ in \eqref{eq:WP} is natural because it results in \eqref{eq:left-hand_side} and the leading term of \eqref{eq:diag_decomposition} balancing. Physically, this occurs because the dispersion relation at the Dirac point \eqref{eq:monolayer_Dirac} is linear.
\end{remark}

In what follows, when we wish to suppress the sublattice degree of freedom, we will adopt the obvious shortened notations $f_i = \left( f_i^A , f_i^B \right)^\top$, $f_{i,0} = \left( f_{i,0}^A , f_{i,0}^B \right)^\top, i \in \{1,2\}$. When we also want to suppress the layer degree of freedom we will write $f = \left( f_1 , f_2 \right)^\top$ and $f_0 = \left( f_{1,0}, f_{2,0} \right)^\top$. We will also simply evaluate the multiscale variables $\vec{X}, T,$ etc., whenever there is no danger of confusion.

Assuming that
\begin{equation}
    f^\sigma_i( \vec{X}, 0 ) = f^\sigma_{i,0}( \vec{X} ), \quad i \in \{1,2\}, \sigma \in \{A,B\},
\end{equation}
the corrector satisfies
\begin{equation} \label{eq:residual_def}
    i \de_\tau \eta = H \eta + r, \quad \eta(0) = 0, \quad r := - ( i \de_\tau - H ) \Psi,
\end{equation}
where $r$ is known as the residual. 

Recall that we denote the electronic wave function Hilbert space $\left( \ell^2(\mathbb{Z}^2;\mathbb{C}^2) \right)^2$ of the twisted bilayer by $\mathcal{H}$. We will denote by $\| \cdot  \|_{\mathcal{H}}$ the standard $\mathcal{H}$-norm, i.e.,
\begin{equation} \label{eq:H_norm}
    \| \phi \|_{\mathcal{H}} := \left( \sum_{i = 1,2} \sum_{\vec{R}_i \in \Lambda_i} \sum_{\sigma = A,B} | \phi^{\sigma}_{\vec{R}_i} |^2 \right)^{\frac12}.
\end{equation}

The following lemma shows that the $\mathcal{H}$-norm of the corrector $\eta$ can be bounded in terms of the $\mathcal{H}$-norm of the residual $r$.
\begin{lemma} \label{lem:approx_sols}
    Let $\eta$ satisfy the IVP \eqref{eq:residual_def}. Then
    \begin{equation}
        \| \eta(\tau) \|_{\mathcal{H}} \leq \inty{0}{\tau}{ \| r(\tau') \|_{\mathcal{H}} }{\tau'}.
    \end{equation}
\end{lemma}
\begin{proof}
    The proof is a standard calculation using self-adjointness of $H$; see, for example, Section 3C of \cite{Hagedorn1994}.
\end{proof}
Lemma \ref{lem:approx_sols} obviously implies that
\begin{equation}
    \| \eta(\tau) \|_\mathcal{H} \rightarrow 0 \quad \text{ as long as } \quad \tau \sup_{\tau' \in [0,\tau]} \| r(\tau') \|_\mathcal{H} \rightarrow 0,
\end{equation}
and hence the time-scale on which the wave-packet {\it ansatz} $\Psi$ \eqref{eq:WP} approximates the solution of \eqref{eq:nondimensional_schro} $\psi$ in the $\mathcal{H}$-norm is controlled by the $\mathcal{H}$-norm of the residual $r$. The key lemma of this work is the following, because it identifies a parameter regime where the Bistritzer-MacDonald model emerges as a necessary condition for smallness of $r$ beyond a certain order, and hence as a necessary condition for convergence of $\Psi$ to $\psi$ over a longer time scale.  
\begin{lemma} \label{lem:residual_lemma}
    Let Assumptions \ref{as:h_regularity} and \ref{as:f_regularity} hold. Assume further that 
\begin{equation} \label{eq:balance}
    \oldhat{h}( |\vec{\kappa}| ; \ell ) = \lambda_0 \gamma, \text{ and } \theta \leq \lambda_1 \gamma,
\end{equation}
    for fixed positive constants $\lambda_0, \lambda_1 > 0$. Introduce the scaled moir\'e potential (recall the definition of the moir\'e potential \eqref{eq:moire_potential}) and scaled Bistritzer-MacDonald Hamiltonian\footnote{We denote by $A^\dagger$ the conjugate transpose of a matrix $A$.}
    \begin{equation} \label{eq:BM_H}
        H_{\text{\emph{BM}}} := \begin{pmatrix} \vec{\sigma} \cdot ( - i \nabla_{\vec{X}} ) & \mathcal{T} (\vec{X}) \\ \mathcal{T}^\dagger(\vec{X}) & \vec{\sigma} \cdot ( - i \nabla_{\vec{X}} ) \end{pmatrix}, \quad \mathcal{T}(\vec{X}) := \lambda_0 T_{\theta} \left( \frac{\vec{X}}{\gamma} \right).
    \end{equation}    
    Then, there exist constants $\gamma_0, C, c > 0$ depending only on $\oldhat{h}$, the monolayer operator \eqref{eq:graphene_block}, and the constants $\lambda_0$ and $\lambda_1$, such that, for all $\gamma \leq \gamma_0$, the residual defined by \eqref{eq:residual_def} \label{eq:residual_decomp} satisfies
    \begin{equation}
        \begin{split}
            &r^\sigma_{\vec{R}_i}(\tau) = \\
            &\gamma \left[ \left. \gamma \left( - \left[ i \de_T - H_{\text{\emph{BM}}} \right] f \right)_i^{\sigma}(\vec{X},T) \right|_{\vec{X} = \gamma(\vec{R}_i + \vec{\tau}_i^\sigma),T = \gamma \tau} e^{i \vec{\kappa}_i \cdot ( \vec{R}_i + \vec{\tau}_i^\sigma )} \right] + \mathfrak{r}^\sigma_{\vec{R}_i}(\tau),
        \end{split}
    \end{equation}
    where
    \begin{equation} \label{eq:residual_bound}
        \| \mathfrak{r}^\sigma_i(\tau) \|_{\mathcal{H}} \leq C \gamma^{1+c} \| f(\cdot,\gamma \tau) \|_{H^8(\mathbb{R}^2;\mathbb{C}^4)}.
    \end{equation}
\end{lemma}
\begin{remark} \label{rem:scaling}
    We remark on the origin of the scalings \eqref{eq:balance}. The scaling of $\ell$ with respect to $\gamma$ comes from balancing the leading terms of \eqref{eq:diag_decomposition} and \eqref{eq:off_diag_decomp}. This scaling forces logarithmic growth of $\ell$ with respect to $\gamma$ \eqref{eq:log_growth}, which leads to rapid decay of the interlayer hopping terms \eqref{eq:off_diag_1}, and finally a Bistritzer-MacDonald model \eqref{eq:BM_H} with only \emph{nearest-neighbor} hopping in momentum space. The scaling of $\theta$ with respect to $\gamma$ is essential so that the regularity of solutions of the BM model is preserved as $\gamma \rightarrow 0$ (see Lemma \ref{lem:BM_sols}). Informally, this requires that the scaled moir\'e potential $\mathcal{T}$ appearing in \eqref{eq:BM_H} varies on the scale $1$ (as opposed to, for example, the scale $\gamma$). To see how the scaling of $\theta$ with respect to $\gamma$ enforces this, note that for small $\theta$ we have $|\vec{s}_n| \approx \theta$, so that we can consider $T_\theta(\vec{r})$ as a function of $\theta \vec{r}$, so that $T_\theta\left(\frac{\vec{X}}{\gamma}\right)$ can be considered a function of $\frac{ \theta \vec{X} }{ \gamma }$. It is then clear that the condition for $\mathcal{T}$ to vary at least over the scale $1$ is precisely boundedness of the ratio $\frac{\theta}{\gamma}$.
\end{remark}
\begin{remark}
    The scaling \eqref{eq:balance} can be stated simply (but a little imprecisely) in terms of length-scales. Approximating $\oldhat{h}(|\vec{\kappa}|;\ell) \approx e^{-\ell |\vec{\kappa}|}$ (note that this is essentially what happens for each of the Examples \ref{ex:example_1}-\ref{ex:example_3}), the scaling can be written as 
    \begin{equation}
        \frac{1}{\gamma} \lesssim \frac{1}{\theta}, \quad \ell \sim \ln \frac{1}{\gamma},
    \end{equation}
    where $\frac{1}{\gamma}$ denotes the lengthscale of variation of the wave-packet envelope, and $\frac{1}{\theta}$ denotes the lengthscale of the moir\'e pattern.
\end{remark}
The key point of Lemma \ref{lem:residual_lemma} is that, formally at least, \emph{the $\mathcal{H}$-norm of the first term of \eqref{eq:residual_decomp} is proportional to $\gamma$, while the $\mathcal{H}$-norm of the second is $o(\gamma)$ as $\gamma \rightarrow 0$.} In particular, the residual is $o(\gamma)$ as $\gamma \rightarrow 0$ in the regime \eqref{eq:balance} \emph{if and only if the wave-packet envelopes evolve according to the Bistritzer-MacDonald model.} This is summed up by the following corollary.
\begin{corollary}
    Under the hypotheses of Lemma \ref{lem:residual_lemma}, if the wave-packet envelope functions $f$ appearing in \eqref{eq:WP} evolve according to the Bistritzer-MacDonald model
\begin{equation} \label{eq:final_effective_model}
    i \de_T f = H_{\text{\emph{BM}}} f, \quad f(0) = f_0,
\end{equation}
    where $H_{\text{\emph{BM}}}$ is as in \eqref{eq:BM_H}, then, there exist constants $\gamma_0, C, c > 0$ depending only on $\oldhat{h}$, the monolayer operator \eqref{eq:graphene_block}, and the constants $\lambda_0$ and $\lambda_1$, such that, for all $\gamma \leq \gamma_0$, the residual appearing in \eqref{eq:residual_def} satisfies
    \begin{equation} \label{eq:simplified_residual_bound}
        \| r(\tau) \|_{\mathcal{H}} \leq C \gamma^{1 + c} \| f(\cdot,\gamma \tau) \|_{H^8(\mathbb{R}^2;\mathbb{C}^4)}.
    \end{equation}
\end{corollary}
Note that starting from the scaled BM model \eqref{eq:BM_H} with~\eqref{eq:final_effective_model}, it is straightforward to derive the more familiar form of the BM model \eqref{eq:BM_model_0}, as follows. First, change variables from $\vec{X}$ and $T$ back to the dimensionless, unscaled, position and time variables (recall \eqref{eq:tauprimerprime}) $\vec{r}$ and $\tau$ as
\begin{equation} \label{eq:reverse_change_of_variables}
    \vec{r} = \frac{\vec{X}}{\gamma}, \quad \tau = \frac{T}{\gamma}.
\end{equation}
Then, multiply both sides of the equation by $\gamma$, and use $\oldhat{h}(|\vec{\kappa}|;\ell) = \lambda_0 \gamma$ \eqref{eq:balance}.
\begin{remark} \label{rem:worst_case_intro}
    Note that estimates \eqref{eq:residual_bound} and \eqref{eq:simplified_residual_bound} are worst-case estimates. When $h$ is any of Examples \ref{ex:example_1}-\ref{ex:example_3}, estimates where the power $1 + c$ is replaced by $2 - c'$ for \emph{any} $c' > 0$ hold; see Remark \ref{rem:worst_case}.   
\end{remark}

The argument so far is purely formal since we have not discussed the existence and uniqueness theory for \eqref{eq:final_effective_model}, nor whether the Sobolev norms of solutions can be controlled. These issues are addressed by the following lemma.
\begin{lemma} \label{lem:BM_sols}
    Consider the general initial value problem for the BM PDE \eqref{eq:final_effective_model}, i.e., 
    \begin{equation} \label{eq:IVP}
        i \de_T f = H_{\text{\emph{BM}}} f, \quad f(0) = f_0,
    \end{equation}
    where $f_0 \in H^1(\mathbb{R}^2;\mathbb{C}^4)$. Then the following facts hold
    \begin{itemize}
        \item The IVP \eqref{eq:IVP} has a unique solution $f$ satisfying 
            \begin{equation} \label{eq:unitary_propagator}
                \| f(\cdot,T) \|_{L^2(\mathbb{R}^2;\mathbb{C}^4)} = \| f_0 \|_{L^2(\mathbb{R}^2;\mathbb{C}^4)}, \quad \forall T \geq 0.
            \end{equation}
        \item The IVP \eqref{eq:IVP} preserves regularity of $f_0$ in the sense that there exist constants $C_s > 0$ depending on $s, \lambda_0$, and $\lambda_1$, such that
            \begin{equation} \label{eq:preservation_of_regularity}
                \| f(\cdot,T) \|_{H^s(\mathbb{R}^2;\mathbb{C}^4)} \leq C_s \| f_0 \|_{H^s(\mathbb{R}^2;\mathbb{C}^4)}, \quad \forall T > 0
            \end{equation}
    for each positive integer $s$.
    \end{itemize}
\end{lemma}
\begin{proof}
    See Appendix \ref{sec:sobolev_bounds}.
\end{proof}
We emphasize that equation \eqref{eq:preservation_of_regularity} of Lemma \ref{lem:BM_sols} relies on the following fact. Indeed, it is the reason why $\theta$ must be scaled with $\gamma$ in \eqref{eq:balance} for our results to hold.
\begin{lemma} \label{lem:scaled_moire}
    In the parameter regime \eqref{eq:balance}, the derivatives of the scaled moir\'e potential \eqref{eq:BM_H} can be bounded independently of $\gamma$ and $\theta$
    \begin{equation} \label{eq:T_theta_bound}
        \sup_{1 \leq a \leq 2} \left| \pdf{X_a} \mathcal{T}^{\sigma \sigma'}(\vec{X}) \right| \leq C_T \lambda_1, \quad \forall \vec{X} \in \mathbb{R}^2, \sigma, \sigma' \in \{A,B\}.
    \end{equation}
\end{lemma}
\begin{proof}
    The result follows immediately from \eqref{eq:moire_potential_theta_bound}.
\end{proof}
We can now state and prove the main theorem of this work.
\begin{theorem} \label{thm:final_result}
    Under the hypotheses of Lemma \ref{lem:residual_lemma}, the solution $\psi$ of \eqref{eq:nondimensional_schro}, where $\psi_0$ is as in \eqref{eq:WP_0}, satisfies
    \begin{equation} \label{eq:ansatz_result}
        \begin{split}
            &\psi = \begin{pmatrix} \psi_1(\tau) \\ \psi_2(\tau) \end{pmatrix}   \\
            &( \psi_i )^\sigma_{\vec{R}_i}(\tau) = \gamma \left. f^\sigma_i( \vec{X}, T )\right|_{\vec{X} = \gamma (\vec{R}_i + \vec{\tau}^\sigma_i), T = \gamma \tau} e^{i \vec{\kappa}_i \cdot (\vec{R}_i + \vec{\tau}^\sigma_i)} + \eta_{\vec{R}_i}^\sigma(\tau),         
        \end{split}
    \end{equation}
where $i \in \{1,2\}, \sigma \in \{A,B\}$, and where the envelopes $f$ evolve according to the BM model \eqref{eq:final_effective_model}, and, there exist constants $C, \gamma_0 > 0$ such that, for all $\gamma < \gamma_0$ and $c > 0$ as in Lemma \ref{lem:residual_lemma},
     \begin{equation} \label{eq:final_result_-1}
        \| \eta(\tau) \|_{\mathcal{H}} \leq C \gamma^{1 + c} \tau.
    \end{equation}
    It follows that, for any fixed $0 \leq \nu < c$ and $C_1 > 0$, the corrector also satisfies 
    \begin{equation} \label{eq:final_result}
        \lim_{\gamma \downarrow 0} \sup_{\tau \in \left[0,C_1 \gamma^{- (1 + \nu)}\right]} \| \eta(\tau) \|_{\mathcal{H}} = 0.
    \end{equation}
\end{theorem}
\begin{proof}
    Estimate \eqref{eq:final_result_-1} follows immediately from combining Lemmas \ref{lem:approx_sols}, \ref{lem:residual_lemma}, and \ref{lem:BM_sols}, and \eqref{eq:final_result} follows immediately.
\end{proof}
\begin{remark} \label{rem:Feff_Wein}
Fefferman-Weinstein \cite{fefferman_weinstein} considered wave-packet propagation in a graphene monolayer modeled by a continuum PDE, obtaining the timescale of validity $\tau = O(\gamma^{- 2 + \nu})$ for any $\nu> 0$. Recalling Remark \ref{rem:worst_case_intro}, when $h$ is as in Examples \ref{ex:example_1}-\ref{ex:example_3}, we recover this timescale of validity.
\end{remark}
We finish this section by proving two simple Lemmas which establish that the (equivalent) BM models \eqref{eq:BM_H}, \eqref{eq:BM_model_0}, and \eqref{eq:physical_BM}, are invariant under arbitrary interlayer displacements $\vec{\mathfrak{d}}$ \eqref{eq:interlayer_displacement} (see Figure~
\ref{fig:TBG_atomic_structure}), and under translation by moir\'e lattice vectors \eqref{eq:moire_lattice}. In the first case, we will prove that, for any fixed interlayer displacement $\mathfrak{d} \neq 0$, there exists a unitary operator $U(\mathfrak{d})$ which conjugates the Hamiltonian to the same model with $\mathfrak{d} = \vec{0}$. Note that a similar calculation was already given by Bistritzer-MacDonald \cite{Bistritzer2011}.
\begin{lemma} \label{lem:interlayer_displacement}
    Let $H_{\text{\emph{BM}}}$ be as in \eqref{eq:BM_H}, where the momentum hops, moir\'e potential, and hopping matrices, are as in \eqref{eq:momentum_space_hops}, \eqref{eq:moire_potential}, and \eqref{eq:hopping_matrices_full}, respectively. Let $S_{\vec{w}} f(\vec{X}) := f(\vec{X} - \vec{w})$, and $\mathfrak{w} := \gamma \left( (\vec{b}_1 \cdot \mathfrak{d}) \vec{a}_{m,1} + (\vec{b}_2 \cdot \mathfrak{d}) \vec{a}_{m,2} \right)$. Then
    \begin{equation} \label{eq:U_def}
    U(\mathfrak{d}) H_{\text{\emph{BM}}} U(\mathfrak{d})^\dagger = \widetilde{H_{\text{\emph{BM}}}}, \quad U(\mathfrak{d}) := \diag\left( I_2 e^{- i \frac{ \mathfrak{s}_1 \cdot \mathfrak{w} }{2 \gamma}} , I_2 e^{i \frac{\mathfrak{s}_1 \cdot \mathfrak{w}}{2 \gamma}} \right) S_{\mathfrak{w}},
    \end{equation}
    where $\widetilde{H_{\text{\emph{BM}}}}$ is exactly as in \eqref{eq:BM_H}-\eqref{eq:momentum_space_hops}-\eqref{eq:moire_potential}-\eqref{eq:hopping_matrices_full}, but with $\mathfrak{d} = 0$.
\end{lemma}
\begin{proof}
    Let $\widetilde{T_n}, n \in \{1,2,3\},$ denote the hopping matrices \eqref{eq:hopping_matrices_full} with $\mathfrak{d} = 0$. Then
    \begin{equation}
        \mathcal{T}(\vec{X}) = \lambda_0 T_\theta\left(\frac{\vec{X}}{\gamma}\right) = \lambda_0 \left( \widetilde{T_1} e^{- i \frac{ \mathfrak{s}_1 }{\gamma} \cdot \vec{X}} + e^{- i \vec{b}_2 \cdot {\mathfrak{d}}} \widetilde{T_2} e^{- i \frac{\mathfrak{s}_2}{\gamma} \cdot \vec{X}} + e^{i \vec{b}_1 \cdot \mathfrak{d}} \widetilde{T_3} e^{- i \frac{\mathfrak{s}_3}{\gamma} \cdot \vec{X}} \right).
    \end{equation}
    Let $\widetilde{\mathcal{T}(\vec{X})}$ denote $\mathcal{T}(\vec{X})$ but with $\mathfrak{d} = 0$. 
    Then
    \begin{equation}
        S_{\mathfrak{w}} \mathcal{T}(\vec{X}) = e^{i \frac{\mathfrak{s}_1 \cdot \mathfrak{w}}{\gamma}} \widetilde{\mathcal{T}(\vec{X})} S_{\mathfrak{w}}.
    \end{equation}
    The result now follows. 
\end{proof}
The second Lemma, whose proof is a straightforward calculation using \eqref{eq:moire_potential_periodic}, establishes that the BM model Hamiltonian commutes with translation operators in the moir\'e lattice, as long as the translation operators are modified by appropriate phase shifts. The origin of these phase shifts is the non-zero distance between the monolayer $\vec{K}$ points, and the fact that the monolayer $\vec{K}$ points are chosen as the momentum space origins in both layers. We include the lemma for completeness, and because of its importance in allowing the spectrum of $H_{\text{BM}}$ to be expressed in terms of ``moir\'e'' Bloch bands \cite{Bistritzer2011}.
\begin{lemma} \label{lem:moire_periodicity}
    Let $H_{\text{\emph{BM}}}$ be as in \eqref{eq:BM_H}, where the momentum hops, moir\'e potential, and hopping matrices, are as in \eqref{eq:momentum_space_hops}, \eqref{eq:moire_potential}, and \eqref{eq:hopping_matrices_full}, respectively, let $S_{\vec{w}} f(\vec{X}) := f(\vec{X} - \vec{w})$, and let $\Lambda_m$ be as in \eqref{eq:moire_lattice}. Then
    \begin{equation} \label{eq:translation}
        \mathcal{S}_{\vec{R}_m} H_{\text{\emph{BM}}} = H_{\text{\emph{BM}}} \mathcal{S}_{\vec{R}_m}, \quad \mathcal{S}_{\vec{R}_m} := \diag\left( 1, 1, e^{i \frac{ \mathfrak{s}_1 \cdot \vec{R}_m }{ \gamma }}, e^{i \frac{ \mathfrak{s}_1 \cdot \vec{R}_m }{ \gamma }} \right) S_{\vec{R}_m}, \quad \vec{R}_m \in \Lambda_m.
    \end{equation}
\end{lemma}
The translation and rotation symmetries of $H_{\text{BM}}$ are also discussed in the supplementary material of \cite{Watson2021}.
\begin{remark}
    We remark on potential generalizations of our results, other than relaxing Assumption \ref{as:h_regularity}, as discussed in Remark \ref{rem:relaxing}.

    First, it shouldn't be necessary to restrict to the nearest-neighbor tight-binding model of graphene \eqref{eq:nearest_neighbor}. The Dirac points of graphene are protected by symmetry \cite{fefferman_weinstein_diracpoints,2018BerkolaikoComech}, and so it should be possible to carry through our derivation with \eqref{eq:nearest_neighbor} replaced by a general Hamiltonian with exponentially-decaying (in real space) hopping which respects $\frac{2 \pi}{3}$-rotation, inversion ($\vec{r} \mapsto - \vec{r}$), and complex conjugation symmetries.

    Second, it should be possible to push the time-scale of validity of the wave-packet approximation to higher orders in $\gamma^{-1}$ by including corrections to the {\it ansatz} \eqref{eq:WP} in the form of a formal series in powers of $\gamma$. It would be interesting to understand the effective dynamics over these longer time-scales. We expect that they will involve longer range momentum space hopping than the nearest-neighbor momentum shifts \eqref{eq:momentum_space_hops}.
\end{remark}

\section{Proof of Lemma \ref{lem:residual_lemma}} \label{sec:proof_key_lemma}

In this section, we prove Lemma \ref{lem:residual_lemma}. The structure of the proof is as follows. In Section \ref{sec:simplify_residual}, we state a sequence of lemmas, whose proofs are given in Appendices \ref{sec:proof_of_diag_lemma}, \ref{sec:proof_of_diag_theta_lemma}, and \ref{sec:proof_of_off_diag_lemma}. The lemmas write $r = - ( i \de_\tau - H )\Psi$ as a sum of (formally) ``leading order'' and ``higher order'' terms. In Section \ref{sec:parameter_regime}, we identify the parameter regime where the leading order terms balance and every higher order term is smaller in order to complete the proof of Lemma \ref{lem:residual_lemma}.

\subsection{Simplification of residual $r = - (i \de_\tau - H) \Psi$} \label{sec:simplify_residual}

We start with the action of $i \de_\tau$, which is very simple. Because of the scaling of the wave-packet envelope, the time derivative acting on the wave-packet ends up being proportional to $\gamma$.
\begin{lemma} \label{lem:left-hand_side}
\begin{equation} \label{eq:left-hand_side}
    \begin{split}
        &i \de_\tau ( {\Psi}_i )^\sigma_{\vec{R}_i}(\tau) = \\
        &\gamma \left. \left( i \gamma \de_T \right) {f}^\sigma_i \left( \gamma \left( \vec{R}_i + \vec{\tau}_i \right) , T \right) \right|_{T = \gamma \tau} e^{i \vec{\kappa}_i \cdot ( \vec{R}_i + \vec{\tau}^\sigma_i )}, \quad i \in \{1,2\}, \sigma \in \{A,B\}.
    \end{split}
\end{equation}
\end{lemma}
We now consider the action of $H$ on the wave-packet. The following lemma decomposes the action of the diagonal parts of $H$ into a leading order part, proportional to $\gamma$ and with the form of a Dirac operator, and a higher-order term $r^{\text{I}}$, which is formally $O\left(\gamma^2\right)$. This term is simply the remainder from Taylor-expanding the monolayer Bloch Hamiltonian up to linear order at the Dirac points.
\begin{lemma} \label{lem:diag_lemma}
    Let $H$ be as in \eqref{eq:block_H_diag}, and $\Psi$ be as in \eqref{eq:WP}. Then, for $i \in \{1,2\}$ and $\sigma \in \{A,B\}$,
    \begin{equation} \label{eq:diag_decomposition}
    \begin{split}
        &( H_{ii} \Psi_i )^\sigma_{\vec{R}_i}(\tau) = \\
        &\gamma \left[ \gamma \sum_{\sigma' \in \{A,B\}} \left( \vec{\sigma}_{\theta_i/2} \cdot (- i \nabla_{\vec{X}}) \right)^{\sigma \sigma'} \left. f^{\sigma'}_i(\vec{X},\gamma \tau) \right|_{\vec{X} = \gamma ( \vec{R}_i + \vec{\tau}^{\sigma'}_i )} e^{i \vec{\kappa}_i \cdot ( \vec{R}_i + \vec{\tau}_i^\sigma )} \right]  \\
        &+ ( r^{\text{\emph{I}}}_{\vec{R}_i} )^\sigma(\tau),
    \end{split}
    \end{equation}
    where $\theta_i$ is as in \eqref{eq:theta_cases}, and there exists a constant $C > 0$, depending only on the monolayer operator \eqref{eq:graphene_block}, such that
    {\color{black}
    \begin{equation} \label{eq:diag_decomposition_remainder}
        \left\| r^{\text{\emph{I}}}_i(\tau) \right\|_{\ell^2(\mathbb{Z}^2;\mathbb{C}^2)} \leq C \gamma^2 \left\| f_i(\cdot,\gamma \tau) \right\|_{H^2\left(\field{R}^2;\mathbb{C}^2\right)} + C \gamma^4 \left\| f_i(\cdot,\gamma \tau) \right\|_{H^6\left(\field{R}^2;\mathbb{C}^2\right)}.
    \end{equation}
    }
\end{lemma}
\begin{proof}
    See Appendix \ref{sec:proof_of_diag_lemma}.
\end{proof}
{\color{black}
\begin{remark} \label{rem:higher_order_terms}
    The higher-order term appearing in \eqref{eq:diag_decomposition_remainder} arises from estimating the $\tilde{\vec{G}}_1 \neq \vec{0}$ terms in \eqref{eq:all_terms} by \eqref{eq:trick_estimate}. Note that similar terms appear in \eqref{eq:diag_theta_decomposition_remainder}, \eqref{eq:r2_bound}, and \eqref{eq:r3_bound}, with very similar origins. Under stronger regularity assumptions, these terms can be made arbitrarily high order in $\gamma$ (for example, the estimate in \eqref{eq:diag_decomposition_remainder} can be improved to $\gamma^{4+s}\|f_i(\cdot,\gamma \tau)\|_{H^{6+2 s}(\mathbb{R}^2;\mathbb{C}^2)}$ for any $s \geq 0$). 
\end{remark}
}
The following lemma shows that the diagonal part of $H$ can be further decomposed into a $\theta$-independent part proportional to $\gamma$, and a $\theta$-dependent part which is $O(\gamma \theta)$. 
\begin{lemma} \label{lem:diag_theta_lemma}
    Let $H$ be as in \eqref{eq:block_H_diag}, and $\Psi$ be as in \eqref{eq:WP}. Then, for $i \in \{1,2\}$ and $\sigma \in \{A,B\}$,
    \begin{equation} \label{eq:diag_theta_decomposition}
    \begin{split}
        &( H_{ii} \Psi_i )^\sigma_{\vec{R}_i}(\tau) = \\
        &\gamma \left[ \gamma \sum_{\sigma' \in \{A,B\}} \left( \vec{\sigma} \cdot (- i \nabla_{\vec{X}}) \right)^{\sigma \sigma'} \left. f^{\sigma'}_i(\vec{X},\gamma \tau) \right|_{\vec{X} = \gamma ( \vec{R}_i + \vec{\tau}^{\sigma'}_i )} e^{i \vec{\kappa}_i \cdot ( \vec{R}_i + \vec{\tau}_i^\sigma )} \right] \\
        &+ ( r^{\text{\emph{I}}}_{\vec{R}_i} )^\sigma(\tau) + ( r^{\text{\emph{II}}}_{\vec{R}_i} )^\sigma(\tau),
    \end{split}
    \end{equation}
    where $r^{\text{\emph{I}}}$ is as in Lemma \ref{lem:diag_lemma}, and there exists a constant $C > 0$ such that
    {\color{black}
    \begin{equation} \label{eq:diag_theta_decomposition_remainder}
        \left\| r^{\text{\emph{II}}}_{i}(\tau) \right\|_{\ell^2(\mathbb{Z}^2;\mathbb{C}^2)} \leq C \theta \gamma \left\| f_i(\cdot,\gamma \tau) \right\|_{H^1\left(\field{R}^2;\mathbb{C}^2\right)} + C \theta \gamma^3 \left\| f_i(\cdot,\gamma \tau) \right\|_{H^5\left(\field{R}^2;\mathbb{C}^2\right)}.
    \end{equation}
    }
\end{lemma}
\begin{proof}
    See Appendix \ref{sec:proof_of_diag_theta_lemma}.
\end{proof}

We now consider the action of the off-diagonal terms of $H$. The following lemma decomposes this action into a leading order part and higher order terms $r^{\text{III}}$ and $r^{\text{IV}}$. The leading order part is proportional to $\oldhat{h}(|\vec{\kappa}|;\ell)$ and describes interlayer coupling through the moir\'e potential. The higher order terms $r^{\text{III}}$ and $r^{\text{IV}}$ come from Taylor-expanding the interlayer hopping function about the wave-packet center and neglecting momentum space coupling terms other than three dominant terms, respectively. Note that we consider $H_{12} \Psi_2$ without loss of generality since $H_{21} \Psi_1$ could be simplified by exactly the same steps. 
\begin{lemma} \label{lem:off_diag_lemma}
    Let $H$ be as in \eqref{eq:block_H_diag}, and $\Psi$ be as in \eqref{eq:WP}. Recall the hopping matrices $T_j$ \eqref{eq:hopping_matrices}, and the momentum hops $\mathfrak{s}_j \eqref{eq:momentum_space_hops}, 1 \leq j \leq 3$. Then, for $\sigma \in \{A,B\}$,
    \begin{equation} \label{eq:off_diag_decomp}
    \begin{split}
        &\left( H_{12} \Psi_2 \right)_{\vec{R}_1}^\sigma(\tau) = \\ 
        &\oldhat{h}(|\vec{\kappa}|;\ell) \left[ \gamma \sum_{j = 1}^3 \sum_{\sigma' \in \{A,B\}} T^{\sigma \sigma'}_j \left. e^{- i \frac{\mathfrak{s}_j \cdot \vec{X}}{\gamma}} f^{\sigma'}_2\left(\vec{X},\gamma \tau\right) \right|_{\vec{X} = \gamma ( \vec{R}_1 + \vec{\tau}_1^\sigma)} e^{i \vec{\kappa}_1 \cdot ( \vec{R}_1 + \vec{\tau}^\sigma_1 )} \right] \\ 
        &+ ( r^{\text{\emph{III}}}_{\vec{R}_1} )^\sigma(\tau) + ( r^{\text{\emph{IV}}}_{\vec{R}_1} )^\sigma(\tau),
    \end{split}
\end{equation}
    where, for any $\ell_1 > 0$ and $M_0 > 1$, there exist constants $C_1 > 0$ and $C_2 > 0$ depending only on $\ell_1, M_0,$ and $\oldhat{h}$ such that for all $\ell > \ell_1$ and $M > M_0$,
    {\color{black}
    \begin{equation} \label{eq:r2_bound}
    \begin{split}
        &\left\| {r}^\text{\emph{III}}_1(\tau) \right\|_{\ell^2(\mathbb{Z}^2;\mathbb{C}^2)} \leq C_1 \gamma \oldhat{h}(|\vec{\kappa}|;\ell)  \\
        &\times \left( e^{\left(D_1 - D_2 \left( 1 - \frac{1}{M} \right) \right) \ell \left| \vec{\kappa} \right|} \left\| f_2(\cdot,\gamma \tau) \right\|_{H^1(\field{R}^2;\mathbb{C}^2)} + M^3 \gamma^3 e^{D_1 \ell |\vec{\kappa}|} \left\| f_2(\cdot,\gamma \tau) \right\|_{H^4(\field{R}^2;\mathbb{C}^2)} \right.  \\
        &\left. \vphantom{e^{\left(D_1 - D_2 \left( 1 - \frac{1}{M} \right) \right) \ell \left| \vec{\kappa} \right|}} \quad \quad \quad \quad \quad \quad + \gamma^2 e^{\left( D_1 - D_2 \left(1 - \frac{1}{M}\right) \right) \ell |\vec{\kappa}|} \left\| f_2(\cdot,\gamma \tau) \right\|_{H^5(\mathbb{R}^2;\mathbb{C}^2)} + M^3 \gamma^5 e^{D_1 \ell |\vec{\kappa}|} \left\| f_2(\cdot,\gamma \tau) \right\|_{H^8(\mathbb{R}^2;\mathbb{C}^2)} \right).
    \end{split}
    \end{equation}
    }
    {\color{black}
    \begin{equation} \label{eq:r3_bound}
    \begin{split}
        &\left\| {r}^{\text{\emph{IV}}}_1(\tau) \right\|_{\ell^2(\mathbb{Z}^2;\mathbb{C}^2)} \leq   \\
        &C_2 \oldhat{h}(|\vec{\kappa}|;\ell) e^{[ D_1 - 2 D_2] \ell |\vec{\kappa}|} \left( \left\| f_2(\cdot,\gamma \tau) \right\|_{L^2(\mathbb{R}^2;\mathbb{C}^2)} + \gamma^2 \left\| f_2(\cdot,\gamma \tau) \right\|_{H^4(\mathbb{R}^2;\mathbb{C}^2)} \right).
    \end{split}
\end{equation}
    }
\end{lemma}
\begin{proof}
    See Appendix \ref{sec:proof_of_off_diag_lemma}.
\end{proof}

\subsection{Identification of parameter regime where the BM model is dominant effective dynamics} \label{sec:parameter_regime}

In this section, we complete the proof of Lemma \ref{lem:residual_lemma}. We start by identifying the regime where the leading-order terms in \eqref{eq:left-hand_side}, \eqref{eq:diag_decomposition}, and \eqref{eq:off_diag_decomp} balance. This occurs when
\begin{equation} \label{eq:balance_00}
    \oldhat{h}\left(|\vec{\kappa}|;\ell\right) = \lambda_0 \gamma,
\end{equation}
for some fixed, positive constant $\lambda_0 > 0$. When \eqref{eq:balance_00} holds, the leading-order terms in \eqref{eq:left-hand_side}, \eqref{eq:diag_decomposition}, and \eqref{eq:off_diag_decomp} are all proportional to $\gamma$, and the residual can be written as
\begin{equation} \label{eq:res_decomp}
    \begin{split}
        r_{\vec{R}_i}^\sigma(\tau) &= - [ (i \de_\tau - H) \Psi ]^\sigma_{\vec{R}_i}(\tau)   \\
        &= \gamma \left[ \left. \gamma \left( - \left[i \de_T - H_{\text{BM}}^{0}\right] f \right)^\sigma_i(\vec{X},T) \right|_{\vec{X} = \gamma(\vec{R}_i + \vec{\tau}_i^\sigma),T = \gamma \tau} e^{i \vec{\kappa}_i \cdot ( \vec{R}_i + \vec{\tau}_i^\sigma )} \right] \\
        &\quad\quad\quad + (r^{\text{I}})^\sigma_{\vec{R}_i}(\tau) + (r^{\text{III}})^\sigma_{\vec{R}_i}(\tau) + (r^{\text{IV}})^\sigma_{\vec{R}_i}(\tau),
    \end{split}
\end{equation}
where
\begin{equation} \label{eq:HBM_theta}
    H^{0}_{\text{BM}} := \begin{pmatrix} \vec{\sigma}_{-\theta/2} \cdot (- i \nabla_{\vec{X}}) & \lambda_0 T_\theta\left(\frac{\vec{X}}{\gamma}\right) \\ \lambda_0 T^\dagger_\theta\left(\frac{\vec{X}}{\gamma}\right) & \vec{\sigma}_{\theta/2} \cdot (- i \nabla_{\vec{X}}) \end{pmatrix},
\end{equation}
 where we write $A^\dagger$ for the conjugate transpose of a matrix $A$, and $r^{\text{I}}, r^{\text{III}}$, and $r^{\text{IV}}$ are as in \eqref{eq:diag_decomposition_remainder} and \eqref{eq:off_diag_decomp}, respectively. Assuming for a moment that all Sobolev norms of the envelope functions $f_i^\sigma$ can be bounded independently of all parameters and $\tau$, the orders of these terms are
\begin{equation} \label{eq:pre_balance}
    \begin{split}
        &\| r^{\text{I}}(\tau) \|_\mathcal{H} \sim \gamma^2, \quad \| r^{\text{III}}(\tau) \|_{\mathcal{H}} \sim \gamma^2 e^{D_1 \ell |\vec{\kappa}|} \left( e^{ \left[ - D_2 \left( 1 - \frac{1}{M} \right) \right] \ell |\vec{\kappa}| } + \gamma^2 M^2 \right), \\
        &\| r^{\text{IV}}(\tau) \|_{\mathcal{H}} \sim \gamma e^{ \left[ D_1 - 2 D_2 \right] \ell |\vec{\kappa}| }.
    \end{split}
\end{equation}
We will choose $M > 0$ so that these errors balance. First, note that \eqref{eq:D_def}, together with \eqref{eq:balance_00}, implies that
\begin{equation} \label{eq:log_growth}
    \frac{ - \ln \gamma }{ D_1 |\vec{\kappa}| } \leq \ell \leq \frac{ - \ln \gamma }{ D_2 |\vec{\kappa}| }.
\end{equation}
It then follows easily that
\begin{equation}
    \| r^{\text{IV}}(\tau) \|_{\mathcal{H}} \lesssim \gamma^{1 + 2 \frac{D_2}{D_1} - \frac{D_1}{D_2}},
\end{equation}
and that
\begin{equation}
    \| r^{\text{III}}(\tau) \|_{\mathcal{H}} \lesssim \gamma^{2 - \frac{D_1}{D_2}} \left( \gamma^{\left(1 - \frac{1}{M}\right) \frac{D_2}{D_1}} + \gamma^2 M^2 \right).
\end{equation}
We can roughly balance these errors by taking
\begin{equation}
    M = \frac{1}{\gamma^{\frac{1}{2}}} \implies \| r^{\text{III}}(\tau) \|_{\mathcal{H}} \lesssim \gamma^{2 - \frac{D_1}{D_2}} \left( \gamma^{\left(1 - \gamma^{1/2}\right) \frac{D_2}{D_1}} + \gamma \right).
\end{equation}
Using \eqref{eq:D1D2}, we then have existence of constants $C > 0$ and $c > 0$ such that
\begin{equation} \label{eq:res_bound_1}
    \left\| \sum_{i \in \{\text{I},\text{III},\text{IV}\}} r^i(\tau) \right\|_\mathcal{H} \leq C \gamma^{1 + c} \| f(\cdot,\gamma \tau) \|_{H^8(\mathbb{R}^2;\mathbb{C}^4)},
\end{equation}
for all $\gamma$ sufficiently small.

The estimate \eqref{eq:res_bound_1}, together with \eqref{eq:res_decomp}, brings us close to the result of Lemma \ref{lem:residual_lemma}. However, the result, so far, is not sufficient for our purposes. Specifically, it is easy to see that, without further assumptions, part \eqref{eq:preservation_of_regularity} of Lemma \ref{lem:BM_sols} will not hold uniformly in $\gamma$ when $H_{\text{BM}}$ is replaced by $H_{\text{BM}}^0$ (as in \eqref{eq:HBM_theta}). The problem is that, for fixed $\theta$ and $\gamma \rightarrow 0$, the derivatives of the 
scaled moir\'e potential grow without bound (cf. \eqref{eq:moire_potential_theta_bound})
\begin{equation} \label{eq:bad_scaled_moire}
    \left| \pdf{X_a} T^{\sigma \sigma'}_\theta \left( \frac{\vec{X}}{\gamma} \right) \right| \leq C \frac{\theta}{\gamma}.
\end{equation}
To avoid this, we assume that
\begin{equation}
    \theta \leq \lambda_1 \gamma,
\end{equation}
introducing another positive constant $\lambda_1 > 0$.
But now, we can use Lemma \ref{lem:diag_theta_lemma} to drop the $\theta$-dependence of the intralayer terms in $H_{\text{BM}}^0$ \eqref{eq:HBM_theta}, so that we can now write the residual as 
\begin{equation} \label{eq:res_decomp_2}
    \begin{split}
        &r_{\vec{R}_i}^\sigma(\tau) =   \\ 
        &\gamma \left[ \left. \gamma \left( - \left[i \de_T - H_{\text{BM}} \right] f \right)^\sigma_i(\vec{X},T) \right|_{\vec{X} = \gamma(\vec{R}_i + \vec{\tau}_i^\sigma),T = \gamma \tau} e^{i \vec{\kappa}_i \cdot ( \vec{R}_i + \vec{\tau}_i^\sigma )} \right] + \mathfrak{r}^\sigma_{\vec{R}_i}(\tau),
    \end{split}
\end{equation}
where $H_{\text{BM}}$ is as in \eqref{eq:BM_H}, and
\begin{equation}
    \mathfrak{r}^\sigma_{\vec{R}_i}(\tau) := (r^{\text{I}})^\sigma_{\vec{R}_i}(\tau) + (r^{\text{II}})^\sigma_{\vec{R}_i}(\tau) + (r^{\text{III}})^\sigma_{\vec{R}_i}(\tau) + (r^{\text{IV}})^\sigma_{\vec{R}_i}(\tau).
\end{equation}
Combining \eqref{eq:res_bound_1} with \eqref{eq:diag_theta_decomposition_remainder}, we have existence of constants $C > 0$ and $c > 0$ such that
\begin{equation} \label{eq:final_res_bound}
    \| \mathfrak{r}^\sigma_i(\tau) \|_{\mathcal{H}} \leq C \gamma^{1 + c}
\end{equation}
for all sufficiently small $\gamma$, from which \eqref{eq:residual_bound} follows immediately, using \eqref{eq:D1D2}.

\begin{remark} \label{rem:worst_case}
It is worth noting that \eqref{eq:res_bound_1} and \eqref{eq:final_res_bound} are \emph{worst-case} estimates. In fact, for each of Examples \ref{ex:example_1}-\ref{ex:example_3}, the constants $D_1, D_2$ can be taken such that $D_1 = 1 + \epsilon$ and $D_2 = 1 - \epsilon$ for any $0 < \epsilon$. In such cases we have the significantly better bounds
\begin{equation} \label{eq:res_bound_1_better}
    \left\| \sum_{i \in \{\text{\emph{I}},\text{\emph{III}},\text{\emph{IV}}\}} r^i(\tau) \right\|_\mathcal{H} \leq C \gamma^{2 - c'} \| f(\cdot,\gamma \tau) \|_{H^8(\mathbb{R}^2;\mathbb{C}^4)},
\end{equation}
and
\begin{equation} \label{eq:final_res_bound_better}
    \| \mathfrak{r}^\sigma_i(\tau) \|_{\mathcal{H}} \leq C \gamma^{2 - c'}
\end{equation}
for any $c' > 0$. 
\end{remark}

\section{Bistritzer-MacDonald's formal derivation of the first magic angle} 
\label{sec:magic_angles}

In this section, we review Bistritzer-MacDonald's formal, analytical, derivation of the first ``magic angle'' of TBG \cite{Bistritzer2011}. We start by recalling the Fourier transform pair 
\begin{equation}
    \begin{split}
        &\oldhat{f}(\vec{\xi}) = \left( \mathcal{F} f \right)(\vec{\xi}) := \inty{\mathbb{R}^2}{}{ e^{- i \vec{\xi} \cdot \vec{r}} f(\vec{r}) }{\vec{r}} \\
        &f(\vec{r}) = \left( \mathcal{F}^{-1} \oldhat{f} \right)(\vec{r}) := \frac{1}{(2 \pi)^2} \inty{\mathbb{R}^2}{}{ e^{i \vec{\xi} \cdot \vec{r}} \oldhat{f}(\vec{\xi}) }{\vec{\xi}},
    \end{split}
\end{equation}
acting componentwise on vector functions. Recall the BM Hamiltonian, in physical units,
\begin{equation} \label{eq:physical_BM_2}
    \mathfrak{H}_{\text{BM}} := \begin{pmatrix} v \vec{\sigma} \cdot (- i \nabla) & w \sum_{n = 1}^3 T_n e^{- i \vec{s}_n \cdot \vec{r}} \\ w \sum_{n = 1}^3 T_n^\dagger e^{i \vec{s}_n \cdot \vec{r}} & v \vec{\sigma}\cdot (- i \nabla) \end{pmatrix},
\end{equation}
where $v := \hbar v_D$, and $w := \frac{\oldhat{\mathfrak{h}}(|\vec{K}|;L)}{|\Gamma|}$. We can write $\mathfrak{H}_{\text{BM}}$ in the Fourier domain, as
\begin{equation}
    \oldhat{\mathfrak{H}}_{\text{BM}} := \mathcal{F} \mathfrak{H}_{\text{BM}} \mathcal{F}^{-1} = \begin{pmatrix} v \vec{\sigma} \cdot \vec{\xi} & w \sum_{j = 1}^3 T_j S_{- \vec{s}_j} \\ w \sum_{j = 1}^3 T_j^\dagger S_{\vec{s}_j} & v \vec{\sigma} \cdot \vec{\xi} \end{pmatrix},
\end{equation}
where $S_{\vec{w}} \oldhat{f}(\vec{\xi}) := \oldhat{f}(\vec{\xi} - \vec{w})$ denotes the shift operator. The idea is to consider a drastically truncated version of \eqref{eq:final_effective_model}, defined by
\begin{equation} \label{eq:truncated_H_BM}
    \oldhat{\mathfrak{H}}_{\text{BM,trunc}} := \begin{pmatrix} v \vec{\sigma} \cdot \vec{\xi} & w T_1 & w T_2 & w T_3 \\
        w T_1^\dagger & v \vec{\sigma} \cdot \left(\vec{\xi} - \vec{s}_1\right)  & 0 & 0 \\
        w T_2^\dagger & 0 & v \vec{\sigma} \cdot \left(\vec{\xi} - \vec{s}_2\right) & 0 \\ 
        w T_3^\dagger & 0 & 0 & v \vec{\sigma} \cdot \left(\vec{\xi} - \vec{s}_3\right) \end{pmatrix}.
\end{equation}
The operator \eqref{eq:truncated_H_BM} can be understood as $\oldhat{\mathfrak{H}}_{\text{BM}}$ projected onto a basis consisting of modes with momentum $\vec{\xi}$ on layer $1$, together with modes at momenta $\vec{\xi} - \vec{s}_j, j \in \{1,2,3\}$, on layer $2$. 

We start by considering the limit without interlayer coupling, i.e., $w = 0$. In this case, we have
\begin{equation} \label{eq:truncated_H_BM_2}
    \oldhat{\mathfrak{H}}_{\text{BM,trunc}} = \diag \left( v \vec{\sigma} \cdot \vec{\xi}, v \vec{\sigma} \cdot \left(\vec{\xi} - \vec{s}_1\right), v \vec{\sigma} \cdot \left(\vec{\xi} - \vec{s}_2 \right), v \vec{\sigma} \cdot \left(\vec{\xi} - \vec{s}_3\right) \right),
\end{equation}
so that the spectrum as a function of $\vec{\xi}$ is the union of cones centered at $\vec{s}_j$, $j \in \{1,2,3\}$, i.e.,
\begin{equation}
    \text{Spec } \oldhat{\mathfrak{H}}_{\text{BM,trunc}} = \left\{ \pm |\vec{\xi}| \right\} \bigcup_{j = 1}^3 \left\{ \pm \left|\vec{\xi} - \vec{s}_j \right| \right\}.
\end{equation}
In particular, $0$ is exactly two-fold degenerate when $\vec{\xi} = \vec{0}$.

We now turn on interlayer coupling, setting $w > 0$. We will first show that the two-fold degeneracy of the eigenvalue $0$ at $\vec{\xi} = \vec{0}$ is unbroken by the perturbation by explicitly constructing the $0$ eigenspace. Let $\Phi := \left( \Phi_0, \Phi_1, \Phi_2, \Phi_3 \right)^\top$. The equation 
\begin{equation} \label{eq:Phi_eq}
    \oldhat{\mathfrak{H}}_{\text{BM,trunc}} \Phi = 0
\end{equation}
is a system of four equations. Three of these equations have the solutions  
\begin{equation} \label{eq:Phi_j}
    \Phi_j = w \left( v \vec{\sigma} \cdot \vec{s}_j \right)^{-1} T_j^\dagger \Phi_0, \quad j \in \{1,2,3\}.
\end{equation}
It is straightforward to check that
\begin{equation}
    \left( v \vec{\sigma} \cdot \vec{s}_j \right)^{-1} = \frac{1}{ v |\Delta \vec{K}|^2 } \vec{\sigma} \cdot \vec{s}_j.
\end{equation}
Remarkably, substituting \eqref{eq:Phi_j} into the fourth equation of \eqref{eq:Phi_eq} yields a trivial equation for $\Phi_0$, because direct calculation shows that
\begin{equation}
    \sum_{j = 1}^3 T_j \vec{\sigma} \cdot \vec{s}_j T_j^\dagger = 0. 
\end{equation}
We conclude that $0$ is a two-fold degenerate eigenvalue of $\oldhat{\mathfrak{H}}_{\text{BM,trunc}}$. 

We now study how the two-fold degenerate $0$-eigenspace of $\oldhat{\mathfrak{H}}_{\text{BM,trunc}}$ is perturbed by turning on non-zero $\vec{\xi}$. We start by fixing an (unnormalized) basis of this degenerate eigenspace as
\begin{equation}
    \Phi^i := \left( \vec{e}_i , \frac{w}{v |\Delta \vec{K}|^2 } \left( \vec{\sigma} \cdot \vec{s}_1 \right) T_1^\dagger \vec{e}_i, \frac{w}{ v |\Delta \vec{K}|^2 } \left( \vec{\sigma} \cdot \vec{s}_2 \right) T_2^\dagger \vec{e}_i, \frac{w}{ v |\Delta \vec{K}|^2 } \left( \vec{\sigma} \cdot \vec{s}_3 \right) T_3^\dagger \vec{e}_i \right)^\top,
\end{equation}
for $i \in \{1,2\}$. Direct calculation shows that
\begin{equation} \label{eq:our_alpha}
    \left| \Phi^i \right|^2 = 1 + 6 \alpha^2, \quad i \in \{1,2\}, \quad \alpha := \frac{ w }{ v |\Delta \vec{K}| },
\end{equation}
so that a normalized basis of the degenerate eigenspace is given by $\Phi^{i'} := \frac{\Phi^i}{\sqrt{1 + 6 \alpha^2}}$, $i \in \{1,2\}$. We now calculate the projection of $\oldhat{\mathfrak{H}}_{\text{BM,trunc}}$ onto the degenerate eigenspace as
\begin{equation} \label{eq:effective_H}
    \begin{split}
        &\begin{pmatrix} \ip{ \Phi^{1'} }{ \oldhat{\mathfrak{H}}_{\text{BM,trunc}} \Phi^{1'} } & \ip{ \Phi^{1'} }{ \oldhat{\mathfrak{H}}_{\text{BM,trunc}} \Phi^{2'} }   \\
        \ip{ \Phi^{2'} }{ \oldhat{\mathfrak{H}}_{\text{BM,trunc}} \Phi^{1'} } & \ip{ \Phi^{2'} }{ \oldhat{\mathfrak{H}}_{\text{BM,trunc}} \Phi^{2'} } \end{pmatrix} \\
        &= \frac{v}{1 + 6 \alpha^2} \left( \vec{\sigma} \cdot \vec{\xi} + \frac{ \alpha^2 }{ |\Delta \vec{K}|^2 } \sum_{j = 1}^3 T_j ( \vec{\sigma} \cdot \vec{s}_j ) ( \vec{\sigma} \cdot \vec{\xi} ) ( \vec{\sigma} \cdot \vec{s}_j ) T_j^\dagger \right).
    \end{split}
\end{equation}
Further direct calculation shows that \eqref{eq:effective_H} can be simplified to the form of the monolayer Dirac Hamiltonian \eqref{eq:monolayer_Dirac}, with a re-scaled Fermi velocity
\begin{equation}
    = \rho v \vec{\sigma} \cdot \vec{\xi}, \quad \rho := \frac{1 - 3 \alpha^2}{1 + 6 \alpha^2}.
\end{equation}
Bistritzer-MacDonald called $\theta$ such that $\rho$ vanishes a ``magic angle.'' Using the physical values \cite{Bistritzer2011}
\begin{equation}
    w \approx 0.11 \text{ eV}, \quad \hbar v_D \approx 6.6 \text{ eV\AA}, \quad a \approx 2.5 \text{ \AA},
\end{equation}
and the approximation
\begin{equation}
    | \Delta \vec{K} | = 2 |\vec{K}| \sin\left(\frac{\theta}{2}\right) \approx |\vec{K}| \theta = \frac{4 \pi}{3 a} \theta,
\end{equation}
we have that $\theta$ is a magic angle if
\begin{equation} \label{eq:theta_val}
    \alpha = \frac{1}{\sqrt{3}} \implies \theta = \frac{ 3 \sqrt{3} (0.11) (2.5) }{ 4 \pi (6.6) } \approx .98^\circ.
\end{equation}
Despite the crudeness of the truncation \eqref{eq:truncated_H_BM}, numerical calculations find that the moir\'e bands of \eqref{eq:BM_model_0} are indeed flat at the moir\'e $\vec{K}$ point at $\theta \approx 1.05^\circ$ \cite{Bistritzer2011}, remarkably close to the value \eqref{eq:theta_val}.

\printbibliography

\appendix

\section{Derivation of equation \eqref{eq:off_diagonal_Bloch}} \label{sec:off_diagonal_Bloch}

By definition,
\begin{equation} \label{eq:full_off_diag}
    \begin{split}
        &\left( \mathcal{G}_1 H_{12} \mathcal{G}_2^{-1}\tilde{\psi_2} \right)^\sigma(\vec{k}_1) = \\
        &\frac{1}{|\Gamma^*|}\sum_{\vec{R}_1 \in \Lambda_1} e^{- i \vec{k}_1 \cdot ( \vec{R}_1 + \vec{\tau}_1^\sigma )} \sum_{\vec{R}_2 \in \Lambda_2} \sum_{\sigma' \in \{A,B\}} \mathfrak{h}\left( \vec{R}_{12} + \vec{\tau}_{12}^{\sigma \sigma'} ; L \right) \inty{\Gamma_2^*}{}{ e^{i \vec{k}_2 \cdot (\vec{R}_2 + \vec{\tau}_2^{\sigma'})} \tilde{\psi}^{\sigma'}_2(\vec{k}_2) }{\vec{k}_2}.
    \end{split}
\end{equation}
Writing $\mathfrak{h}$ as an inverse Fourier transform, and exchanging the order of summation and integration, expression \eqref{eq:full_off_diag} simplifies to
\begin{equation}
    \frac{1}{(2 \pi)^2 |\Gamma^*|} \sum_{\sigma' \in \{A,B\}} \inty{\Gamma_2^*}{}{ \inty{\field{R}^2}{}{ \sum_{\vec{R}_1 \in \Lambda_1} \sum_{\vec{R}_2 \in \Lambda_2} e^{- i \vec{k}_1 \cdot ( \vec{R}_1 + \vec{\tau}_1^\sigma )} e^{i \vec{\xi} \cdot ( \vec{R}_{12} + \vec{\tau}_{12}^{\sigma \sigma'} )} e^{i \vec{k}_2 \cdot (\vec{R}_2 + \vec{\tau}_2^{\sigma'})} \oldhat{\mathfrak{h}}(\vec{\xi}; L) }{\vec{\xi}} \tilde{\psi}^{\sigma'}_2(\vec{k}_2) }{\vec{k}_2}.
\end{equation}
Using the Poisson summation formula
\begin{equation}
    \sum_{\vec{R}_i \in \Lambda_i} e^{i \vec{k}_i \cdot \vec{R}_i} = |\Gamma_i^*| \sum_{\vec{G}_i \in \Lambda_i^*} \delta(\vec{k}_i - \vec{G}_i), \quad i \in \{1,2\},
\end{equation}
we have
\begin{equation} \label{eq:partial_off_diag}
    \begin{split}
        &\frac{1}{|\Gamma^*|}\inty{\field{R}^2}{}{ \sum_{\vec{R}_1 \in \Lambda_1} \sum_{\vec{R}_2 \in \Lambda_2} e^{- i \vec{k}_1 \cdot ( \vec{R}_1 + \vec{\tau}_1^\sigma )} e^{i \vec{\xi} \cdot ( \vec{R}_{12} + \vec{\tau}_{12}^{\sigma \sigma'} )} e^{i \vec{k}_2 \cdot (\vec{R}_2 + \vec{\tau}_2^{\sigma'})} \oldhat{\mathfrak{h}}(\vec{\xi};L) }{\vec{\xi}} \\
        &= |\Gamma^*| \inty{\field{R}^2}{}{ \sum_{\vec{G}_1 \in \Lambda_1^*} \sum_{\vec{G}_2 \in \Lambda_2^*} e^{i (\vec{\xi} - \vec{k}_1) \cdot \vec{\tau}_1^\sigma} e^{- i (\vec{\xi} - \vec{k}_2) \cdot \vec{\tau}_2^{\sigma'}} \delta( \vec{\xi} - \vec{k}_1 - \vec{G}_1 ) \delta( \vec{\xi} - \vec{k}_2 - \vec{G}_2 ) \oldhat{\mathfrak{h}}(\vec{\xi};L) }{\vec{\xi}}.
    \end{split}
\end{equation}
Exchanging the order of summation and integration again and carrying out the integration, expression \eqref{eq:partial_off_diag} simplifies to
\begin{equation}
    |\Gamma^*| \sum_{\vec{G}_2 \in \Lambda_1^*} \sum_{\vec{G}_2 \in \Lambda_2^*} e^{i \left[ \vec{G}_1 \cdot \vec{\tau}_1^\sigma - \vec{G}_2 \cdot \vec{\tau}_2^{\sigma'}\right]} \oldhat{\mathfrak{h}}(\vec{k}_1 + \vec{G}_1;L) \delta(\vec{k}_1 + \vec{G}_1 - \vec{k}_2 - \vec{G}_2),
\end{equation}
from which \eqref{eq:off_diagonal_Bloch} immediately follows, using $\frac{|\Gamma^*|}{(2 \pi)^2} = \frac{1}{|\Gamma|}$.



\section{Proof of Lemma \ref{lem:diag_lemma}} \label{sec:proof_of_diag_lemma}

We organize the proof of Lemma \ref{lem:diag_lemma} into sections. First, we will review formulas for the wave-packet {\it ansatz} \eqref{eq:WP} in Section \ref{sec:WP_momentum}. Then, we will separate the momentum space expression for $H_{11} \Psi_1$ into ``leading'' and ``higher-order'' terms in Section \ref{sec:separate_leading_higherorder}. We will then estimate the higher-order terms in Section \ref{sec:estimate_residual} and transform the leading terms back to real space in Section \ref{sec:transform_back}.

\subsection{wave-packet {\it ansatz} in momentum space} \label{sec:WP_momentum}

In this section we record formulas for the wave-packet {\it ansatz} in momentum space, which will be important in the remainder of this appendix as well as in Appendices \ref{sec:proof_of_diag_theta_lemma} and \ref{sec:proof_of_off_diag_lemma}
\begin{equation}
    \tilde{\psi}_0 = \begin{pmatrix} \tilde{\psi}_{1,0} \\ \tilde{\psi}_{2,0} \end{pmatrix}, \ ( \tilde{\psi}_{i,0} )^\sigma(\vec{k}_i) = \frac{\gamma^{-1} |\Gamma^*|^{1/2}}{(2 \pi)^2} \sum_{\vec{G}_i \in \Lambda^*_i} \left. e^{- i \vec{\tau}_i^\sigma \cdot \vec{G}_i} \oldhat{f}^\sigma_{i,0}\left( \vec{\mathcal{K}}_i \right) \right|_{\vec{\mathcal{K}}_i = \frac{\vec{k}_i - \vec{\kappa}_i - \vec{G}_i}{\gamma} }, \ i \in \{1,2\}.
\end{equation}
\begin{equation} \label{eq:WP_momentum}
    \begin{split}
        &\tilde{\Psi}(\tau) = \begin{pmatrix} \tilde{\Psi}_1(\tau) \\ \tilde{\Psi}_2(\tau) \end{pmatrix}, \\
        &( \tilde{\Psi}_i )^\sigma(\vec{k}_i,\tau) = \frac{\gamma^{-1} |\Gamma^*|^{1/2}}{(2 \pi)^2} \sum_{\vec{G}_i \in \Lambda^*_i} \left. e^{- i \vec{\tau}_i^\sigma \cdot \vec{G}_i} \oldhat{f}^\sigma_i\left( \vec{\mathcal{K}}_i, T \right) \right|_{\vec{\mathcal{K}}_i = \frac{\vec{k}_i - \vec{\kappa}_i - \vec{G}_i}{\gamma}, T = \gamma \tau}, \quad i \in \{1,2\}.
    \end{split}
\end{equation}

With respect to $\vec{q}_i = \vec{k}_i - \vec{\kappa}_i$, the initial condition and {\it ansatz} take the simpler forms
\begin{equation} \label{eq:initial_data_k}
    ( \tilde{\psi}_{i,0} )^\sigma(\vec{\kappa}_i + \vec{q}_i) = \frac{\gamma^{-1} |\Gamma^*|^{1/2}}{(2 \pi)^2}  \sum_{\vec{G}_i \in \Lambda^*_i} e^{- i \vec{\tau}_i^\sigma \cdot \vec{G}_i} \left. \oldhat{f}^\sigma_{i,0}( \vec{\mathcal{Q}}_i ) \right|_{\vec{\mathcal{Q}}_i = \frac{ \vec{q}_i - \vec{G}_i }{ \gamma }} , \quad i \in \{1,2\},
\end{equation}
\begin{equation} \label{eq:ansatz_k}
    ( \tilde{\Psi}_i )^\sigma(\vec{\kappa}_i + \vec{q}_i,\tau) = \frac{\gamma^{-1} |\Gamma^*|^{1/2}}{(2 \pi)^2} \sum_{\vec{G}_i \in \Lambda^*_i} e^{- i \vec{\tau}_i^\sigma \cdot \vec{G}_i} \left. \oldhat{f}^\sigma_i( \vec{\mathcal{Q}}_i , T ) \right|_{\vec{\mathcal{Q}}_i = \frac{\vec{q}_i - \vec{G}_i}{\gamma}, T = \gamma \tau}
\end{equation}
where $\vec{q}_i$ varies over the translated Brillouin zones $\Gamma^*_i - \vec{\kappa}_i, i \in \{1,2\}$. Note that we will generally simply evaluate the multiscale variables $\vec{\mathcal{K}}_i$ and $\vec{\mathcal{Q}}_i$ when there is no danger of confusion. 
\begin{remark}
    Note that we work with Fourier transforms of the wave-packet amplitudes, rather than with their continuum Bloch transforms. Although this leads to formulas involving sums over reciprocal lattice vectors, we find this less cumbersome than introducing a continuum Bloch transform with respect to reciprocal lattice vectors scaled by $\gamma$.  
\end{remark}

\subsection{Separation of terms into leading and higher order in momentum space} \label{sec:separate_leading_higherorder}

Since simplifying $H_{22} \Psi_2$ is so similar, we concentrate on simplifying $H_{11} \Psi_1$. This is equivalent to simplifying the momentum space expression 
\begin{equation} \label{eq:simplify_diagonal}
    \left( \mathcal{G}_1 H_{11} \mathcal{G}_1^{-1} \tilde{\Psi}_1 \right)(\vec{\kappa}_1 + \vec{q}_1) = H_1(\vec{\kappa}_1 + \vec{q}_1) \tilde{\Psi}_1(\vec{\kappa}_1 + \vec{q}_1).
\end{equation}
    The idea is to Taylor expand $H_1(\vec{\kappa}_1 + \vec{q}_1)$ in $\vec{q}_1$ and then transform back to real space. First recall that, using \eqref{eq:H_translation}, we have
\begin{equation}
    H^{\sigma \sigma'}_1(\vec{\kappa}_1 + \vec{q}_1) e^{- i \vec{G}_1 \cdot \vec{\tau}_1^{\sigma'}} = e^{- i \vec{G}_1 \cdot \vec{\tau}_1^\sigma} H_1^{\sigma \sigma'}(\vec{\kappa}_1 + (\vec{q}_1 - \vec{G}_1)) \quad \sigma, \sigma' \in \{ A, B \}.
\end{equation}
It follows that
\begin{equation} \label{eq:simplify_diagonal_2}
    \begin{split}
        &\left( H_1(\vec{\kappa}_1 + \vec{q}_1) \tilde{\Psi}_1(\vec{\kappa}_1 + \vec{q}_1) \right)^\sigma = \\
        &\frac{\gamma^{-1} |\Gamma^*|^{1/2}}{(2 \pi)^2} \sum_{\vec{G}_1 \in \Lambda_1^*} \sum_{\sigma' \in \{A,B\}} e^{- i \vec{G}_1 \cdot \vec{\tau}^\sigma_1} H^{\sigma \sigma'}_1(\vec{\kappa}_1 + (\vec{q}_1 - \vec{G}_1)) \oldhat{f}^{\sigma'}_1\left(\frac{\vec{q}_1 - \vec{G}_1}{\gamma},\gamma \tau\right).
    \end{split}
\end{equation}
Now, substituting \eqref{eq:monolayer_Dirac_twisted} (and its obvious generalization to non-zero $\vec{G}_1$) into \eqref{eq:simplify_diagonal_2}, we have
\begin{equation} \label{eq:diag}
    \begin{split}
        &\left( \mathcal{G}_1 H_{11} \mathcal{G}_1^{-1} \tilde{\Psi}_1 \right)^{\sigma}(\vec{\kappa}_1 + \vec{q}_1) =  \\
        &\frac{\gamma^{-1} |\Gamma^*|^{1/2}}{(2 \pi)^2} \sum_{\vec{G}_1 \in \Lambda_1^*} \sum_{\sigma' \in \{A,B\}} e^{- i \vec{G}_1 \cdot \vec{\tau}_1^\sigma} \gamma \left( \vec{\sigma}_{-\theta/2} \cdot \left(\frac{\vec{q}_1 - \vec{G}_1}{\gamma}\right) \right)^{\sigma \sigma'} \oldhat{f}^{\sigma'}_1 \left(\frac{\vec{q}_1 - \vec{G}_1}{\gamma},\gamma \tau\right) \\
        &+ (\tilde{r}^{\text{I}}_1)^\sigma(\vec{\kappa}_1 + \vec{q}_1,\tau), 
    \end{split}
\end{equation}
where $\tilde{r}^{\text{I}}_1$ is exactly
\begin{equation} \label{eq:r_1}
    \begin{split}
        &( \tilde{r}^{\text{I}}_1 )^\sigma(\vec{\kappa}_1 + \vec{q}_1,\tau) = \frac{\gamma^{-1} |\Gamma^*|^{1/2}}{(2 \pi)^2} \times   \\
        &\sum_{\vec{G}_1 \in \Lambda^*_1} \sum_{\sigma' \in \{A,B\}} e^{- i \vec{G}_1 \cdot \vec{\tau}^\sigma_1} \left( H_1\left(\vec{\kappa}_1 + \gamma \left( \frac{\vec{q}_1 - \vec{G}_1}{\gamma} \right) \right) - \gamma \vec{\sigma}_{-\theta/2} \cdot \left( \frac{\vec{q}_1 - \vec{G}_1}{\gamma} \right) \right)^{\sigma \sigma'} \oldhat{f}^{\sigma'}_1\left(\frac{\vec{q}_1 - \vec{G}_1}{\gamma},\gamma \tau\right).
    \end{split}
\end{equation}
We then define
\begin{equation}
    ( r^{\text{I}}_{\vec{R}_1} )^\sigma(\tau) := \left[ \mathcal{G}_1^{-1} \tilde{r}^{\text{I}}_1 \right]_{\vec{R}_1}^\sigma(\tau). 
\end{equation}

\subsection{Estimation of ``higher-order'' terms} \label{sec:estimate_residual} 

{\color{black} We now prove \eqref{eq:diag_decomposition_remainder}. Note, first, that, for any $\vec{G}_1 \in \Lambda_1^*$, we have
\begin{equation} \label{eq:first_est}
    \begin{split}
        &\left| \sum_{\sigma' \in \{A,B\}} \left( H_1\left( \vec{\kappa}_1 + \gamma \left( \frac{\vec{q}_1 - \vec{G}_1}{\gamma} \right) \right) - \gamma \vec{\sigma}_{-\theta/2} \cdot \left( \frac{\vec{q}_1 - \vec{G}_1}{\gamma} \right) \right)^{\sigma \sigma'} \oldhat{f}_1^{\sigma'}\left( \frac{\vec{q}_1 - \vec{G}_1}{\gamma} , \gamma \tau \right) \right|  \\
        &\leq C \gamma^2 \sup_{\sigma' \in \{A,B\}} \left| \frac{\vec{q}_1 - \vec{G}_1}{\gamma} \right|^2 \left| \oldhat{f}_1^{\sigma'}\left( \frac{\vec{q}_1 - \vec{G}_1}{\gamma} , \gamma \tau \right) \right|,  \\
    \end{split}
\end{equation}
where $C > 0$ is two times the supremum of the second derivatives of $H_1$, which are uniformly bounded because $H_1$ is smooth and periodic in $\vec{q}_1$. We can now estimate using the triangle inequality
\begin{equation}
    \begin{split}
        &\left\| ( r^{\text{I}}_1 )^\sigma(\tau) \right\|^2_{\ell^2(\mathbb{Z}^2)} = \left\| ( \tilde{r}^{\text{I}}_1 )^\sigma(\cdot,\tau) \right\|^2_{L^2(\Gamma_1^*)} \leq \frac{\gamma^{-2} |\Gamma^*|}{(2 \pi)^4} \int_{\Gamma_1^* - \vec{\kappa}_1} \\
        &\sum_{\vec{G}_1 \in \Lambda_1^*} \left| \sum_{\sigma' \in \{A,B\}} \left( H_1\left( \vec{\kappa}_1 + \gamma \left( \frac{\vec{q}_1 - \vec{G}_1}{\gamma} \right) \right) - \gamma \vec{\sigma}_{-\theta/2} \cdot \left( \frac{\vec{q}_1 - \vec{G}_1}{\gamma} \right) \right)^{\sigma \sigma'} \oldhat{f}_1^{\sigma'}\left( \frac{\vec{q}_1 - \vec{G}_1}{\gamma} , \gamma \tau \right) \right| \\
        &\sum_{\vec{G}_1' \in \Lambda_1^*} \left| \sum_{\sigma'' \in \{A,B\}} \left( H_1\left( \vec{\kappa}_1 + \gamma \left( \frac{\vec{q}_1 - \vec{G}_1'}{\gamma} \right) \right) - \gamma \vec{\sigma}_{-\theta/2} \cdot \left( \frac{\vec{q}_1 - \vec{G}_1'}{\gamma} \right) \right)^{\sigma \sigma''} \oldhat{f}_1^{\sigma''}\left( \frac{\vec{q}_1 - \vec{G}_1'}{\gamma} , \gamma \tau \right) \right| \; \text{d}\vec{q}_1,
    \end{split}
\end{equation}
and then \eqref{eq:first_est},
\begin{equation}
    \begin{split}
        &\left\| ( \tilde{r}^{\text{I}}_1 )^\sigma(\cdot,\tau) \right\|^2_{L^2(\Gamma_1^*)} \leq C \gamma^{2} \sup_{\sigma' \in \{A,B\}} \sup_{\sigma'' \in \{A,B\}} \int_{\Gamma_1^* - \vec{\kappa}_1} \\
        &\sum_{\vec{G}_1 \in \Lambda_1^*} \sum_{\vec{G}_1' \in \Lambda_1^*} \left| \frac{\vec{q}_1 - \vec{G}_1}{\gamma} \right|^2 \left| \oldhat{f}_1^{\sigma'}\left( \frac{\vec{q}_1 - \vec{G}_1}{\gamma} , \gamma \tau \right) \right| \left| \frac{\vec{q}_1 - \vec{G}_1'}{\gamma} \right|^2 \left| \oldhat{f}_1^{\sigma''}\left( \frac{\vec{q}_1 - \vec{G}_1'}{\gamma} , \gamma \tau \right) \right| \; \text{d}\vec{q}_1.
    \end{split}
\end{equation}
Writing $\vec{q}_1 - \vec{G}_1' = \vec{q}_1 - \vec{G}_1 + \vec{G}_1 - \vec{G}_1'$ and then changing variables in the sum over $\vec{G}_1'$ to $\tilde{\vec{G}}_1 = \vec{G}_1' - \vec{G}_1$, we obtain
\begin{equation}
    \begin{split}
        &\left\| ( \tilde{r}^{\text{I}}_1 )^\sigma(\cdot,\tau) \right\|^2_{L^2(\Gamma_1^*)} \leq C \gamma^{2} \sup_{\sigma' \in \{A,B\}} \sup_{\sigma'' \in \{A,B\}} \int_{\Gamma_1^* - \vec{\kappa}_1} \\
        &\sum_{\vec{G}_1 \in \Lambda_1^*} \sum_{\tilde{\vec{G}}_1 \in \Lambda_1^*} \left| \frac{\vec{q}_1 - \vec{G}_1}{\gamma} \right|^2 \left| \oldhat{f}_1^{\sigma'}\left( \frac{\vec{q}_1 - \vec{G}_1}{\gamma} , \gamma \tau \right) \right| \left| \frac{\vec{q}_1 - \vec{G}_1 - \tilde{\vec{G}}_1}{\gamma} \right|^2 \left| \oldhat{f}_1^{\sigma''}\left( \frac{\vec{q}_1 - \vec{G}_1 - \tilde{\vec{G}}_1}{\gamma} , \gamma \tau \right) \right| \; \text{d}\vec{q}_1.
    \end{split}
\end{equation}
Now we can replace the sum over $\vec{G}_1$ with an integral over all of $\mathbb{R}^2$ as 
\begin{equation}
    \begin{split}
        &\left\| ( \tilde{r}^{\text{I}}_1 )^\sigma(\cdot,\tau) \right\|^2_{L^2(\Gamma_1^*)} \leq C \gamma^{2} \sup_{\sigma' \in \{A,B\}} \sup_{\sigma'' \in \{A,B\}} \int_{\mathbb{R}^2} \\
        &\sum_{\tilde{\vec{G}}_1 \in \Lambda_1^*} \left| \frac{\vec{q}_1}{\gamma} \right|^2 \left| \oldhat{f}_1^{\sigma'}\left( \frac{\vec{q}_1}{\gamma} , \gamma \tau \right) \right| \left| \frac{\vec{q}_1 - \tilde{\vec{G}}_1}{\gamma} \right|^2 \left| \oldhat{f}_1^{\sigma''}\left( \frac{\vec{q}_1 - \tilde{\vec{G}}_1}{\gamma} , \gamma \tau \right) \right| \; \text{d}\vec{q}_1.
    \end{split}
\end{equation}
Changing variables in the integral gives
\begin{equation} \label{eq:all_terms}
    \begin{split}
        &\left\| ( \tilde{r}^{\text{I}}_1 )^\sigma(\cdot,\tau) \right\|^2_{L^2(\Gamma_1^*)} \leq C \gamma^{4} \sup_{\sigma' \in \{A,B\}} \sup_{\sigma'' \in \{A,B\}} \int_{\mathbb{R}^2} \\
        &\sum_{\tilde{\vec{G}}_1 \in \Lambda_1^*} \left| \vec{q}_1 \right|^2 \left| \oldhat{f}_1^{\sigma'}\left( \vec{q}_1 , \gamma \tau \right) \right| \left| \vec{q}_1 - \frac{\tilde{\vec{G}}_1}{\gamma} \right|^2 \left| \oldhat{f}_1^{\sigma''}\left( \vec{q}_1 - \frac{\tilde{\vec{G}}_1}{\gamma} , \gamma \tau \right) \right| \; \text{d}\vec{q}_1.
    \end{split}
\end{equation}
We now split the sum into the $\tilde{\vec{G}}_1 = \vec{0}$ and the $\tilde{\vec{G}}_1 \neq \vec{0}$ terms
\begin{equation} 
    \begin{split}
        &\left\| ( \tilde{r}^{\text{I}}_1 )^\sigma(\cdot,\tau) \right\|^2_{L^2(\Gamma_1^*)} \leq C \gamma^{4} \sup_{\sigma' \in \{A,B\}} \sup_{\sigma'' \in \{A,B\}} \int_{\mathbb{R}^2} \\
        &\sum_{\tilde{\vec{G}}_1 \in \Lambda_1^*, \tilde{\vec{G}}_1 \neq \vec{0}} \left| \vec{q}_1 \right|^2 \left| \oldhat{f}_1^{\sigma'}\left( \vec{q}_1 , \gamma \tau \right) \right| \left| \vec{q}_1 - \frac{\tilde{\vec{G}}_1}{\gamma} \right|^2 \left| \oldhat{f}_1^{\sigma''}\left( \vec{q}_1 - \frac{\tilde{\vec{G}}_1}{\gamma} , \gamma \tau \right) \right| \; \text{d}\vec{q}_1 \\
        &+ C \gamma^{4} \sup_{\sigma' \in \{A,B\}} \left\| (\cdot)^2 \oldhat{f}_1^{\sigma'}(\cdot;\gamma \tau) \right\|^2_{L^2(\mathbb{R}^2)}.
    \end{split}
\end{equation}
We now claim that the sum of the $\tilde{\vec{G}}_1 \neq \vec{0}$ terms is higher-order in $\gamma$, assuming sufficient regularity of the envelope functions. For fixed $\tilde{\vec{G}}_1 \neq \vec{0}$, we can partition $\mathbb{R}^2$ into half-planes $\mathbb{R}^2_+\left(\frac{\tilde{\vec{G}}_1}{\gamma}\right)$ and $\mathbb{R}^2_-\left(\frac{\tilde{\vec{G}}_1}{\gamma}\right)$, where
\begin{equation} \label{eq:regions}
    \vec{q}_1 \in \mathbb{R}^2_+\left(\frac{\tilde{\vec{G}}_1}{\gamma}\right) \implies |\vec{q}_1| \geq \left| \frac{\tilde{\vec{G}}_1}{2 \gamma} \right|, \quad \vec{q}_1 \in \mathbb{R}^2_-\left(\frac{\tilde{\vec{G}}_1}{\gamma}\right) \implies \left|\vec{q}_1 - \frac{\tilde{\vec{G}}_1}{2 \gamma}\right| \geq \left| \frac{\tilde{\vec{G}}_1}{2 \gamma} \right|,
\end{equation}
by splitting $\mathbb{R}^2$ by the line which bisects $\frac{\tilde{\vec{G}}_1}{\gamma}$ (i.e., the line which passes through $\frac{\tilde{\vec{G}}_1}{2 \gamma}$ perpendicular to $\tilde{\vec{G}}_1$). We now consider the integral over one of these regions (the other is similar) 
\begin{equation}
    \inty{\mathbb{R}^2_+\left(\frac{\tilde{\vec{G}}_1}{\gamma}\right)}{}{ \left| \vec{q}_1 \right|^2 \left| \oldhat{f}_1^{\sigma'}\left( \vec{q}_1 , \gamma \tau \right) \right| \left| \vec{q}_1 - \frac{\tilde{\vec{G}}_1}{\gamma} \right|^2 \left| \oldhat{f}_1^{\sigma''}\left( \vec{q}_1 - \frac{\tilde{\vec{G}}_1}{\gamma} , \gamma \tau \right) \right| }{\vec{q}_1}.
\end{equation}
Multiplying and dividing by $|\vec{q}_1|^4$ and using the first inequality of \eqref{eq:regions} we obtain
\begin{equation} \label{eq:trick_estimate}
    \begin{split}
        &\inty{\mathbb{R}^2_+\left(\frac{\tilde{\vec{G}}_1}{\gamma}\right)}{}{ \left| \vec{q}_1 \right|^2 \left| \oldhat{f}_1^{\sigma'}\left( \vec{q}_1 , \gamma \tau \right) \right| \left| \vec{q}_1 - \frac{\tilde{\vec{G}}_1}{\gamma} \right|^2 \left| \oldhat{f}_1^{\sigma''}\left( \vec{q}_1 - \frac{\tilde{\vec{G}}_1}{\gamma} , \gamma \tau \right) \right| }{\vec{q}_1} \\
        &\leq \frac{16 \gamma^4}{|\tilde{\vec{G}}_1|^4} \inty{\mathbb{R}^2_+\left(\frac{\tilde{\vec{G}}_1}{\gamma}\right)}{}{ \left| \vec{q}_1 \right|^6 \left| \oldhat{f}_1^{\sigma'}\left( \vec{q}_1 , \gamma \tau \right) \right| \left| \vec{q}_1 - \frac{\tilde{\vec{G}}_1}{\gamma} \right|^2 \left| \oldhat{f}_1^{\sigma''}\left( \vec{q}_1 - \frac{\tilde{\vec{G}}_1}{\gamma} , \gamma \tau \right) \right| }{\vec{q}_1} \\
        &\leq \frac{16 \gamma^4}{|\tilde{\vec{G}}_1|^4} \| (\cdot)^6 \oldhat{f}_1^{\sigma'}(\cdot,\gamma \tau) \|_{L^2(\mathbb{R}^2)} \| (\cdot)^2 \oldhat{f}_1^{\sigma''}(\cdot,\gamma \tau) \|_{L^2(\mathbb{R}^2)},
    \end{split}
\end{equation}
where we used Cauchy-Schwarz for the final inequality. Combining this with the analogous estimate for the integral over $\mathbb{R}^2_-$, we have
\begin{equation}
    \begin{split}
        &\left\| ( \tilde{r}^{\text{I}}_1 )^\sigma(\cdot,\tau) \right\|^2_{L^2(\Gamma_1^*)} \leq \\
        &C \gamma^{8} \sup_{\sigma' \in \{A,B\}} \sup_{\sigma'' \in \{A,B\}} \| (\cdot)^6 \oldhat{f}_1^{\sigma'}(\cdot,\gamma \tau) \|_{L^2(\mathbb{R}^2)} \| (\cdot)^2 \oldhat{f}_1^{\sigma''}(\cdot,\gamma \tau) \|_{L^2(\mathbb{R}^2)} \sum_{\tilde{\vec{G}}_1 \in \Lambda_1^*, \tilde{\vec{G}}_1 \neq \vec{0}} \frac{1}{|\tilde{\vec{G}}_1|^4}  \\
        &+ C \gamma^{4} \sup_{\sigma' \in \{A,B\}} \left\| (\cdot)^2 \oldhat{f}_1^{\sigma'}(\cdot;\gamma \tau) \right\|^2_{L^2(\mathbb{R}^2)}.
    \end{split}
\end{equation}
Simplifying this we end up with the estimate
\begin{equation}
    \left\| ( \tilde{r}^{\text{I}}_1 )(\cdot,\tau) \right\|^2_{L^2(\Gamma_1^*;\mathbb{C}^2)} \leq C \gamma^4 \left\| f_1(\cdot,\gamma \tau) \right\|^2_{H^2(\mathbb{R}^2;\mathbb{C}^2)} + C \gamma^8 \left\| f_1(\cdot,\gamma \tau) \right\|^2_{H^6(\mathbb{R}^2;\mathbb{C}^2)},
\end{equation}
from which \eqref{eq:diag_decomposition_remainder} follows via equivalence of finite-dimensional norms.
}

\subsection{Transforming leading terms back to real space} \label{sec:transform_back}

To derive \eqref{eq:diag_decomposition}, it remains only to transform the leading order term in \eqref{eq:diag}
\begin{equation}
    \frac{\gamma^{-1} |\Gamma^*|^{1/2}}{(2 \pi)^2} \sum_{\vec{G}_1 \in \Lambda_1^*} e^{- i \vec{G}_1 \cdot \vec{\tau}_1^\sigma} \gamma \left( \vec{\sigma}_{-\theta/2} \cdot \left(\frac{\vec{q}_1 - \vec{G}_1}{\gamma}\right) \right)^{\sigma \sigma'} \oldhat{f}^{\sigma'}_1 \left(\frac{\vec{q}_1 - \vec{G}_1}{\gamma},\gamma \tau\right),
\end{equation}
back to real space. Taking the inverse Bloch transform and inserting the Fourier transform formula the expression becomes 
\begin{equation}
    \begin{split}
        &\frac{\gamma^{-1}}{(2 \pi)^2} \sum_{\vec{G}_1 \in \Lambda_1^*} e^{- i \vec{G}_1 \cdot \vec{\tau}_1^\sigma} \gamma \\
        &\times\inty{\Gamma_1^*}{}{ e^{i \vec{k}_1 \cdot ( \vec{R}_1 + \vec{\tau}_1^\sigma )} \inty{\field{R}^2}{}{ e^{- i \left( \frac{\vec{k}_1 - \vec{\kappa}_1 - \vec{G}_1}{\gamma} \right) \cdot \vec{X}} \left( \vec{\sigma}_{-\theta/2} \cdot ( - i \nabla_{\vec{X}} ) \right)^{\sigma \sigma'} {f}^{\sigma'}_1 (\vec{X},\gamma \tau) }{\vec{X}} }{\vec{k}_1}.
    \end{split}
\end{equation}
Now, noting that $e^{i \vec{k}_1 \cdot \vec{R}_1} = e^{i (\vec{k}_1 - \vec{G}_1) \cdot \vec{R}_1}$ for any $\vec{R}_1 \in \Lambda_1$ and $\vec{G}_1 \in \Lambda_1^*$, we can write all dependence on $\vec{k}_1$ in terms of dependence on $\vec{k}_1 - \vec{G}_1$, and then re-write the sum over $\vec{G}_1$ and integral over $\vec{k}_1$ as an integral over $\field{R}^2$ as
\begin{equation}
    \begin{split}
        &\frac{\gamma^{-1}}{(2 \pi)^2} \gamma \\
        &\times\inty{\field{R}^2}{}{ e^{i \vec{k}_1 \cdot ( \vec{R}_1 + \vec{\tau}_1^\sigma )} \inty{\field{R}^2}{}{ e^{- i \left( \frac{\vec{k}_1 - \vec{\kappa}_1}{\gamma} \right) \cdot \vec{X}} \left( \vec{\sigma}_{-\theta/2} \cdot ( - i \nabla_{\vec{X}} ) \right)^{\sigma \sigma'} {f}^{\sigma'}_1 (\vec{X},\gamma \tau) }{\vec{X}} }{\vec{k}_1}.
    \end{split}
\end{equation}
Now using the identity
\begin{equation}
    \frac{\gamma^{-1}}{(2 \pi)^2} \inty{\field{R}^2}{}{ e^{i \frac{\vec{k}_1}{\gamma} \cdot \left( \gamma (\vec{R}_1 + \vec{\tau}_1^\sigma) - \vec{X} \right)} }{\vec{k}_1} = \gamma \delta( \gamma ( \vec{R}_1 + \vec{\tau}_1^\sigma ) - \vec{X} ),
\end{equation}
and then integrating over $\vec{X}$, yields the leading term of \eqref{eq:diag_decomposition}.

\section{Proof of Lemma \ref{lem:diag_theta_lemma}} \label{sec:proof_of_diag_theta_lemma}

The proof is straightforward. Substituting 
\begin{equation}
    \vec{\sigma}_{\theta_i/2} = \vec{\sigma} + \left[ \vec{\sigma}_{\theta_i/2} - \vec{\sigma} \right]
\end{equation}
into \eqref{eq:diag_decomposition} yields \eqref{eq:diag_theta_decomposition}, with
\begin{equation} \label{eq:diag_theta_remainder}
\begin{split}
    ( r^\text{II}_{\vec{R}_i} )^\sigma = \gamma \sum_{\sigma' \in \{A,B\}} \left( \gamma \left[ \vec{\sigma}_{\theta_i/2} - \vec{\sigma} \right] \cdot (- i \nabla_{\vec{X}}) \right)^{\sigma \sigma'} \left. f^{\sigma'}_i(\vec{X},\gamma \tau) \right|_{\vec{X} = \gamma ( \vec{R}_i + \vec{\tau}^{\sigma'}_i )} e^{i \vec{\kappa}_i \cdot ( \vec{R}_i + \vec{\tau}_i^\sigma )}.
\end{split}
\end{equation}
{\color{black} Following the same steps as in Section \ref{sec:estimate_residual} together with Taylor-expansion of the exponentials $e^{\pm i \theta/2}$ with respect to $\theta$ results in
\begin{equation}
    \begin{split}
        &\| (\tilde{r}^{\text{II}}_1)^\sigma(\cdot,\tau)) \|^2_{L^2(\Gamma_1^*)} \leq C \theta^2 \gamma^2 \sup_{\sigma' \in \{A,B\}} \sup_{\sigma'' \in \{A,B\}}  \\
        &\sum_{\tilde{\vec{G}}_1 \in \Lambda_1^*} \inty{\mathbb{R}^2}{}{ \left| \vec{q}_1 \right| \left| \oldhat{f}_1^{\sigma'}(\vec{q}_1,\gamma \tau) \right| \left| \vec{q}_1 - \frac{ \tilde{\vec{G}}_1 }{ \gamma } \right| \left| \oldhat{f}_1^{\sigma''}\left( \vec{q}_1 - \frac{ \tilde{\vec{G}}_1 }{ \gamma } , \gamma \tau \right) \right| }{\vec{q}_1}.
    \end{split}
\end{equation}
Just as in Section \ref{sec:estimate_residual}, the leading term comes from the $\tilde{\vec{G}}_1 = \vec{0}$ term, and the $\tilde{\vec{G}}_1 \neq \vec{0}$ terms provide a correction
\begin{equation}
    \| \tilde{r}^{\text{II}}_1(\cdot,\tau) \|^2_{L^2(\Gamma_1^*;\mathbb{C}^2)} \leq C \theta^2 \gamma^2 \| f_1(\cdot,\gamma \tau) \|^2_{H^1(\mathbb{R}^2;\mathbb{C}^2)} + C \theta^2 \gamma^6 \| f_1(\cdot,\gamma \tau) \|^2_{H^5(\mathbb{R}^2;\mathbb{C}^2)},
\end{equation}
from which \eqref{eq:diag_theta_decomposition_remainder} follows.
}

\section{Proof of Lemma \ref{lem:off_diag_lemma}} \label{sec:proof_of_off_diag_lemma}

We organize the proof of Lemma \ref{lem:off_diag_lemma} into sections just like the proof of Lemma \ref{lem:diag_lemma}. First, we will separate the momentum space expression for $H_{12} \Psi_2$ into ``leading'' and ``higher-order'' terms in Section \ref{sec:separate_leading_higherorder_2}. We will then estimate the higher-order terms in Section \ref{sec:estimate_residual_2} and transform the leading terms back to real space in Section \ref{sec:transform_back_2}.

\subsection{Separation of terms into leading and higher order in momentum space} \label{sec:separate_leading_higherorder_2}

In momentum space, $( H_{12} \Psi_2 )^\sigma$ takes the form (recall \eqref{eq:off_diag_1})
\begin{equation} 
    \begin{split}
        &\left( \mathcal{G}_1 H_{12} \mathcal{G}_2^{-1} \tilde{\Psi}_2 \right)^\sigma(\vec{\kappa}_1 + \vec{q}_1) =  \\
        &\sum_{\sigma' \in \{A,B\}} \int_{\Gamma^*_2 - \vec{\kappa}_2} \sum_{\vec{G}_1 \in \Lambda_1^*} \sum_{\vec{G}_2 \in \Lambda_2^*} \\
        &e^{i [\vec{G}_1 \cdot \vec{\tau}_1^\sigma - \vec{G}_2 \cdot \vec{\tau}_2^{\sigma'}]} \oldhat{h}( \vec{\kappa}_1 + \vec{G}_1 + \vec{q}_1 ; \ell ) \delta( \vec{\kappa}_1 - \vec{\kappa}_2 + \vec{G}_1 - \vec{G}_2 + \vec{q}_1 - \vec{q}_2 ) \tilde{\Psi}^{\sigma'}_2(\vec{\kappa}_2 + \vec{q}_2) \; \text{d}\vec{q}_2.
    \end{split}
\end{equation}
Substituting the wave-packet {\it ansatz} \eqref{eq:WP_momentum} we have
\begin{equation}
    \begin{split}
        &\frac{\gamma^{-1} |\Gamma^*|^{1/2}}{(2 \pi)^2} \sum_{\sigma' \in \{A,B\}} \int_{\Gamma^*_2 - \vec{\kappa}_2} \sum_{\vec{G}_1 \in \Lambda_1^*} \sum_{\vec{G}_2 \in \Lambda_2^*} \sum_{\vec{G}_2' \in \Lambda_2^*} \\
        &e^{i [\vec{G}_1 \cdot \vec{\tau}_1^\sigma - (\vec{G}_2 + \vec{G}_2') \cdot \vec{\tau}_2^{\sigma'}]} \oldhat{h}( \vec{\kappa}_1 + \vec{G}_1 + \vec{q}_1 ; \ell ) \delta( \vec{\kappa}_1 - \vec{\kappa}_2 + \vec{G}_1 - \vec{G}_2 + \vec{q}_1 - \vec{q}_2 ) \oldhat{f}_2^{\sigma'}\left( \frac{ \vec{q}_2 - \vec{G}_2' }{ \gamma } , \gamma \tau \right) \; \text{d}\vec{q}_2.
    \end{split}
\end{equation}
Replacing the sum over $\vec{G}_2$ with a sum over $\tilde{\vec{G}}_2 = \vec{G}_2 + \vec{G}_2'$ and then dropping the tilde we have
\begin{equation}
    \begin{split}
        &\frac{\gamma^{-1} |\Gamma^*|^{1/2}}{(2 \pi)^2} \sum_{\sigma' \in \{A,B\}} \int_{\Gamma^*_2 - \vec{\kappa}_2} \sum_{\vec{G}_1 \in \Lambda_1^*} \sum_{\vec{G}_2 \in \Lambda_2^*} \sum_{\vec{G}_2' \in \Lambda_2^*} \\
        &e^{i [\vec{G}_1 \cdot \vec{\tau}_1^\sigma - \vec{G}_2 \cdot \vec{\tau}_2^{\sigma'}]} \oldhat{h}( \vec{\kappa}_1 + \vec{G}_1 + \vec{q}_1 ; \ell ) \delta( \vec{\kappa}_1 - \vec{\kappa}_2 + \vec{G}_1 - \vec{G}_2 + \vec{q}_1 - (\vec{q}_2 - \vec{G}_2') ) \oldhat{f}_2^{\sigma'}\left( \frac{ \vec{q}_2 - \vec{G}_2' }{ \gamma } , \gamma \tau \right) \; \text{d}\vec{q}_2.
    \end{split}
\end{equation}
But now we can replace the sum over $\vec{G}_2'$ and integral with respect to $\vec{q}_2$ over $\Gamma^*_2 - \vec{\kappa}_2$ with an integral over $\mathbb{R}^2$. Carrying out this integral we have
\begin{equation}
    \begin{split}
        &\frac{\gamma^{-1} |\Gamma^*|^{1/2}}{(2 \pi)^2} \sum_{\sigma' \in \{A,B\}} \sum_{\vec{G}_1 \in \Lambda_1^*} \sum_{\vec{G}_2 \in \Lambda_2^*} \\
        &e^{i [\vec{G}_1 \cdot \vec{\tau}_1^\sigma - \vec{G}_2 \cdot \vec{\tau}_2^{\sigma'}]} \oldhat{h}( \vec{\kappa}_1 + \vec{G}_1 + \vec{q}_1 ; \ell ) \oldhat{f}_2^{\sigma'}\left(\frac{ \vec{\kappa}_1 - \vec{\kappa}_2 + \vec{G}_1 - \vec{G}_2 + \vec{q}_1 }{ \gamma } , \gamma \tau \right).
    \end{split}
\end{equation}
Making the substitution
\begin{equation} \label{eq:local_approx}
    \begin{split}
        &\oldhat{h}(\vec{\kappa}_1 + \vec{q}_1 + \vec{G}_1;\ell) = \\
        &\oldhat{h}(\vec{\kappa}_2 + \vec{G}_2;\ell) + \left[ \oldhat{h}(\vec{\kappa}_1 + \vec{q}_1 + \vec{G}_1;\ell) - \oldhat{h}(\vec{\kappa}_2 + \vec{G}_2;\ell) \right],
    \end{split}
\end{equation}
puts the off-diagonal terms into the form
\begin{equation} \label{eq:off_diag_terms}
    \begin{split}
        &\frac{\gamma^{-1} |\Gamma^*|^{1/2}}{(2 \pi)^2} \sum_{\sigma' \in \{A,B\}} \sum_{\vec{G}_1 \in \Lambda_1^*} \sum_{\vec{G}_2 \in \Lambda_2^*} \\
        &e^{i [\vec{G}_1 \cdot \vec{\tau}_1^\sigma - \vec{G}_2 \cdot \vec{\tau}_2^{\sigma'}]} \oldhat{h}( \vec{\kappa}_2 + \vec{G}_2 ; \ell ) \oldhat{f}^{\sigma'}_2\left(\frac{\vec{\kappa}_1 - \vec{\kappa}_2 + \vec{G}_1 - \vec{G}_2 + \vec{q}_1 }{\gamma},\gamma \tau\right) + \left( \tilde{r}^{\text{III}}_1 \right)^\sigma(\vec{\kappa}_1 + \vec{q}_1,\tau),
    \end{split}
\end{equation}
where
\begin{equation}
    \begin{split}
        &\left( \tilde{r}^\text{III}_1 \right)^\sigma(\vec{\kappa}_1 + \vec{q}_1,\tau) := \frac{\gamma^{-1} |\Gamma^*|^{1/2}}{(2 \pi)^2} \sum_{\sigma' \in \{A,B\}} \sum_{\vec{G}_1 \in \Lambda_1^*} \sum_{\vec{G}_2 \in \Lambda_2^*} \\
        &e^{i [\vec{G}_1 \cdot \vec{\tau}_1^\sigma - \vec{G}_2 \cdot \vec{\tau}_2^{\sigma'}]} \left[ \oldhat{h}( \vec{\kappa}_1 + \vec{G}_1 + \vec{q}_1 ; \ell ) - \oldhat{h}( \vec{\kappa}_2 + \vec{G}_2 ; \ell ) \right] \oldhat{f}^{\sigma'}_2\left(\frac{\vec{\kappa}_1 - \vec{\kappa}_2 + \vec{G}_1 - \vec{G}_2 + \vec{q}_1 }{\gamma},\gamma \tau\right).
    \end{split}
\end{equation}
    Now recall that $\oldhat{h}$ decays (Assumption \ref{as:h_regularity}). It follows that the dominant terms in the sum \eqref{eq:off_diag_terms} are those which minimize $|\vec{\kappa}_2 + \vec{G}_2|$, i.e. those with $\vec{G}_2 = \vec{0}$, $\vec{G}_2 = \vec{b}_{2,2}$, and $\vec{G}_2 = - \vec{b}_{2,1}$, which all give $|\vec{\kappa}_2 + \vec{G}_2| = |\vec{\kappa}|$. It is straightforward to see that the next largest term will have $|\vec{\kappa}_2 + \vec{G}_2| = 2 |\vec{\kappa}|$. So we can further simplify the off-diagonal terms as
    \begin{equation}  \label{eq:off_diag_again}
    \begin{split}
        &\frac{\gamma^{-1} |\Gamma^*|^{1/2}}{(2 \pi)^2} \sum_{\sigma' \in \{A,B\}} \sum_{\vec{G}_1 \in \Lambda_1^*} e^{i \vec{G}_1 \cdot \vec{\tau}_1^\sigma} \oldhat{h}( \vec{\kappa}_2 ; \ell ) \oldhat{f}^{\sigma'}_2\left(\frac{\vec{\kappa}_1 - \vec{\kappa}_2 + \vec{G}_1 + \vec{q}_1 }{\gamma},\gamma \tau\right)    \\
        &+ e^{i [\vec{G}_1 \cdot \vec{\tau}_1^\sigma - \vec{b}_{2,2} \cdot \vec{\tau}_2^{\sigma'}]} \oldhat{h}( \vec{\kappa}_2 + \vec{b}_{2,2} ; \ell ) \oldhat{f}^{\sigma'}_2\left(\frac{\vec{\kappa}_1 - \vec{\kappa}_2 + \vec{G}_1 - \vec{b}_{2,2} + \vec{q}_1 }{\gamma},\gamma \tau\right)  \\
        &+ e^{i [\vec{G}_1 \cdot \vec{\tau}_1^\sigma + \vec{b}_{2,1} \cdot \vec{\tau}_2^{\sigma'}]} \oldhat{h}( \vec{\kappa}_2 - \vec{b}_{2,1} ; \ell) \oldhat{f}^{\sigma'}_2\left(\frac{\vec{\kappa}_1 - \vec{\kappa}_2 + \vec{G}_1 + \vec{b}_{2,1} + \vec{q}_1 }{\gamma},\gamma \tau\right)  \\
        &+ \left( \tilde{r}^\text{III}_1 \right)^\sigma(\vec{\kappa}_1 + \vec{q}_1,\tau) + \left( \tilde{r}^\text{IV}_1 \right)^\sigma(\vec{\kappa}_1 + \vec{q}_1,\tau),
    \end{split}
\end{equation}
where
\begin{equation} 
    \begin{split}
        &\left( \tilde{r}^{\text{IV}}_1 \right)^\sigma(\vec{\kappa}_1 + \vec{q}_1,\tau) := \frac{\gamma^{-1} |\Gamma^*|^{1/2}}{(2 \pi)^2} \sum_{\sigma' \in \{A,B\}} \sum_{\vec{G}_1 \in \Lambda_1^*} \sum_{\substack{\vec{G}_2 \in \Lambda_2^* \\ \vec{G}_2 \notin \{ \vec{0}, \vec{b}_{2,2}, - \vec{b}_{2,1} \} }} \\
        &e^{i [\vec{G}_1 \cdot \vec{\tau}_1^\sigma - \vec{G}_2 \cdot \vec{\tau}_2^{\sigma'}]} \oldhat{h}( \vec{\kappa}_2 + \vec{G}_2 ; \ell ) \oldhat{f}^{\sigma'}_2\left(\frac{\vec{\kappa}_1 - \vec{\kappa}_2 + \vec{G}_1 - \vec{G}_2 + \vec{q}_1 }{\gamma},\gamma \tau\right).
    \end{split}
\end{equation}
We now simplify further. First, using radial symmetry of $\oldhat{h}$ (Assumption \ref{as:h_regularity}), we have
\begin{equation} \label{eq:simplify_moire_potential}
    \oldhat{h}( \vec{\kappa}_2 ; \ell ) = \oldhat{h}( \vec{\kappa}_2 + \vec{b}_{2,2} ; \ell ) = \oldhat{h}( \vec{\kappa}_2 - \vec{b}_{2,1} ; \ell ) = \oldhat{h}( |\vec{\kappa}| ; \ell ).
\end{equation}
Then, note that we can replace $\vec{G}_1$ by $- \vec{G}_1$, $\vec{b}_{1,2} - \vec{G}_1$, and $- \vec{b}_{1,1} - \vec{G}_1$ in the first, second, and third terms of \eqref{eq:off_diag_again}, respectively. Each term can then be written simply in terms of the vectors \eqref{eq:momentum_space_hops} as
\begin{equation} 
    \begin{split}
        &\frac{\gamma^{-1} |\Gamma^*|^{1/2}}{(2 \pi)^2} \oldhat{h}( |\vec{\kappa}| ; \ell ) \sum_{\sigma' \in \{A,B\}} \sum_{\vec{G}_1 \in \Lambda_1^*} e^{- i \vec{G}_1 \cdot \vec{\tau}_1^\sigma} \oldhat{f}^{\sigma'}_2\left(\frac{\vec{q}_1 + \mathfrak{s}_1 - \vec{G}_1 }{\gamma},\gamma \tau\right)    \\
        &+ e^{i [\vec{b}_{1,2} \cdot \vec{\tau}_1^\sigma - \vec{b}_{2,2} \cdot \vec{\tau}_2^{\sigma'}]} e^{- i \vec{G}_1 \cdot \vec{\tau}_1^\sigma} \oldhat{f}^{\sigma'}_2\left(\frac{\vec{q}_1 + \mathfrak{s}_2 - \vec{G}_1}{\gamma},\gamma \tau\right)  \\
        &+ e^{i [- \vec{b}_{1,1} \cdot \vec{\tau}_1^\sigma + \vec{b}_{2,1} \cdot \vec{\tau}_2^{\sigma'}]} e^{- i \vec{G}_1 \cdot \vec{\tau}_1^{\sigma}} \oldhat{f}^{\sigma'}_2\left(\frac{\vec{q}_1 + \mathfrak{s}_3 - \vec{G}_1 }{\gamma},\gamma \tau\right)  \\
        &+ \left( \tilde{r}^{\text{III}}_1 \right)^\sigma(\vec{\kappa}_1 + \vec{q}_1,\tau) + \left( \tilde{r}^{\text{IV}}_1 \right)^\sigma(\vec{\kappa}_1 + \vec{q}_1,\tau).
    \end{split}
\end{equation}
This expression can be written more simply in terms of the hopping matrices \eqref{eq:hopping_matrices} as
\begin{equation} \label{eq:off_diagonal_terms}
    \begin{split}
        &\left( \mathcal{G}_1 H_{12} \mathcal{G}_2^{-1} \tilde{\Psi}_2 \right)^\sigma(\vec{\kappa}_1 + \vec{q}_1) =  \\
        &\frac{\gamma^{-1} |\Gamma^*|^{1/2}}{(2 \pi)^2} \oldhat{h}(|\vec{\kappa}|;\ell) \sum_{\sigma' \in \{A,B\}} \sum_{\vec{G}_1 \in \Lambda_1^*} e^{- i \vec{G}_1 \cdot \vec{\tau}_1^\sigma} T^{\sigma \sigma'}_1 \oldhat{f}^{\sigma'}_2\left(\frac{\vec{q}_1 + \mathfrak{s}_1 - \vec{G}_1 }{\gamma},\gamma \tau\right)    \\
        &+ e^{- i \vec{G}_1 \cdot \vec{\tau}_1^\sigma} T^{\sigma \sigma'}_2 \oldhat{f}^{\sigma'}_2\left(\frac{\vec{q}_1 + \mathfrak{s}_2 - \vec{G}_1}{\gamma},\gamma \tau\right) + e^{- i \vec{G}_1 \cdot \vec{\tau}_1^{\sigma}} T^{\sigma \sigma'}_3 \oldhat{f}^{\sigma'}_2\left(\frac{\vec{q}_1 + \mathfrak{s}_3 - \vec{G}_1 }{\gamma},\gamma \tau\right)  \\
        &+ \left( \tilde{r}^{\text{III}}_1 \right)^\sigma(\vec{\kappa}_1 + \vec{q}_1,\tau) + \left( \tilde{r}^{\text{IV}}_1 \right)^\sigma(\vec{\kappa}_1 + \vec{q}_1,\tau).
    \end{split}
\end{equation}
We define
\begin{equation}
    ( r^{\text{III}}_{\vec{R}_1} )^\sigma(\tau) := [ \mathcal{G}^{-1}_1 \tilde{r}^\text{III}_1 ]_{\vec{R}_1}^\sigma(\tau), \quad ( r^{\text{IV}}_{\vec{R}_1} )^\sigma(\tau) := [ \mathcal{G}^{-1}_1 \tilde{r}^\text{IV}_1 ]_{\vec{R}_1}^\sigma(\tau).
\end{equation}

\subsection{Estimation of ``higher-order'' terms} \label{sec:estimate_residual_2}

We now prove \eqref{eq:r2_bound},
{\color{black} starting from
\begin{equation} 
    \begin{split}
        &\left\| \left( \tilde{r}^\text{III}_1 \right)^\sigma \right\|_{L^2(\Gamma_1^*;\mathbb{C})}^2 = \frac{\gamma^{-2} |\Gamma^*|}{(2 \pi)^4} \int_{\Gamma^*_1} \left| \sum_{\sigma' \in \{A,B\}} \sum_{\vec{G}_1 \in \Lambda_1^*} \sum_{\vec{G}_2 \in \Lambda_2^*} \right. \\
        &\left. e^{i \left[ \vec{G}_1 \cdot \vec{\tau}_1^\sigma - \vec{G}_2 \cdot \vec{\tau}_2^{\sigma'}\right]} \left[ \oldhat{h}( \vec{\kappa}_1 + \vec{q}_1 + \vec{G}_1;\ell) - \oldhat{h}( \vec{\kappa}_2 + \vec{G}_2 ; \ell ) \right] \oldhat{f}^{\sigma'}_2\left(\frac{\vec{\kappa}_1 - \vec{\kappa}_2 + \vec{G}_1 - \vec{G}_2 + \vec{q}_1 }{\gamma},\gamma \tau\right) \right|^2 \; \text{d}\vec{q}_1.
    \end{split}
\end{equation}
Using the triangle inequality we have
\begin{equation}
    \begin{split}
        &\left\| \left( \tilde{r}^\text{III}_1 \right)^\sigma \right\|_{L^2(\Gamma_1^*;\mathbb{C})}^2 \leq \frac{\gamma^{-2} |\Gamma^*|}{(2 \pi)^4} \int_{\Gamma^*_1} \sum_{\sigma' \in \{A,B\}} \sum_{\sigma'' \in \{A,B\}} \sum_{\vec{G}_1 \in \Lambda_1^*} \sum_{\vec{G}_2 \in \Lambda_2^*} \sum_{\vec{G}_1' \in \Lambda_1^*} \sum_{\vec{G}_2' \in \Lambda_2^*} \\ 
        &\left| \left[ \oldhat{h}( \vec{\kappa}_1 + \vec{q}_1 + \vec{G}_1;\ell) - \oldhat{h}( \vec{\kappa}_2 + \vec{G}_2 ; \ell ) \right] \oldhat{f}^{\sigma'}_2\left(\frac{\vec{\kappa}_1 - \vec{\kappa}_2 + \vec{G}_1 - \vec{G}_2 + \vec{q}_1 }{\gamma},\gamma \tau\right) \right|    \\
        &\left| \left[ \oldhat{h}( \vec{\kappa}_1 + \vec{q}_1 + \vec{G}_1';\ell) - \oldhat{h}( \vec{\kappa}_2 + \vec{G}_2' ; \ell ) \right] \oldhat{f}^{\sigma''}_2\left(\frac{\vec{\kappa}_1 - \vec{\kappa}_2 + \vec{G}_1' - \vec{G}_2' + \vec{q}_1 }{\gamma},\gamma \tau\right) \right| \; \text{d}\vec{q}_1.
    \end{split}
\end{equation}
Now using Assumption \ref{as:h_regularity}, we have
\begin{equation}
    \begin{split}
        &\left\| \left( \tilde{r}^\text{III}_1 \right)^\sigma \right\|_{L^2(\Gamma_1^*;\mathbb{C})}^2 \leq C \sum_{\sigma' \in \{A,B\}} \sum_{\sigma'' \in \{A,B\}} \sum_{\vec{G}_1 \in \Lambda_1^*} \sum_{\vec{G}_2 \in \Lambda_2^*} \sum_{\vec{G}_1' \in \Lambda_1^*} \sum_{\vec{G}_2' \in \Lambda_2^*} \int_{\Gamma^*_1}  \\ 
        &\left| \frac{\vec{\kappa}_1 - \vec{\kappa}_2 + \vec{G}_1 - \vec{G}_2 + \vec{q}_1}{\gamma} \right| e^{- D_2 \ell d( \vec{0} , [ \vec{\kappa}_1 + \vec{G}_1 + \vec{q}_1 , \vec{\kappa}_2 + \vec{G}_2 ] ) } \left| \oldhat{f}^{\sigma'}_2\left(\frac{\vec{\kappa}_1 - \vec{\kappa}_2 + \vec{G}_1 - \vec{G}_2 + \vec{q}_1 }{\gamma},\gamma \tau\right) \right|    \\
        &\left| \frac{\vec{\kappa}_1 - \vec{\kappa}_2 + \vec{G}_1' - \vec{G}_2' + \vec{q}_1}{\gamma} \right| e^{- D_2 \ell d( \vec{0} , [ \vec{\kappa}_1 + \vec{G}_1' + \vec{q}_1 , \vec{\kappa}_2 + \vec{G}_2' ] ) } \left| \oldhat{f}^{\sigma''}_2\left(\frac{\vec{\kappa}_1 - \vec{\kappa}_2 + \vec{G}_1' - \vec{G}_2' + \vec{q}_1 }{\gamma},\gamma \tau\right) \right| \; \text{d}\vec{q}_1.
    \end{split}
\end{equation}
Taking the sup over $\sigma', \sigma''$, writing $\vec{q}_1 + \vec{G}_1' = \vec{q}_1 + \vec{G}_1 + \tilde{\vec{G}}_1$, $\tilde{\vec{G}}_1 := \vec{G}_1' - \vec{G}_1$, and then combining the sum over $\vec{G}_1$ with the integral over $\vec{q}_1$ yields
\begin{equation}
    \begin{split}
        &\left\| \left( \tilde{r}^\text{III}_1 \right)^\sigma \right\|_{L^2(\Gamma_1^*;\mathbb{C})}^2 \leq C \sup_{\sigma' \in \{A,B\}} \sup_{\sigma'' \in \{A,B\}} \sum_{\vec{G}_2 \in \Lambda_2^*} \sum_{\tilde{\vec{G}}_1 \in \Lambda_1^*} \sum_{\vec{G}_2' \in \Lambda_2^*} \int_{\mathbb{R}^2}  \\ 
        &\left| \frac{\vec{\kappa}_1 - \vec{\kappa}_2 - \vec{G}_2 + \vec{q}_1}{\gamma} \right| e^{- D_2 \ell d( \vec{0} , [ \vec{\kappa}_1 + \vec{q}_1 , \vec{\kappa}_2 + \vec{G}_2 ] ) } \left| \oldhat{f}^{\sigma'}_2\left(\frac{\vec{\kappa}_1 - \vec{\kappa}_2 - \vec{G}_2 + \vec{q}_1 }{\gamma},\gamma \tau\right) \right|    \\
        &\left| \frac{\vec{\kappa}_1 - \vec{\kappa}_2 + \tilde{\vec{G}}_1 - \vec{G}_2' + \vec{q}_1}{\gamma} \right| e^{- D_2 \ell d( \vec{0} , [ \vec{\kappa}_1 + \tilde{\vec{G}}_1 + \vec{q}_1 , \vec{\kappa}_2 + \vec{G}_2' ] ) } \left| \oldhat{f}^{\sigma''}_2\left(\frac{\vec{\kappa}_1 - \vec{\kappa}_2 + \tilde{\vec{G}}_1 - \vec{G}_2' + \vec{q}_1 }{\gamma},\gamma \tau\right) \right| \; \text{d}\vec{q}_1.
    \end{split}
\end{equation}
Changing variables in the integral $\vec{q}_1 \rightarrow \frac{\vec{\kappa}_1 - \vec{\kappa}_2 - \vec{G}_2 + \vec{q}_1}{\gamma}$ and introducing $\tilde{\vec{G}}_2 := \vec{G}_2' - \vec{G}_2$ yields
\begin{equation} \label{eq:all_G1_G2_terms}
    \begin{split}
        &\left\| \left( \tilde{r}^\text{III}_1 \right)^\sigma \right\|_{L^2(\Gamma_1^*;\mathbb{C})}^2 \leq C \gamma^2 \sup_{\sigma' \in \{A,B\}} \sup_{\sigma'' \in \{A,B\}} \sum_{\vec{G}_2 \in \Lambda_2^*} \sum_{\tilde{\vec{G}}_1 \in \Lambda_1^*} \sum_{\tilde{\vec{G}}_2 \in \Lambda_2^*} \int_{\mathbb{R}^2}  \\ 
        &\left| \vec{q}_1 \right| e^{- D_2 \ell d( \vec{0} , [ \vec{\kappa}_2 + \vec{G}_2 + \gamma \vec{q}_1 , \vec{\kappa}_2 + \vec{G}_2 ] ) } \left| \oldhat{f}^{\sigma'}_2\left(\vec{q}_1 ,\gamma \tau\right) \right|    \\
        &\left| \vec{q}_1 +  \frac{\tilde{\vec{G}}_1 - \tilde{\vec{G}}_2}{\gamma} \right| e^{- D_2 \ell d\left( \vec{0} , \left[ \vec{\kappa}_2 + \vec{G}_2 + \tilde{\vec{G}}_2 + \gamma \left( \vec{q}_1 + \frac{\tilde{\vec{G}}_1 - \tilde{\vec{G}}_2}{\gamma} \right) , \vec{\kappa}_2 + \vec{G}_2 + \tilde{\vec{G}}_2 \right] \right) } \left| \oldhat{f}^{\sigma''}_2\left(\vec{q}_1 + \frac{\tilde{\vec{G}}_1 - \tilde{\vec{G}}_2}{\gamma} ,\gamma \tau\right) \right| \; \text{d}\vec{q}_1.
    \end{split}
\end{equation}
We now split the sum over $\tilde{\vec{G}}_1$ as $\sum_{\tilde{\vec{G}}_1 \in \Lambda_1^*, |\tilde{\vec{G}}_1 - \tilde{\vec{G}}_2| \leq L} + \sum_{\tilde{\vec{G}}_1, |\tilde{\vec{G}}_1 - \tilde{\vec{G}}_2| > L}$ for some fixed $L > 0$ fixed sufficiently small such that the first sum has at most one term. We start by bounding this term. The expression to be estimated is
\begin{equation} \label{eq:small_G1_G2_terms}
    \begin{split}
        &C \gamma^2 \sup_{\sigma' \in \{A,B\}} \sup_{\sigma'' \in \{A,B\}} \sum_{\vec{G}_2 \in \Lambda_2^*} \sum_{\tilde{\vec{G}}_2 \in \Lambda_2^*} \sum_{\substack{\tilde{\vec{G}}_1 \in \Lambda_1^* \\ |\tilde{\vec{G}}_1 - \tilde{\vec{G}}_2| \leq L}} \int_{\mathbb{R}^2} \left| \vec{q}_1 \right| e^{- D_2 \ell d( \vec{0} , [ \vec{\kappa}_2 + \vec{G}_2 + \gamma \vec{q}_1 , \vec{\kappa}_2 + \vec{G}_2 ] ) } \left| \oldhat{f}^{\sigma'}_2\left(\vec{q}_1 ,\gamma \tau\right) \right|    \\
        &\left| \vec{q}_1 +  \frac{\tilde{\vec{G}}_1 - \tilde{\vec{G}}_2}{\gamma} \right| e^{- D_2 \ell d\left( \vec{0} , \left[ \vec{\kappa}_2 + \vec{G}_2 + \tilde{\vec{G}}_2 + \gamma \left( \vec{q}_1 + \frac{\tilde{\vec{G}}_1 - \tilde{\vec{G}}_2}{\gamma} \right) , \vec{\kappa}_2 + \vec{G}_2 + \tilde{\vec{G}}_2 \right] \right) } \left| \oldhat{f}^{\sigma''}_2\left(\vec{q}_1 + \frac{\tilde{\vec{G}}_1 - \tilde{\vec{G}}_2}{\gamma} ,\gamma \tau\right) \right| \; \text{d}\vec{q}_1.
    \end{split}
\end{equation}
Using Cauchy-Schwarz, this expression can be estimated by
\begin{equation} \label{eq:product_of_L2}
    \begin{split}
        &C \gamma^2 \sup_{\sigma' \in \{A,B\}} \sup_{\sigma'' \in \{A,B\}} \sum_{\vec{G}_2 \in \Lambda_2^*} \sum_{\tilde{\vec{G}}_2 \in \Lambda_2^*} \sum_{\substack{\tilde{\vec{G}}_1 \in \Lambda_1^* \\ |\tilde{\vec{G}}_1 - \tilde{\vec{G}}_2| \leq L}}    \\ 
        &\left\| (\cdot)  e^{- D_2 \ell d( \vec{0} , [ \vec{\kappa}_2 + \vec{G}_2 + \gamma (\cdot) , \vec{\kappa}_2 + \vec{G}_2 ] ) } \oldhat{f}^{\sigma'}_2\left(\cdot ,\gamma \tau\right) \right\|_{L^2(\mathbb{R}^2)} \left\| (\cdot)  e^{- D_2 \ell d( \vec{0} , [ \vec{\kappa}_2 + \vec{G}_2 + \tilde{\vec{G}}_2 + \gamma (\cdot) , \vec{\kappa}_2 + \vec{G}_2 + \tilde{\vec{G}}_2 ] ) } \oldhat{f}^{\sigma''}_2\left(\cdot ,\gamma \tau\right) \right\|_{L^2(\mathbb{R}^2)}.
    \end{split}
\end{equation}
Note that the summand is independent of $\tilde{\vec{G}}_1$, so we can drop this sum, and that, by changing variables in the sum over $\tilde{\vec{G}}_2$, we can write the expression as a square
\begin{equation} \label{eq:product_of_L2_simplified}
    = C \gamma^2 \sup_{\sigma' \in \{A,B\}} \left( \sum_{\vec{G}_2 \in \Lambda_2^*} \left\| (\cdot)  e^{- D_2 \ell d( \vec{0} , [ \vec{\kappa}_2 + \vec{G}_2 + \gamma (\cdot) , \vec{\kappa}_2 + \vec{G}_2 ] ) } \oldhat{f}^{\sigma'}_2\left(\cdot ,\gamma \tau\right) \right\|_{L^2(\mathbb{R}^2)} \right)^2.
\end{equation}
Let us now consider the $L^2$ norm in \eqref{eq:product_of_L2_simplified}
\begin{equation}
    \begin{split}
        &\left\| (\cdot)  e^{- D_2 \ell d( \vec{0} , [ \vec{\kappa}_2 + \vec{G}_2 + \gamma (\cdot) , \vec{\kappa}_2 + \vec{G}_2 ] ) } \oldhat{f}^{\sigma'}_2\left(\cdot ,\gamma \tau\right) \right\|_{L^2}^2 \\
        &= \inty{\mathbb{R}^2}{}{ | \vec{q}_1 |^2 e^{- 2 D_2 \ell d( \vec{0} , [ \vec{\kappa}_2 + \vec{G}_2 + \gamma \vec{q}_1 , \vec{\kappa}_2 + \vec{G}_2 ] ) } \left| \oldhat{f}^{\sigma'}_2\left(\vec{q}_1 ,\gamma \tau\right) \right|^2 }{\vec{q}_1}.
    \end{split}
\end{equation}
We can split the integral into integrals over the regions $|\vec{q}_1| \leq \frac{1}{\gamma} \left| \frac{\vec{\kappa}_2 + \vec{G}_2}{M} \right|$ and $|\vec{q}_1| > \frac{1}{\gamma} \left| \frac{\vec{\kappa}_2 + \vec{G}_2}{M} \right|$ where $M > 1$ but is otherwise arbitrary and will be fixed later to optimize the estimate. In the first integration region we have
\begin{equation}
    \begin{split}
        &|\vec{q}_1| \leq \frac{1}{\gamma} \left| \frac{\vec{\kappa}_2 + \vec{G}_2}{M} \right|  \\
        &\implies d\left(\vec{0},\left[\vec{\kappa}_2 + \vec{G}_2 + \gamma \vec{q}_1,\vec{\kappa}_2 + \vec{G}_2\right]\right) \geq \left(1 - \frac{1}{M}\right) \left| \vec{\kappa}_2 + \vec{G}_2 \right|,
    \end{split}
\end{equation}
and so, using \eqref{eq:Lipschitz}, we have
\begin{equation} \label{eq:estt}
    \begin{split}
        &\inty{|\vec{q}_1| \leq \frac{1}{\gamma} \left| \frac{\vec{\kappa}_2 + \vec{G}_2}{M} \right|}{}{ e^{- 2 D_2 \ell d\left( \vec{0}, \left[ \vec{\kappa}_2 + \vec{G}_2 +\gamma \vec{q}_1 , \vec{\kappa}_2 + \vec{G}_2 \right] \right) } \left|\vec{q}_1 \oldhat{f}^{\sigma'}_2\left(\vec{q}_1,\gamma \tau\right) \right|^2 }{\vec{q}_1}   \\
        &\leq C e^{- 2 D_2 \ell \left( 1 - \frac{1}{M} \right) \left| \vec{\kappa}_2 + \vec{G}_2 \right|} \left\| f^{\sigma'}_2(\cdot,\gamma \tau) \right\|^2_{H^1(\field{R}^2)},
    \end{split}
\end{equation}
where $C > 0$ is a constant depending only on $\oldhat{h}$, and $D_2$ is as in Assumption \ref{as:h_regularity}. The second integral can be bounded using decay of $\oldhat{f}^{\sigma'}_2$ as
\begin{equation} \label{eq:third_sob_estimate}
    \begin{split}
        &\inty{|\vec{q}_1| > \frac{1}{\gamma}\left| \frac{\vec{\kappa}_2 + \vec{G}_2}{M} \right|}{}{ e^{- 2 D_2 \ell d \left( \vec{0}, \left[ \vec{\kappa}_2 + \vec{G}_2 + \gamma \vec{q}_1 , \vec{\kappa}_2 + \vec{G}_2 \right] \right)} \left|\vec{q}_1 \oldhat{f}^{\sigma'}_2\left(\vec{q}_1,\gamma \tau\right) \right|^2 }{\vec{q}_1}   \\
        &\leq \frac{C M^6 \gamma^6}{\left| \vec{\kappa}_2 + \vec{G}_2 \right|^6} \left\| f^{\sigma'}_2(\cdot,\gamma \tau) \right\|^2_{H^4(\field{R}^2)}.
    \end{split}
\end{equation}
We can thus bound \eqref{eq:product_of_L2_simplified} by
\begin{equation}
    = C \gamma^2 \sup_{\sigma' \in \{A,B\}} \left( \sum_{\vec{G}_2 \in \Lambda_2^*} \left( \frac{M^6 \gamma^6}{\left| \vec{\kappa}_2 + \vec{G}_2 \right|^6} \left\| f^{\sigma'}_2(\cdot,\gamma \tau) \right\|^2_{H^4(\field{R}^2)} + e^{- 2 D_2 \ell \left( 1 - \frac{1}{M} \right) \left| \vec{\kappa}_2 + \vec{G}_2 \right|} \left\| f^{\sigma'}_2(\cdot,\gamma \tau) \right\|^2_{H^1(\field{R}^2)} \right)^{1/2} \right)^2,
\end{equation}
and then, using equivalence of finite-dimensional norms, by 
\begin{equation}
    = C \gamma^2 \sup_{\sigma' \in \{A,B\}} \left( \sum_{\vec{G}_2 \in \Lambda_2^*} \frac{M^3 \gamma^3}{\left| \vec{\kappa}_2 + \vec{G}_2 \right|^3} \left\| f^{\sigma'}_2(\cdot,\gamma \tau) \right\|_{H^4(\field{R}^2)} + e^{- D_2 \ell \left( 1 - \frac{1}{M} \right) \left| \vec{\kappa}_2 + \vec{G}_2 \right|} \left\| f^{\sigma'}_2(\cdot,\gamma \tau) \right\|_{H^1(\field{R}^2)} \right)^2.
\end{equation}
Bounding the sums by
\begin{equation}
    \begin{split}
        &\sum_{\vec{G}_2 \in \Lambda_2^*} \frac{M^3 \gamma^3}{\left| \vec{\kappa}_2 + \vec{G}_2 \right|^3} \left\| f^{\sigma'}_2(\cdot,\gamma \tau) \right\|_{H^4(\field{R}^2)} \leq C M^3 \gamma^3 \left\| f^{\sigma'}_2(\cdot,\gamma \tau) \right\|_{H^4(\field{R}^2)}, \\ 
        &\sum_{\vec{G}_2 \in \Lambda_2^*} e^{- D_2 \ell \left( 1 - \frac{1}{M} \right) \left| \vec{\kappa}_2 + \vec{G}_2 \right|} \left\| f^{\sigma'}_2(\cdot,\gamma \tau) \right\|_{H^1(\field{R}^2)} \leq C e^{- D_2 \ell \left( 1 - \frac{1}{M} \right) \left| \vec{\kappa} \right|} \left\| f^{\sigma'}_2(\cdot,\gamma \tau) \right\|_{H^1(\field{R}^2)},
    \end{split}
\end{equation}
we can finally bound \eqref{eq:small_G1_G2_terms} by
\begin{equation} \label{eq:small_G1_G2_terms_bound}
    \leq C \gamma^2 \sup_{\sigma' \in \{A,B\}} \left( M^3 \gamma^3 \left\| f^{\sigma'}_2(\cdot,\gamma \tau) \right\|_{H^4(\field{R}^2)} + e^{- D_2 \ell \left( 1 - \frac{1}{M} \right) \left| \vec{\kappa} \right|} \left\| f^{\sigma'}_2(\cdot,\gamma \tau) \right\|_{H^1(\field{R}^2)} \right)^2.
\end{equation}
It remains to bound the other terms in \eqref{eq:all_G1_G2_terms}, i.e., to bound the sum
\begin{equation} \label{eq:large_G1_G2_terms}
    \begin{split}
        &C \gamma^2 \sup_{\sigma' \in \{A,B\}} \sup_{\sigma'' \in \{A,B\}} \sum_{\vec{G}_2 \in \Lambda_2^*} \sum_{\tilde{\vec{G}}_2 \in \Lambda_2^*} \sum_{\substack{\tilde{\vec{G}}_1 \in \Lambda_1^* \\ |\tilde{\vec{G}}_1 - \tilde{\vec{G}}_2| > L}} \int_{\mathbb{R}^2} \left| \vec{q}_1 \right| e^{- D_2 \ell d( \vec{0} , [ \vec{\kappa}_2 + \vec{G}_2 + \gamma \vec{q}_1 , \vec{\kappa}_2 + \vec{G}_2 ] ) } \left| \oldhat{f}^{\sigma'}_2\left(\vec{q}_1 ,\gamma \tau\right) \right|    \\
        &\left| \vec{q}_1 +  \frac{\tilde{\vec{G}}_1 - \tilde{\vec{G}}_2}{\gamma} \right| e^{- D_2 \ell d\left( \vec{0} , \left[ \vec{\kappa}_2 + \vec{G}_2 + \tilde{\vec{G}}_2 + \gamma \left( \vec{q}_1 + \frac{\tilde{\vec{G}}_1 - \tilde{\vec{G}}_2}{\gamma} \right) , \vec{\kappa}_2 + \vec{G}_2 + \tilde{\vec{G}}_2 \right] \right) } \left| \oldhat{f}^{\sigma''}_2\left(\vec{q}_1 + \frac{\tilde{\vec{G}}_1 - \tilde{\vec{G}}_2}{\gamma} ,\gamma \tau\right) \right| \; \text{d}\vec{q}_1.
    \end{split}
\end{equation}
The idea to bound these terms is just as in equation \eqref{eq:regions}. We split the plane into two half-planes, as
\begin{equation} \label{eq:regions_2}
    \begin{split}
        &\vec{q}_1 \in \mathbb{R}^2_+\left(\frac{\tilde{\vec{G}}_1 - \tilde{\vec{G}}_2}{\gamma}\right) \implies |\vec{q}_1| \geq \left| \frac{\tilde{\vec{G}}_1 - \tilde{\vec{G}}_2}{2 \gamma} \right|, \\
        &\vec{q}_1 \in \mathbb{R}^2_-\left(\frac{\tilde{\vec{G}}_1 - \tilde{\vec{G}}_2}{\gamma}\right) \implies \left|\vec{q}_1 - \frac{\tilde{\vec{G}}_1 - \tilde{\vec{G}}_2}{2 \gamma}\right| \geq \left| \frac{\tilde{\vec{G}}_1 - \tilde{\vec{G}}_2}{2 \gamma} \right|.
    \end{split}
\end{equation}
By this procedure we can bound \eqref{eq:large_G1_G2_terms} by
\begin{equation}
    \begin{split}
        &C \gamma^2 \sup_{\sigma' \in \{A,B\}} \sup_{\sigma'' \in \{A,B\}} \sum_{\vec{G}_2 \in \Lambda_2^*} \sum_{\tilde{\vec{G}}_2 \in \Lambda_2^*} \sum_{\substack{\tilde{\vec{G}}_1 \in \Lambda_1^* \\ |\tilde{\vec{G}}_1 - \tilde{\vec{G}}_2| > L}} \frac{\gamma^4}{|\tilde{\vec{G}}_1 - \tilde{\vec{G}}_2|^4} \int_{\mathbb{R}^2}  \\ 
        &\left( \left| \vec{q}_1 \right|^5 \left| \vec{q}_1 +  \frac{\tilde{\vec{G}}_1 - \tilde{\vec{G}}_2}{\gamma} \right| + \left| \vec{q}_1 \right| \left| \vec{q}_1 +  \frac{\tilde{\vec{G}}_1 - \tilde{\vec{G}}_2}{\gamma} \right|^5 \right) e^{- D_2 \ell d( \vec{0} , [ \vec{\kappa}_2 + \vec{G}_2 + \gamma \vec{q}_1 , \vec{\kappa}_2 + \vec{G}_2 ] ) } \left| \oldhat{f}^{\sigma'}_2\left(\vec{q}_1 ,\gamma \tau\right) \right|    \\
        &e^{- D_2 \ell d\left( \vec{0} , \left[ \vec{\kappa}_2 + \vec{G}_2 + \tilde{\vec{G}}_2 + \gamma \left( \vec{q}_1 + \frac{\tilde{\vec{G}}_1 - \tilde{\vec{G}}_2}{\gamma} \right) , \vec{\kappa}_2 + \vec{G}_2 + \tilde{\vec{G}}_2 \right] \right) } \left| \oldhat{f}^{\sigma''}_2\left(\vec{q}_1 + \frac{\tilde{\vec{G}}_1 - \tilde{\vec{G}}_2}{\gamma} ,\gamma \tau\right) \right| \; \text{d}\vec{q}_1.
    \end{split}
\end{equation}
We now repeat the argument of \eqref{eq:small_G1_G2_terms}-\eqref{eq:product_of_L2} to bound this expression by 
\begin{equation} \label{eq:big_G1_G2_terms_again}
    \begin{split}
        &C \gamma^2 \sup_{\sigma' \in \{A,B\}} \sup_{\sigma'' \in \{A,B\}} \sum_{\vec{G}_2 \in \Lambda_2^*} \sum_{\tilde{\vec{G}}_2 \in \Lambda_2^*} \sum_{\substack{\tilde{\vec{G}}_1 \in \Lambda_1^* \\ |\tilde{\vec{G}}_1 - \tilde{\vec{G}}_2| > L}} \frac{\gamma^4}{|\tilde{\vec{G}}_1 - \tilde{\vec{G}}_2|^4}  \\
        &\left\| (\cdot)^5 e^{- D_2 \ell d( \vec{0} , [\vec{\kappa}_2 + \vec{G}_2 + \gamma(\cdot),\vec{\kappa}_2 + \vec{G}_2] ) } \oldhat{f}^{\sigma'}_2(\cdot,\gamma \tau) \right\|_{L^2(\mathbb{R}^2)} \left\| (\cdot) e^{- D_2 \ell d( \vec{0} , [\vec{\kappa}_2 + \vec{G}_2 + \tilde{\vec{G}}_2 + \gamma(\cdot),\vec{\kappa}_2 + \vec{G}_2 + \tilde{\vec{G}}_2] ) } \oldhat{f}^{\sigma''}_2(\cdot,\gamma \tau) \right\|_{L^2(\mathbb{R}^2)} \\
        &+ \left\| (\cdot) e^{- D_2 \ell d( \vec{0} , [\vec{\kappa}_2 + \vec{G}_2 + \gamma(\cdot),\vec{\kappa}_2 + \vec{G}_2] ) } \oldhat{f}^{\sigma'}_2(\cdot,\gamma \tau) \right\|_{L^2(\mathbb{R}^2)} \left\| (\cdot)^5 e^{- D_2 \ell d( \vec{0} , [\vec{\kappa}_2 + \vec{G}_2 + \tilde{\vec{G}}_2 + \gamma(\cdot),\vec{\kappa}_2 + \vec{G}_2 + \tilde{\vec{G}}_2] ) } \oldhat{f}^{\sigma''}_2(\cdot,\gamma \tau) \right\|_{L^2(\mathbb{R}^2)}.
    \end{split}
\end{equation}
Again, the summand is independent of $\tilde{\vec{G}}_1$, so we can subsume that sum into the constant $C > 0$. We are then back to the setting of \eqref{eq:product_of_L2}. Following the same steps we can bound \eqref{eq:big_G1_G2_terms_again} by
\begin{equation} \label{eq:big_G1_G2_terms_bound}
    \begin{split}
        &\leq C \gamma^6 \sup_{\sigma' \in \{A,B\}} \sup_{\sigma'' \in \{A,B\}} \left( M^3 \gamma^3 \left\| f_2^{\sigma'}(\cdot,\gamma \tau) \right\|_{H^8(\mathbb{R}^2)} + e^{- D_2 \ell \left(1 - \frac{1}{M}\right)|\vec{\kappa}|} \left\| f_2^{\sigma'}(\cdot,\gamma \tau) \right\|_{H^5(\mathbb{R}^2)} \right) \\
        &\quad \quad \quad \quad \times \left( M^3 \gamma^3 \left\| f_2^{\sigma''}(\cdot,\gamma \tau) \right\|_{H^4(\mathbb{R}^2)} + e^{- D_2 \ell \left(1 - \frac{1}{M}\right)|\vec{\kappa}|} \left\| f_2^{\sigma''}(\cdot,\gamma \tau) \right\|_{H^1(\mathbb{R}^2)} \right).
    \end{split}
\end{equation}
Adding \eqref{eq:small_G1_G2_terms_bound} to \eqref{eq:big_G1_G2_terms_bound} we have the following estimate for \eqref{eq:all_G1_G2_terms}
\begin{equation}
    \begin{split}
        &\left\| \left( \tilde{r}^\text{III}_1 \right)^\sigma \right\|_{L^2(\Gamma_1^*;\mathbb{C})}^2 \leq  \\
        &C \gamma^2 \sup_{\sigma' \in \{A,B\}} \left( M^3 \gamma^3 \left\| f^{\sigma'}_2(\cdot,\gamma \tau) \right\|_{H^4(\field{R}^2)} + e^{- D_2 \ell \left( 1 - \frac{1}{M} \right) \left| \vec{\kappa} \right|} \left\| f^{\sigma'}_2(\cdot,\gamma \tau) \right\|_{H^1(\field{R}^2)} \right)^2  \\
        &+ C \gamma^6 \sup_{\sigma' \in \{A,B\}} \sup_{\sigma'' \in \{A,B\}} \left( M^3 \gamma^3 \left\| f_2^{\sigma'}(\cdot,\gamma \tau) \right\|_{H^8(\mathbb{R}^2)} + e^{- D_2 \ell \left(1 - \frac{1}{M}\right)|\vec{\kappa}|} \left\| f_2^{\sigma'}(\cdot,\gamma \tau) \right\|_{H^5(\mathbb{R}^2)} \right)    \\
        &\quad \quad \quad \quad \times \left( M^3 \gamma^3 \left\| f_2^{\sigma''}(\cdot,\gamma \tau) \right\|_{H^4(\mathbb{R}^2)} + e^{- D_2 \ell \left(1 - \frac{1}{M}\right)|\vec{\kappa}|} \left\| f_2^{\sigma''}(\cdot,\gamma \tau) \right\|_{H^1(\mathbb{R}^2)} \right),
    \end{split}
\end{equation}
It will be convenient in what follows to simplify this expression, and then use Assumption \ref{as:h_regularity} to factor out $\oldhat{h}( |\vec{\kappa}| ; \ell )^2$, to obtain
\begin{equation}
    \begin{split}
        &\left\| \left( \tilde{r}^\text{III}_1 \right)^\sigma \right\|_{L^2(\Gamma_1^*;\mathbb{C})}^2 \leq  \\
        &C \oldhat{h}(|\vec{\kappa}|;\ell)^2 \sup_{\sigma' \in \{A,B\}} \left( M^3 \gamma^4 e^{D_1 \ell |\vec{\kappa}|} \left\| f^{\sigma'}_2(\cdot,\gamma \tau) \right\|_{H^4(\field{R}^2)} + \gamma e^{\left(D_1 - D_2 \left( 1 - \frac{1}{M} \right) \right) \ell \left| \vec{\kappa} \right|} \left\| f^{\sigma'}_2(\cdot,\gamma \tau) \right\|_{H^1(\field{R}^2)} \right)^2  \\
        &+ C \oldhat{h}(|\vec{\kappa}|;\ell)^2 \sup_{\sigma' \in \{A,B\}} \left( M^3 \gamma^6 e^{D_1 \ell |\vec{\kappa}|} \left\| f_2^{\sigma'}(\cdot,\gamma \tau) \right\|_{H^8(\mathbb{R}^2)} + \gamma^3 e^{\left( D_1 - D_2 \left(1 - \frac{1}{M}\right) \right) \ell |\vec{\kappa}|} \left\| f_2^{\sigma'}(\cdot,\gamma \tau) \right\|_{H^5(\mathbb{R}^2)} \right)^2,
    \end{split}
\end{equation}
from which \eqref{eq:r2_bound} follows.
}

{\color{black}
We now prove \eqref{eq:r3_bound}. After applying the triangle inequality, the quantity to bound is
\begin{equation} 
    \begin{split}
        &\left\| \left( \tilde{r}^{\text{IV}}_1 \right)^\sigma \right\|^2_{L^2(\Gamma^*_1)} \leq C \gamma^{-2} \sup_{\sigma' \in \{A,B\}} \sup_{\sigma'' \in \{A,B\}} \int_{\Gamma_1^*} \sum_{\vec{G}_1 \in \Lambda_1^*} \sum_{\substack{\vec{G}_2 \in \Lambda_2^* \\ \vec{G}_2 \notin \{ \vec{0}, \vec{b}_{2,2}, - \vec{b}_{2,1} \} }} \sum_{\vec{G}_1' \in \Lambda_1^*} \sum_{\substack{\vec{G}_2' \in \Lambda_2^* \\ \vec{G}_2' \notin \{ \vec{0}, \vec{b}_{2,2}, - \vec{b}_{2,1} \} }} \\
        &\left| \oldhat{h}( \vec{\kappa}_2 + \vec{G}_2 ; \ell ) \right| \left| \oldhat{f}^{\sigma'}_2\left(\frac{\vec{\kappa}_1 - \vec{\kappa}_2 + \vec{G}_1 - \vec{G}_2 + \vec{q}_1 }{\gamma},\gamma \tau\right) \right|   \\
        &\left| \oldhat{h}( \vec{\kappa}_2 + \vec{G}_2' ; \ell ) \right| \left| \oldhat{f}^{\sigma''}_2\left(\frac{\vec{\kappa}_1 - \vec{\kappa}_2 + \vec{G}_1' - \vec{G}_2' + \vec{q}_1 }{\gamma},\gamma \tau\right) \right| \; \text{d}\vec{q}_1.
    \end{split}
\end{equation}
We follow similar steps as in previous sections. We replace the integral over $\Gamma_1^*$ and sum over $\vec{G}_1$ by an integral over $\mathbb{R}^2$, introduce $\tilde{\vec{G}}_1 := \vec{G}_1' - \vec{G}_1$, and then change variables $\vec{q}_1 \rightarrow \frac{\vec{\kappa}_1 - \vec{\kappa}_2 - \vec{G}_2 + \vec{q}_1}{\gamma}$ to obtain 
\begin{equation} 
    \begin{split}
        &\left\| \left( \tilde{r}^{\text{IV}}_1 \right)^\sigma \right\|^2_{L^2(\Gamma^*_1)} \leq C \sup_{\sigma' \in \{A,B\}} \sup_{\sigma'' \in \{A,B\}} \int_{\mathbb{R}^2} \sum_{\substack{\vec{G}_2 \in \Lambda_2^* \\ \vec{G}_2 \notin \{ \vec{0}, \vec{b}_{2,2}, - \vec{b}_{2,1} \} }} \sum_{\tilde{\vec{G}}_1 \in \Lambda_1^*} \sum_{\substack{\vec{G}_2' \in \Lambda_2^* \\ \vec{G}_2' \notin \{ \vec{0}, \vec{b}_{2,2}, - \vec{b}_{2,1} \} }} \\
        &\left| \oldhat{h}( \vec{\kappa}_2 + \vec{G}_2 ; \ell ) \right| \left| \oldhat{f}^{\sigma'}_2( \vec{q}_1, \gamma \tau ) \right| \left| \oldhat{h}( \vec{\kappa}_2 + \vec{G}_2' ; \ell ) \right| \left| \oldhat{f}^{\sigma''}_2\left(\vec{q}_1 + \frac{\tilde{\vec{G}}_1 + \vec{G}_2 - \vec{G}_2'}{\gamma},\gamma \tau\right) \right| \; \text{d}\vec{q}_1.
    \end{split}
\end{equation}
Just as in \eqref{eq:all_G1_G2_terms}, we split the sum over $\tilde{\vec{G}}_1$ as $\sum_{\tilde{\vec{G}}_1 \in \Lambda^*_1, |\tilde{\vec{G}}_1 + \vec{G}_2 - \vec{G}_2'| \leq L} + \sum_{\tilde{\vec{G}}_1 \in \Lambda^*_1, |\tilde{\vec{G}}_1 + \vec{G}_2 - \vec{G}_2'| > L}$, fixing $L > 0$ sufficiently small that the first sum only involves one term for each $\vec{G}_2, \vec{G}_2'$. The first sum is then straightforward to estimate using Cauchy-Schwarz as
\begin{equation} \label{eq:small_G_terms}
    \leq C \sup_{\sigma' \in \{A,B\}} \sum_{\substack{\vec{G}_2 \in \Lambda_2^* \\ \vec{G}_2 \notin \{ \vec{0}, \vec{b}_{2,2}, - \vec{b}_{2,1} \} }} \sum_{\substack{\vec{G}_2' \in \Lambda_2^* \\ \vec{G}_2' \notin \{ \vec{0}, \vec{b}_{2,2}, - \vec{b}_{2,1} \} }} \left| \oldhat{h}( \vec{\kappa}_2 + \vec{G}_2 ; \ell ) \right| \left| \oldhat{h}( \vec{\kappa}_2 + \vec{G}_2' ; \ell ) \right| \left\| \oldhat{f}^{\sigma'}_2(\cdot,\gamma \tau) \right\|^2_{L^2(\mathbb{R}^2)}.
\end{equation}
It is straightforward to check that
\begin{equation}
    \min_{\substack{\vec{G}_2 \in \Lambda_2^* \\ \vec{G}_2 \notin \{ \vec{0}, \vec{b}_{2,2}, - \vec{b}_{2,1} \} }} | \vec{\kappa}_2 + \vec{G}_2 | = 2 |\vec{\kappa}|,
\end{equation}
from which it follows immediately that the terms \eqref{eq:small_G_terms} are bounded by
\begin{equation} \label{eq:terms_1}
    \leq C e^{- 4 D_2 \ell |\vec{\kappa}|} \sup_{\sigma' \in \{A,B\}} \| f_2^{\sigma'}(\cdot,\gamma \tau) \|^2_{L^2(\mathbb{R}^2)}.
\end{equation}
It remains to bound the terms
\begin{equation} \label{eq:large_G_terms}
    \begin{split}
        &C \sup_{\sigma' \in \{A,B\}} \sup_{\sigma'' \in \{A,B\}} \int_{\mathbb{R}^2} \sum_{\substack{\vec{G}_2 \in \Lambda_2^* \\ \vec{G}_2 \notin \{ \vec{0}, \vec{b}_{2,2}, - \vec{b}_{2,1} \} }} \sum_{\substack{\vec{G}_2' \in \Lambda_2^* \\ \vec{G}_2' \notin \{ \vec{0}, \vec{b}_{2,2}, - \vec{b}_{2,1} \} }} \sum_{\substack{\tilde{\vec{G}}_1 \in \Lambda_1^* \\ |\tilde{\vec{G}}_1 + \vec{G}_2 - \vec{G}_2'| > L}} \\
        &\left| \oldhat{h}( \vec{\kappa}_2 + \vec{G}_2 ; \ell ) \right| \left| \oldhat{f}^{\sigma'}_2( \vec{q}_1, \gamma \tau ) \right| \left| \oldhat{h}( \vec{\kappa}_2 + \vec{G}_2' ; \ell ) \right| \left| \oldhat{f}^{\sigma''}_2\left(\vec{q}_1 + \frac{\tilde{\vec{G}}_1 + \vec{G}_2 - \vec{G}_2'}{\gamma},\gamma \tau\right) \right| \; \text{d}\vec{q}_1.
    \end{split}
\end{equation}
Just as in \eqref{eq:regions_2}, we can bound \eqref{eq:large_G_terms} by
\begin{equation} \label{eq:these}
    \begin{split}
        &C \sup_{\sigma' \in \{A,B\}} \sup_{\sigma'' \in \{A,B\}} \int_{\mathbb{R}^2} \sum_{\substack{\vec{G}_2 \in \Lambda_2^* \\ \vec{G}_2 \notin \{ \vec{0}, \vec{b}_{2,2}, - \vec{b}_{2,1} \} }} \sum_{\substack{\vec{G}_2' \in \Lambda_2^* \\ \vec{G}_2' \notin \{ \vec{0}, \vec{b}_{2,2}, - \vec{b}_{2,1} \} }} \sum_{\substack{\tilde{\vec{G}}_1 \in \Lambda_1^* \\ |\tilde{\vec{G}}_1 + \vec{G}_2 - \vec{G}_2'| > L}} \frac{\gamma^4}{|\tilde{\vec{G}}_1 + \vec{G}_2 - \vec{G}_2'|^4}  \\
        &\left( |\vec{q}_1|^4 + \left| \vec{q}_1 + \frac{\tilde{\vec{G}}_1 + \vec{G}_2 - \vec{G}_2'}{\gamma} \right|^4 \right) \\
        &\left| \oldhat{h}( \vec{\kappa}_2 + \vec{G}_2 ; \ell ) \right| \left| \oldhat{f}^{\sigma'}_2( \vec{q}_1, \gamma \tau ) \right| \left| \oldhat{h}( \vec{\kappa}_2 + \vec{G}_2' ; \ell ) \right| \left| \oldhat{f}^{\sigma''}_2\left(\vec{q}_1 + \frac{\tilde{\vec{G}}_1 + \vec{G}_2 - \vec{G}_2'}{\gamma},\gamma \tau\right) \right| \; \text{d}\vec{q}_1.
    \end{split}
\end{equation}
We can bound each term in the sum \eqref{eq:these} using Cauchy-Schwarz, and then sum over $\tilde{\vec{G}}_1$ to obtain
\begin{equation} \label{eq:terms_2}
    \leq C \gamma^4 e^{- 4 D_2 \ell |\vec{\kappa}|} \sup_{\sigma' \in \{A,B\}} \sup_{\sigma'' \in \{A,B\}} \| f^{\sigma'}_2(\cdot,\gamma \tau) \|_{L^2(\mathbb{R}^2)} \| f^{\sigma''}_2(\cdot,\gamma \tau) \|_{H^4(\mathbb{R}^2)}.
\end{equation}
Adding \eqref{eq:terms_1} to \eqref{eq:terms_2} we conclude that
\begin{equation}
    \begin{split}
        &\left\| \left( \tilde{r}^{\text{IV}}_1 \right)^\sigma \right\|^2_{L^2(\Gamma^*_1)} \leq C e^{- 4 D_2 \ell |\vec{\kappa}|} \sup_{\sigma' \in \{A,B\}} \| f_2^{\sigma'}(\cdot,\gamma \tau) \|^2_{L^2(\mathbb{R}^2)} \\
        &\quad \quad \quad \quad \quad \quad \quad \quad \quad + C \gamma^4 e^{- 4 D_2 \ell |\vec{\kappa}|} \sup_{\sigma' \in \{A,B\}} \sup_{\sigma'' \in \{A,B\}} \| f^{\sigma'}_2(\cdot,\gamma \tau) \|_{L^2(\mathbb{R}^2)} \| f^{\sigma''}_2(\cdot,\gamma \tau) \|_{H^4(\mathbb{R}^2)}.
    \end{split}
\end{equation}
We can simplify and again use Assumption \ref{as:h_regularity} to factor out $\oldhat{h}(|\vec{\kappa}|;\ell)^2$ to obtain
\begin{equation}
    \begin{split}
        &\left\| \left( \tilde{r}^{\text{IV}}_1 \right)^\sigma \right\|^2_{L^2(\Gamma^*_1)} \leq C \oldhat{h}(|\vec{\kappa}|;\ell)^2 e^{( 2 D_1 - 4 D_2 ) \ell |\vec{\kappa}|} \sup_{\sigma' \in \{A,B\}} \| f_2^{\sigma'}(\cdot,\gamma \tau) \|^2_{L^2(\mathbb{R}^2)} \\
        &\quad \quad \quad \quad \quad \quad \quad \quad \quad + C \gamma^4 \oldhat{h}(|\vec{\kappa}|;\ell)^2 e^{(2 D_1 - 4 D_2) \ell |\vec{\kappa}|} \sup_{\sigma' \in \{A,B\}} \| f^{\sigma'}_2(\cdot,\gamma \tau) \|^2_{H^4(\mathbb{R}^2)},
    \end{split}
\end{equation}
from which \eqref{eq:r3_bound} follows.
}

\subsection{Transforming leading terms back to real space} \label{sec:transform_back_2}

It remains only to transform the first three terms of the right-hand side of \eqref{eq:off_diagonal_terms} back into real space. These terms are
\begin{equation}
    \begin{split}
        &\frac{\gamma^{-1} |\Gamma^*|^{1/2}}{(2 \pi)^2} \oldhat{h}(|\vec{\kappa}| ; \ell) \sum_{\sigma' \in \{A,B\}} \sum_{\vec{G}_1 \in \Lambda_1^*} e^{- i \vec{G}_1 \cdot \vec{\tau}_1^\sigma} T^{\sigma \sigma'}_1 \oldhat{f}^{\sigma'}_2\left(\frac{\vec{q}_1 + \mathfrak{s}_1 - \vec{G}_1 }{\gamma},\gamma \tau\right)    \\
        &+ e^{- i \vec{G}_1 \cdot \vec{\tau}_1^\sigma} T^{\sigma \sigma'}_2 \oldhat{f}^{\sigma'}_2\left(\frac{\vec{q}_1 + \mathfrak{s}_2 - \vec{G}_1}{\gamma},\gamma \tau\right) + e^{- i \vec{G}_1 \cdot \vec{\tau}_1^{\sigma}} T^{\sigma \sigma'}_3 \oldhat{f}^{\sigma'}_2\left(\frac{\vec{q}_1 + \mathfrak{s}_3 - \vec{G}_1 }{\gamma},\gamma \tau\right).
    \end{split}
\end{equation}
Taking the inverse Bloch transform and inserting the Fourier transform formula we have
\begin{equation}
    \begin{split}
        &\frac{\gamma^{-1}}{(2 \pi)^2} \oldhat{h}(|\vec{\kappa}|;\ell) \int_{\Gamma_1^* - \vec{\kappa}_1} e^{i (\vec{q}_1 + \vec{\kappa}_1) \cdot ( \vec{R}_1 + \vec{\tau}_1^\sigma)} \sum_{\sigma' \in \{A,B\}} \sum_{\vec{G}_1 \in \Lambda_1^*} \\
        &\sum_{j = 1}^3 e^{- i \vec{G}_1 \cdot \vec{\tau}_1^\sigma} T^{\sigma \sigma'}_j \int_{\field{R}^2} e^{- i \left( \frac{\vec{q}_1 + \mathfrak{s}_j - \vec{G}_1}{\gamma} \right) \cdot \vec{X}} f^{\sigma'}_2\left(\vec{X} ,\gamma \tau\right) \; \text{d}\vec{X} \; \text{d}\vec{q}_1.
    \end{split}
\end{equation}
Since $e^{i \vec{G}_1 \cdot \vec{R}_1} = 1$ for all $\vec{R}_1 \in \Lambda_1$ and $\vec{G}_1 \in \Lambda_1^*$, this can be re-written as
\begin{equation}
    \begin{split}
        &\frac{\gamma^{-1}}{(2 \pi)^2} \oldhat{h}(|\vec{\kappa}|;\ell) \sum_{\sigma' \in \{A,B\}} \sum_{\vec{G}_1 \in \Lambda_1^*} \int_{\Gamma_1^* - \vec{\kappa}_1} e^{i (\vec{q}_1 - \vec{G}_1) \cdot ( \vec{R}_1 + \vec{\tau}_1^\sigma)} \\
        &\sum_{j = 1}^3 T^{\sigma \sigma'}_j \int_{\field{R}^2} e^{- i \left( \frac{\vec{q}_1 - \vec{G}_1}{\gamma} \right) \cdot \vec{X}} e^{- i \frac{\mathfrak{s}_j \cdot \vec{X}}{\gamma}} f^{\sigma'}_2\left(\vec{X} ,\gamma \tau\right) \; \text{d}\vec{X} \; \text{d}\vec{q}_1 e^{i \vec{\kappa}_1 \cdot (\vec{R}_1 + \vec{\tau}_1^\sigma)}.
    \end{split}
\end{equation}
Replacing the sum over $\vec{G}_1$ and integral over $\vec{q}_1$ by a single integral over $\field{R}^2$, and then evaluating this integral, we find 
\begin{equation}
    \gamma^{-1} \oldhat{h}(|\vec{\kappa}|;\ell) \sum_{\sigma' \in \{A,B\}} \int_{\field{R}^2} \delta \left( \vec{R}_1 + \vec{\tau}_1^\sigma - \frac{\vec{X}}{\gamma} \right) \sum_{j = 1}^3 T^{\sigma \sigma'}_j e^{- i \frac{\mathfrak{s}_j \cdot \vec{X}}{\gamma}} f^{\sigma'}_2\left(\vec{X} ,\gamma \tau\right) \; \text{d}\vec{X} \; e^{i \vec{\kappa}_1 \cdot (\vec{R}_1 + \vec{\tau}_1^\sigma)}.
\end{equation}
Carrying out the integral over $\vec{X}$ after changing variables to $\frac{\vec{X}}{\gamma}$ yields \eqref{eq:off_diag_decomp}.

\section{Proof of Lemma \ref{lem:BM_sols}} \label{sec:sobolev_bounds}

First, note that the operator $H_{\text{BM}}$ with the domain $H^1(\mathbb{R}^2;\mathbb{C}^4)$ is an unbounded self-adjoint operator in $L^2(\mathbb{R}^2;\mathbb{C}^4)$. To see this, note that the diagonal part of $H_{\text{BM}}$ is straightforwardly self-adjoint using the Fourier transform (see, for example, \cite{1975ReedSimon}), while the off-diagonal terms are a bounded symmetric perturbation and hence do not affect this property (see, for example, \cite{CyconFroeseKirschSimon}). So \eqref{eq:unitary_propagator} now follows since $- i H_{\text{BM}}$ generates a unitary propagator in $L^2(\mathbb{R}^2;\mathbb{C}^4)$ (see, for example, \cite{Pazy}).

To see \eqref{eq:preservation_of_regularity}, we first note that, using Lemma \ref{lem:scaled_moire}, we can find $\mu > 0$ depending on $\lambda_0$ and $\lambda_1$, but independent of $\gamma$, such that
    \begin{equation}
        \| H_{\text{BM}}^2 + \mu \|_{L^2(\mathbb{R}^2;\mathbb{C}^4) \rightarrow L^2(\mathbb{R}^2;\mathbb{C}^4)} \geq \frac{\mu}{2},
    \end{equation}
so that $H_{\text{BM}}^2 + \mu$ is invertible. It is also, clearly, a second-order elliptic operator. Hence, we now have, by elliptic theory \cite{Taylor1}, the existence of constants $C_{s,1}, C_{s,2} > 0$ for each positive integer $s$ such that
\begin{equation}
    C_{s,1} \| f \|_{H^s(\mathbb{R}^2;\mathbb{C}^4)} \leq \| ( H^2_{\text{BM}} + \mu )^{s/2} f \|_{L^2(\mathbb{R}^2;\mathbb{C}^4)} \leq C_{s,2} \| f \|_{H^s(\mathbb{R}^2;\mathbb{C}^4)}.
\end{equation}
Using Lemma \ref{lem:scaled_moire} again, these constants may depend on $\lambda_0$ and $\lambda_1$, but can be chosen independently of $\gamma$. We now use the fact that the unitary propagator generated by $- i H_{\text{BM}}$ commutes with functions of $H_{\text{BM}}$ to see that, for any positive integer $s$,
    \begin{equation}
    \begin{split}
        &\| f(\cdot,T) \|_{H^s(\mathbb{R}^2;\mathbb{C}^4)} \leq \\
        &C_{1,s}^{-1} \left\| ( H_{\text{BM}}^2 + \mu )^{s/2} f(\cdot,T) \right\|_{L^2(\mathbb{R}^2;\mathbb{C}^4)} = C_{1,s}^{-1} \left\| ( H_{\text{BM}}^2 + \mu )^{s/2} f_0(\cdot) \right\|_{L^2(\mathbb{R}^2;\mathbb{C}^4)} \\
        &\quad \quad \quad \quad \leq C_{1,s}^{-1} C_{2,s} \| f_0 \|_{H^s(\mathbb{R}^2;\mathbb{C}^4)},
    \end{split}
    \end{equation}
which proves \eqref{eq:preservation_of_regularity}. We emphasize again the importance of Lemma \ref{lem:scaled_moire}, which follows from the scaling of $\theta$ with respect to $\gamma$ \eqref{eq:balance}. Without this scaling, \eqref{eq:preservation_of_regularity} would not hold with constants independent of $\gamma$ because $H_{\text{BM}}^2$ involves derivatives of the moir\'e potential which grow without bound for fixed $\theta$ as $\gamma \downarrow 0$ (see \eqref{eq:bad_scaled_moire}).

\end{document}